\pdfoutput=1

\documentclass[manuscript,screen]{acmart}

\AtBeginDocument{%
  \providecommand\BibTeX{{%
    \normalfont B\kern-0.5em{\scshape i\kern-0.25em b}\kern-0.8em\TeX}}}

\setcopyright{acmcopyright}
\copyrightyear{2022}
\acmYear{2022}
\acmDOI{XXXXXXX.XXXXXXX}

\acmJournal{TKDD}
\acmVolume{00}
\acmNumber{0}
\acmArticle{000}
\acmMonth{0}

\usepackage{mathtools}
\usepackage{tikz}
\usepackage{pgfplots}
\pgfplotsset{compat=newest}
\usepgfplotslibrary{groupplots}
\usetikzlibrary{arrows.meta}
\usepackage{subcaption}
\usepackage{bm}
\usepackage{mathmacros}
\usepackage{wrapfig}
\usepackage{cancel}
\usepackage[ruled, linesnumbered, noline]{algorithm2e}
\SetStartEndCondition{ }{}{}%
\SetKwProg{Fn}{def}{\string:}{}
\SetKwFunction{Range}{range}
\SetKw{KwTo}{in}\SetKwFor{For}{for}{\string:}{}%
\SetKwIF{If}{ElseIf}{Else}{if}{:}{elif}{else:}{}%
\SetKwFor{While}{while}{:}{fintq}%

\AlgoDontDisplayBlockMarkers \SetAlgoNoEnd \SetAlgoNoLine%
\SetKwComment{Comment}{/* }{ */}

\definecolor{niceOrange}{rgb}{1, 0.6, 0}
\definecolor{niceOrange2}{rgb}{1, 0.49, 0}
\definecolor{kthGreen}{RGB}{98, 146, 46}
\definecolor{kthYellow}{RGB}{250,185, 25}
\definecolor{burgundy}{rgb}{0.5, 0.0, 0.13}

\begin{document}

\title{\textsc{Digraphwave}: Scalable Extraction of Structural Node Embeddings via Diffusion on Directed Graphs}

\author{Ciwan Ceylan}

\affiliation{%
  \institution{KTH Royal Institute of Technology}
  \city{Stockholm}
  \country{Sweden}
}
\affiliation{%
  \institution{SEB Group}
  \city{Stockholm}
  \country{Sweden}
}
\email{ciwan@kth.se}
\orcid{0000-0002-8044-4773}

\author{Kambiz Ghoorchian}
\affiliation{%
  \institution{SEB Group}
  \city{Stockholm}
  \country{Sweden}
}
\email{kambiz.ghoorchian@seb.se}
\orcid{0000-0003-1007-8533}

\author{Danica Kragic}
\affiliation{%
  \institution{KTH Royal Institute of Technology}
  \city{Stockholm}
  \country{Sweden}
}
\email{dani@kth.se}
\orcid{0000-0003-2965-2953}


\begin{abstract}

Structural node embeddings, vectors capturing local connectivity information for each node in a graph, have many applications in data mining and machine learning, e.g., network alignment and node classification, clustering and anomaly detection.
For the analysis of directed graphs, e.g., transactions graphs, communication networks and social networks, the capability to capture directional information in the structural node embeddings is highly desirable, as is scalability of the embedding extraction method.
Most existing methods are nevertheless only designed for undirected graph.
Therefore, we present \textsc{Digraphwave} -- a scalable algorithm for extracting structural node embeddings on directed graphs.
The \textsc{Digraphwave} embeddings consist of compressed diffusion pattern signatures, which are twice enhanced to increase their discriminate capacity.
By proving a lower bound on the heat contained in the local vicinity of a diffusion initialization node, theoretically justified diffusion timescale values are established, and \textsc{Digraphwave} is left with only two easy-to-interpret hyperparameters: the embedding dimension and a neighbourhood resolution specifier.
In our experiments, the two embedding enhancements, named transposition and aggregation, are shown to lead to a significant increase in macro F1 score for classifying automorphic identities, with \textsc{Digraphwave} outperforming all other structural embedding baselines.
Moreover, \textsc{Digraphwave} either outperforms or matches the performance of all baselines on real graph datasets, displaying a particularly large  performance gain in a network alignment task, while also being scalable to graphs with millions of nodes and edges, running up to $30$x faster than a previous diffusion pattern based method and with a fraction of the memory consumption. 

\end{abstract}




\maketitle

\section{Introduction}

Node embeddings, i.e., vector representations of each node in a graph, are ubiquitous for solving generic machine learning and data mining tasks on graph structured data,
e.g., node classification \cite{henderson_its_2011, LINE_15, grover_node2vec_2016, donnat_learning_2018, netsmf_19, qiu_gcc_2020}, clustering \cite{henderson_rolx_2012, score_emb_clustering, donnat_learning_2018, node2vec_community_detection_2020}, regression \cite{regression_14,ceylan21a} and anomaly detection \cite{akoglu_oddball:_2010, akoglu_graph-based_2014, anom_det_21}.
The information captured in the embeddings is determined by the extraction method used to obtain them, which exist in a large variety.
One may narrow down available methods by distinguishing between those aimed at graphs with additional node attributes, which are then often incorporated into the embeddings, versus methods aimed at plain graphs, i.e., graphs without additional node attributes.
In this work, we focus on extracting node embeddings for plain graphs, and we refer the reader to other works on extracting node embeddings on attributed graphs, e.g., \cite{kipf2016semi, hamilton2017inductive, velickovic2018graph, ahmed2018learning, HONE_18, digraph_20, zhang2021magnet}.

For plain graphs, node embeddings can generally be divided into two complementary types: positional embeddings and structural embeddings  \cite{rossi_embeddings_20, zhu_2021_proximity_is_all_you_need}.
Methods for extracting positional embeddings follow the axiom that nodes in proximity, or which are well-connected, should have similar embeddings. 
Extraction techniques involve skip-gram models using random walks \cite{perozzi_deepwalk_2014, grover_node2vec_2016}, matrix factorization of similarity matrices capturing multi-hop similarities \cite{qiu_network_2018}, explicit modelling of first and second order proximity \cite{LINE_15} and mixed approaches \cite{netsmf_19}.
These embeddings are particularly useful for machine learning tasks aimed at inferring node properties which are homophilous, i.e., when connected nodes tend to share similar properties, and for link prediction \cite{grover_node2vec_2016, zhang2018link}.

However, for inference of node properties which are not homophilous, or for applications like network alignment \cite{regal_xnetmf_18} and node similarity search \cite{qiu_gcc_2020} where nodes in disconnected graphs are to be compared, structural embeddings are more suitable.
Structural embeddings aim to capture the local connectivity patterns around each node, e.g., a centre of a star pattern, leaf node pattern or bridge node pattern.
Such patterns are generic to graphs, which is why structural embeddings can be compared between disconnected components of the same graph, or even different graph datasets, something which would be meaningless using positional embeddings without some additional processing.

The node degrees are the most basic structural features, and some structural embedding extraction methods \cite{henderson_its_2011, regal_xnetmf_18, ember_19} use these as a building blocks to construct more advanced handcrafted features, e.g., via recursive feature concatenation or by building degree histograms over k-hop neighbourhoods. 
Additionally, dimensionality reduction via matrix factorization is commonly applied to obtain the final embeddings \cite{henderson_rolx_2012, regal_xnetmf_18, ember_19}.
Other approaches combine handcrafted similarities with self-supervised learning \cite{ribeiro_struc2vec_2017}, avoid handcrafted features by pretraining graph neural networks using contrastive learning \cite{qiu_gcc_2020}, or construct embeddings using the heat kernel \cite{donnat_learning_2018}. 

Unfortunately, many existing structural embedding extraction algorithms are only defined for undirected graphs \cite{qiu_gcc_2020, ribeiro_struc2vec_2017, donnat_learning_2018, drne18, nikolentzos_19, riwalk19}.
While it is always possible to treat a directed graph as undirected to extract the embeddings, doing so is undesirable if the direction information of the edges is important, as it commonly is for distinguishing follower and followee in social networks, citer and citee in citation networks and sender and receiver in a communication or transaction network.
Furthermore, due to the large size of these types of graph datasets, scalability, both in terms of space and time complexity, is always important for embedding extraction methods.
This motivates our work in which we present \textsc{Digraphwave}, a structural node embedding extraction algorithm for directed graphs, possibly weighted, which scales to graphs with millions of nodes and edges.

\begin{figure}[ht]
    \centering
    \includegraphics[width=\textwidth]{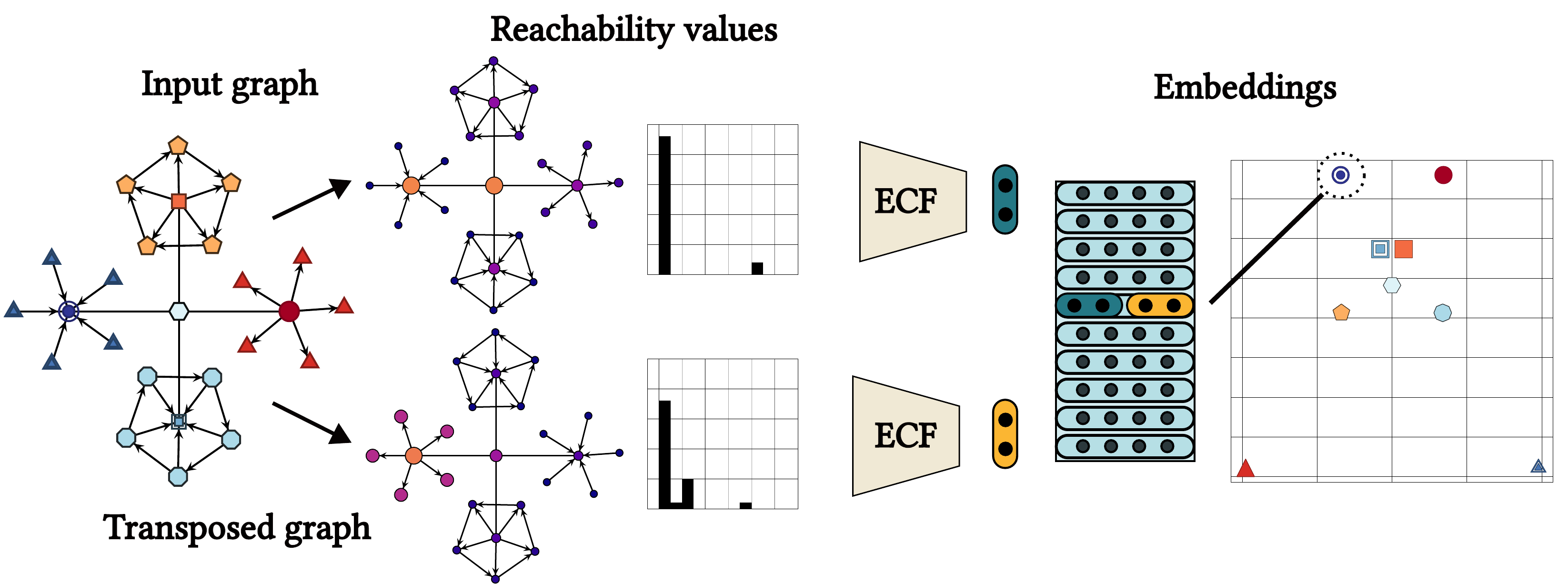}
    \caption{Overview of the \textsc{Digraphwave} structural embedding extraction algorithm. Reachability values are computed via diffusion, and these values are compressed into vectors using the empirical characteristic function. With the transposition enhancement, this process is repeated for the transposed graph and the extracted embeddings are concatenated.
    The final scatter plot shows actual \textsc{Digraphwave} embeddings extracted from the input graph projected into 2D using PCA.
    The colours and shapes of the scatter points and input graph nodes correspond to automorphic equivalence identities. Edges without arrowheads mean that both edges directions are present.
    }
    \label{fig:digraphwave}
    \Description{A set of figures aligned horizontally depicting the Digraphwave embeddding extraction process.}
\end{figure}

The \textsc{Digraphwave} algorithm consists of two parts: the extraction of the core embeddings, and two enhancements.
The extraction of the core embeddings is itself composed of two steps: (1) definition and calculation of node reachability value signatures via the heat diffusion equation, and (2), compression of these signatures into fixed-sized vector embeddings using the empirical characteristic function. 
The two enhancements, transposition and aggregation, address two observed issues with the representational capacity of the core embeddings.
Specifically, the core embeddings are insufficient to distinguish the local structure of nodes with no out-edges, and similarly, some nodes with distinct local structures share similar symmetries, also making them indistinguishable using the core embeddings.
The transposition enhancement resolves the first issue by extracting core embeddings both from the original graph and the transpose graph, i.e., the graph with all edge directions reversed, and the aggregation enhancement addresses the second issue via aggregations of each node's neighbours' embeddings.
The algorithm is visualized in Figure \ref{fig:digraphwave}, albeit without the aggregation enhancement.
As visible in the rightmost plot, \textsc{Digraphwave} produces embeddings where the different automorphic node identities are distinguishable.

The mentioned reachability values are calculated on multiple neighbourhood resolutions by evaluating the diffusion process on a range of timescales. 
As these timescales should neither be too small nor too large, we develop a theoretical analysis of the diffusion process on directed graphs to aid in the choice of timescale values, specifically proving a lower bound on the diffusion heat contained in the local vicinity of an initialization node.
Using this analysis, the choice is reduced to a single integer hyperparameter, $R$, which can be interpreted as an effective radius of the largest neighbourhood to consider for the extraction of the embeddings.
Empirically, we find $R=2$ or $R=3$ to work well for real-world graphs, which commonly have small characteristic path lengths \cite{watts_collective_1998}.

Computation of the reachability values is the main computational challenge to overcome to apply \textsc{Digraphwave} to large graphs.
This is because it involves the matrix exponential, which is known to be numerically challenging to compute \cite{moler_nineteen_2003}, and for which naive implementations will suffer from quadratic space complexity.
We address this challenge by (1) employing truncated Taylor series approximation with a numerical error bound, and (2) computing the embeddings in batches to circumvent the quadratic space complexity.
The resulting implementation is easy to parallelize over multiple CPUs and GPUs, and thereby scales to large graphs.
 
To demonstrate the ability of the \textsc{Digraphwave} embeddings to capture various local connectivity patterns, we construct a synthetic dataset of directed graph with known automorphic equivalence relationships.
Our dataset is based on a previously suggested structural embedding benchmark for undirected graphs \cite{junchen_2021}.
Furthermore, real-world graph datasets are used to show the practical applicability of \textsc{Digraphwave} for node classification and network alignment tasks.
In all the experiments, \textsc{Digraphwave} outperforms recent structural embedding methods for directed graphs, often by a wide margin. 
A new implementation \cite{refex_reimplementation22} of an older method, \textsc{ReFeX} \cite{henderson_its_2011}, does match the performance of \textsc{Digraphwave} in some experiment settings, which is surprising given its simplicity and that it has been ignored in recent benchmarks \cite{qiu_gcc_2020, donnat_learning_2018, junchen_2021, nikolentzos_19, ribeiro_struc2vec_2017, regal_xnetmf_18, riwalk19, ember_19}.

The main contributions of our paper can be summarized as (1) the \textsc{Digraphwave} algorithm for extracting structural node embeddings for directed graphs with a scalable implementation released via our project page\footnote{\url{https://ciwanceylan.github.io/digraphwave-project/}}, (2) a lower bound on the diffusion heat contained in the $R$-hop neighbourhood of initialization node, which is used to rigorously justify the hyperparameter values of \textsc{Digraphwave}, (3) a benchmark dataset of directed graphs with node labels based on automorphic identities, and (4) empirical evidence of the quality of the \textsc{Digraphwave} embeddings compared to other embedding methods.

The rest of the paper is structured as follows. 
Mathematical notation and background concepts related to diffusion on graphs are introduced in Section \ref{sec:diffusion_on_graphs}, followed by the detailed presentation of the \textsc{Digraphwave} algorithm in Section \ref{sec:digraphwave}.
Section \ref{sec:theory_contribution} is devoted to the theoretical analysis used to determine appropriate timescale hyperparameter values, and Section \ref{sec:related_works} treats the relation between \textsc{Digraphwave} and other structural embedding methods for directed graphs.
Section \ref{sec:experiments} contains all details of the performed experiments and their results.
Section \ref{sec:conclusion_and_outlook} concludes the paper with an outlook into future research directions.
Several minor and major proofs are found in the Appendix \ref{app:rw_intepretation}-\ref{app:proofs}, as well as some additional experiments results in Appendix \ref{app:additional_results}.

\section{Diffusion on Graphs} \label{sec:diffusion_on_graphs}

This section establishes background information on diffusion on graphs, starting with mathematical definitions and notation in Section \ref{sec:notation_and_definitions}, see Table \ref{tab:notation} for an overview.
In Section \ref{sec:diffusion_intro}, the diffusion equation on graphs is introduced together with important concepts and properties, and in Section \ref{sec:numerical_computation_of_mat_exp} the numerical computation of the matrix exponential, the general solution to the diffusion equation, is discussed.

\subsection{Notation and Definitions for Graph Concepts} \label{sec:notation_and_definitions}

\begin{table}[t]
    \centering
    \caption{Major symbols and their definitions.}
    \label{tab:notation}
    \begin{tabular}{p{0.17\textwidth}p{0.6\textwidth}}
    \toprule
    \textbf{Symbol} & \textbf{Definition} \\
    \midrule
    $G = (V, E)$ & A directed graph without self-loops with node set $V$ and edge set $E$, where an edge $\edgeji$ has weight $w_{\edgeji} > 0$. \\
    $n=|V|$ & Number of nodes in $G$. \\
    $m=|E|$ & Number of edges in $G$. \\
    $\weight_{\edgeji}$ & The weight of edge $\edgeji$. \\
    $\adj$ & The $n \times n$ weighted adjacency matrix of $G$, with the element at row $i$ and column $j$ corresponding to the edge $\edgeji$. That is $A_{ij} = A_{\edgeji}$. \\
    $[\adj]_{ij} = A_{ij}$ & Bracket notation for an element in a matrix or vector. \\
    $\degD$  & The $n \times n$ diagonal weighted out-degree matrix, see \eqref{eq:adj_deg_def}.\\
    $\Ncal_{\text{out}}(j)$ & The out-neighbourhood of node $j$, i.e., $\{i | \edgeji \in E\}$. \\
    $\Ncal_{\text{in}}(j)$ & The in-neighbourhood of node $j$, i.e., $\{i | \edgeij \in E\}$. \\
    $\Ncal(j)$ & The joint neighbourhood of node $j$, i.e., $\Ncal_{\text{out}}(j) \cup \Ncal_{\text{in}}(j)$.\\
    $\lap = \degD - \adj$ & The $n \times n$ directed graph Laplacian. \\
    $\eyestarb$ & The $n \times n$ identity matrix but with elements corresponding to nodes with no out-neighbours set to zero, see \eqref{eq:star_defs}. \\
    $\degDstarb^{-1}$  & The “inverse” of $\degD$ but with elements corresponding to nodes with no out-neighbours set to zero, see \eqref{eq:star_defs}.\\
    $\lapnorm = \eyestarb - \adj \degDstarb^{-1}$ & The $n \times n$ out-degree normalized Laplacian. \\
    $\weightnorm_{\edgeji}$ & The out-degree normalized weight of edge $\edgeji$. \\
    $\alpha_{\edgeji}$ & Element $ij$ in $\adj \degDstarb^{-1}$. The probability to transition from $j$ to $i$ under a random walk on the graph. See \eqref{eq:alpha_values}. \\
    $\ubold(\tau)$ & A length $n$ vector and the solution to the differential equation \eqref{eq:advection_diffusion_diff_eq} for a given initial condition. \\
    $\bm{\Psi}(\tau) = \expmlap$ & The $n \times n$ reachability value matrix. \\
    $\heatdist$ & Column $j$ of $\bm{\Psi}(\tau)$ and the solution to the diffusion equation \eqref{eq:advection_diffusion_diff_eq} at time $\tau$ with all heat initialized at node $j$. \\
    $\bm{\chi}_j$ & The \textsc{Digraphwave} embedding for node $j$. \\
    $R$ & The main hyperparameter of \textsc{Digraphwave}. Determines the largest scale used for the structural embeddings in terms of numbers of hops. See Sections \ref{sec:digw:reachability_values} and \ref{sec:theory_contribution}. \\
    $\spdist(j, i)$ & Directed shortest path distance from node $j$ to node $i$.\\
    $\Ncore(j, r)$ & The set of nodes within $r$-hops from $j$, i.e., $\{j | \spdist(j, i) \leq r\}$. \\
    $\Nperiph(j, r)$ & The set of nodes which cannot be reached with $r$-hops from $j$, i.e., $V -  \Ncore(j, r)$. \\
    $\mathcal{C}(j, r)$ & The set of nodes exactly $r$-hops from $j$, i.e., $\{i | \spdist(j, i) = r\}$. \\
    \midrule
    \end{tabular}
\end{table}

Let $G=(V, E)$ be a directed graph without self-loops, with $n=|V|$ and $m=|E|$ edges.
The graph may have positive, nonzero, edge weights, i.e., $w_{\edgeji} \in \Rreal_{>0}$ denotes the weight of the edge $\edgeji \in E$.
For an unweighted graph, $w_{\edgeji} = 1$ for all edges.
We define the weighted adjacency matrix $\adj$ and the diagonal weighted out-degree matrix $\degD$ as
\begin{align}
    A_{ij} = A_{\edgeji} =
    \begin{cases}
    w_{\edgeji} & \text{if } \edgeji \in E \\
    0      & \text{otherwise}
    \end{cases},
    & &&
     D_{jj} = \sum_{i \in V} A_{\edgeji} = \sum_{i \in \Ncal_{\text{out}}(j)} w_{\edgeji}. \label{eq:adj_deg_def}
\end{align}
Note that $A_{ij}$, the element of $\adj$ in row $i$ and column $j$, corresponds to an edge $\edgeji$.
In \eqref{eq:adj_deg_def}, we used the notation $\Ncal_{\text{out}}(j) = \{i | \edgeji \in E\}$ for the out-neighbourhood of the node $i$, and we similarly define $\Ncal_{\text{in}}(j) = \{i | \edgeij \in E\}$ as the in-neighbourhood, and $\Ncal(j) = \Ncal_{\text{out}}(j) \cup \Ncal_{\text{in}}(j)$ as the joint neighbourhood.

The weighted and directed graph Laplacian is a central object for studying diffusion on graphs, and we define it in accordance with previous work as $\lap = \degD - \adj$ \cite{ bauer_normalized_2012, veerman_primer_2020}.
This work focuses on the out-degree normalized version of the Laplacian, since it has more convenient numerical stability properties than its unnormalized counterpart.
We define the out-degree normalized graph Laplacian as
\begin{equation} \label{eq:normalized laplacian}
    \lapnorm = \lap \degDstarb^{-1} = \eyestarb - \adj \degDstarb^{-1},
\end{equation}
where $\eyestarb$ and $\degDstarb^{-1}$ are defined as
\begin{align}
    \eyestar_{ij} &= \begin{cases}
    1 & \text{if $i = j$ and $|\Ncal_{\text{out}}(j)| > 0$}\\
    0 & \text{otherwise},
    \end{cases} 
    &&
    \degDstar^{-1}_{ij} = \begin{cases}
    \frac{1}{D_{jj}} & \text{if $i = j$ and $|\Ncal_{\text{out}}(j)| > 0$}\\
    0 & \text{otherwise},
    \end{cases} 
    \label{eq:star_defs}
\end{align}
i.e., diagonal elements corresponding to nodes with no out-neighbours are set to zero.
The elements of $\lapnorm$ can thus be expressed as
\begin{equation} \label{eq:laplacian_elements}
    \mathcal{L}_{ij} = \mathcal{L}_{\edgeji} = \begin{cases}
        1 & \text{if $i = j$ and $|\Ncal_{\text{out}}(j)| > 0$} \\
        - \weightnorm_{\edgeji}  & \text{if $i \neq j$ and $\edgeji \in E$} \\
        0 & \text{otherwise},
    \end{cases}
\end{equation}
where $\weightnorm_{\edgeji}$ are out-degree normalized weights,
\begin{align*}
    \weightnorm_{\edgeji} &= \frac{\weight_{\edgeji}}{\displaystyle \sum_{k \in \Ncal_{\text{out}}(j)} \weight_{\edgeuv{j}{k}}} = \frac{\weight_{\edgeji}}{D_{jj}}.
\end{align*}
Another related and useful definition is
\begin{align}
    \Alphab &= \adj \degDstarb^{-1},
    &&
     \alpha_{\edgeji} = [ \Alphab ]_{ij} = 
    \begin{cases}
        \weightnorm_{\edgeji} & \text{if $\edgeji \in E$} \\
        0 & \text{otherwise,}
    \end{cases}\label{eq:alpha_values}
\end{align}
where each $\alpha_{\edgeji} $ is interpreted as the transition probability from node $j$ to node $i$ under a random-walk on $G$.

\subsection{The Diffusion Equation and its Properties} \label{sec:diffusion_intro}

The heat diffusion equation on $G$ with an initial heat distribution $\mathbf{b} \in \Rreal^n_{+}$ is defined as: 
\begin{align} \label{eq:advection_diffusion_diff_eq}
        \frac{\text{d}\ubold}{\text{d} \tau} (\tau) &= -\lapnorm \ubold (\tau) = (\Alphab - \eyestarb) \ubold &&
        \ubold (\tau=0) = \mathbf{b}, \\
        \label{eq:advection_diffusion_diff_eq_scalar}
        \frac{\text{d}u_i}{\text{d} \tau} (\tau) &= -\sum_{j \in V} \mathcal{L}_{\edgeji} u_j (\tau)=\sum_{j \in V} (\alpha_{\edgeji} - \eyestar_{ij})u_j (\tau) &&
        u_i (\tau=0) = b_i,
\end{align}
where $\ubold (\tau)$ is the heat distribution over all nodes in $G$ at time $\tau$.
The Laplacian $\lapnorm$ has the key property that the sum over its rows is zero,
\begin{align*}
    \sum_{i \in V} \mathcal{L}_{ij} =
    \sum_{i \in V} \mathcal{L}_{\edgeji}
    &= \sum_{i \in V} \eyestar_{ij} -  \alpha_{\edgeji} 
    = \eyestar_{jj} - \sum_{i \in V}  \alpha_{\edgeji}
    = 
    \begin{cases}
        0 & \text{if $|\Ncal_{\text{out}}(j)| = 0$}\\
        1 - \frac{\sum_{i \in \Ncal_{\text{out}}(j)} \weight_{ij}}{\sum_{k \in \Ncal_{\text{out}}(j)}\weight_{kj}}  & \text{otherwise}
    \end{cases}
    = 0.
\end{align*}
This properly directly implies a conservation property for the heat distribution, namely that the total heat, i.e., the sum of elements $u_i(\tau)$, is constant with respect to the diffusion time $\tau$:
\begin{align*}
    \frac{\mathrm{d}}{\mathrm{d}\tau}\sum_{i \in V} u_i(\tau)  
    &=  \sum_{i \in V} \frac{\mathrm{d}u_i}{\mathrm{d}\tau}(\tau) 
    = - \sum_{i \in V} \sum_{j \in V} \mathcal{L}_{ij} u_j(\tau)
    = -\sum_{j \in V} u_j(\tau) \sum_{i \in V} \mathcal{L}_{ij}
    = 0.
\end{align*}
Consequently, the sum of the elements of the initial condition is preserved throughout the diffusion process,
\begin{equation} \label{eq:heat_conservation}
    \sum_{i \in V} u_i(\tau)  = \sum_{i \in V} u_i(0) =  \sum_{i \in V} b_i .
\end{equation}

The general solution to the diffusion equation \eqref{eq:advection_diffusion_diff_eq} is given by the matrix exponential $\expmlap$:
\begin{equation} \label{eq:digraphwave:diff_solution}
    \ubold (\tau) = \expmlap \mathbf{b} = \bm{\Psi}(\tau) \mathbf{b},
\end{equation}
where we introduce $\bm{\Psi}(\tau) = \expmlap$ as a more compact notation.
Following \cite{Higham08}, the matrix exponential itself can be defined for any complex $n \times n$ matrix $\mathbf{X}$ as the Taylor polynomial 
\begin{equation} \label{eq:exp_taylor}
    \exp (\mathbf{X}) = \eye + \mathbf{X} + \frac{\mathbf{X}^2}{2!} + \frac{\mathbf{X}^3}{3!} + \cdots = \sum_{k=0}^\infty \frac{\mathbf{X}^k}{k!}.
\end{equation}
An important property of $\expmlap$ is that all its elements are nonnegative, $[\expmlap]_{ij} \geq 0$.
This follows from the fact that ${-\lapnorm}$ is an \emph{essentially nonnegative} matrix, i.e., its off-diagonal elements are nonnegative, and Theorem 10.29 in \cite{Higham08}.
Consequently, if a nonnegative initial condition is used, $b_i \geq 0$, we also have $u_i(\tau) \geq 0$.

By combining the nonnegativity property with the conservation property \eqref{eq:heat_conservation}, and by using a one-hot initial condition,
\begin{equation*}
    b_i =
    \begin{cases}
        1 & \text{if $i=j$}\\
        0 & \text{otherwise,}
    \end{cases}
\end{equation*}
two intuitive interpretations of the elements of the heat distribution emerge.
The first interpretation is that each $u_i(\tau)$ is the proportion of heat which has been transferred from the initial node $j$ to node $i$ at time $\tau$, which is why we will refer to elements $u_i(\tau)$ as \emph{heat coefficients}.
The second interpretation is that $u_i(\tau)$ is the probability of a random walk process initialized at the node $j$ to reside at the node $i$ after a number of steps sampled from a Poisson distribution with parameter $\tau$, see the Appendix \ref{app:rw_intepretation} for the mathematics of this interpretation.
This second interpretation motivates the use of the term \emph{reachability values} for the elements $u_i(\tau)$.

We will use the terms heat coefficients and reachability values interchangeably in this work, depending on which interpretation the best suites the context.
Moreover, since the column $j$ of $\bm{\Psi}(\tau)$ exactly corresponds to the heat distribution $\ubold (\tau)$ resulting from a one-hot initialization at the node $j$, the terminology of heat coefficients and reachability values will likewise be used for elements of the matrix exponential $\bm{\Psi}(\tau)$ and its columns $\heatdist = \bm{\Psi}(\tau) \mathbf{b}$.

A final property of the diffusion equation relevant for the design of \textsc{Digraphwave} is the existence of stationary solutions and the behaviour of $\ubold (\tau)$ for large $\tau$.
In the undirected case, any heat distributions proportional to the degree of each node in a weakly connected component of the graph is stationary, see Appendix \ref{app:stationary} for a short proof.
Detailed analysis of stationary solutions for directed graphs requires additional theoretical machinery, and is beyond the scope of this work.
We refer interested readers to other sources \cite{bauer_normalized_2012, veerman_primer_2020}.
Important for this work is that the heat will eventually get trapped in sinks made up from strongly connected components in the graph.


\subsection{Numerical Computation of the Matrix Exponential} \label{sec:numerical_computation_of_mat_exp}

Numerical computation of the matrix exponential is a well-studied subject and there exist several methods, both general and for matrices with specific characteristics \cite{moler_nineteen_2003}.
For nondefective matrices, the exponential can be computed accurately via the eigenvalue decomposition,
\begin{equation} \label{eq:wavelets_eigen}
    \exp(\mathbf{X}) = \mathbf{U} \exp (\Lambda) \mathbf{U}^T
\end{equation}
where $\mathbf{U} \Lambda \mathbf{U}^T = \mathbf{X}$ is the eigenvalue decomposition of $\mathbf{X}$.
The Laplacian matrix for directed graphs, however, may be defective, and thus the eigenvalue decomposition is not guaranteed to exist or be numerically unstable.
While alternative decomposition methods, e.g., Jordan canonical form or Schur, can be used \cite{moler_nineteen_2003, horn_johnson_2012}, decomposition methods generally fail to scale to large graphs since the outer matrix factors, i.e., $\mathbf{U}$, are dense $n \times n$ matrices and can thus not be stored in the computer memory.
Quadratic space complexity is not only an issue for decomposition methods.
It is a general difficulty caused by $\exp(\mathbf{X})$ typically being dense even for sparse $\mathbf{X}$, meaning that any method attempting to compute the entire matrix $\exp(\mathbf{X})$ at once will not be scalable.
In Section \ref{sec:digw:reachability_values_computation}, we address this issue via  batch computation.

Batched computation of $\exp(\mathbf{X})$ can be implemented via series approximation, for which there exist different choices.
Chebyshev series approximation have previously been used in the context of machine learning on graphs \cite{Defferrard16, donnat_learning_2018}.
This is a good choice for matrices with eigenvalues constrained to an interval on the real line, e.g., $[0,2]$ for the normalized Laplacian of an undirected graph \cite{chung1997spectral, von_luxburg_tutorial_2007}, since Chebyshev approximation can be made “near-best” on such intervals \cite{trefethen_approximation_2019}.
However, for matrices with eigenvalues containing imaginary components, the Chebyshev approximation error is not guaranteed to converge to zero \cite{moler_nineteen_2003}, which is the case for the Laplacian matrices of directed graphs.

Faber series approximation \cite{ellacott_83} is a natural replacement, since they generalize the near-minimax property of Chebyshev approximation from an interval on the real line to a general elliptical region in the complex plane \cite{geddes_78}.
Faber series approximation have also previously been used to compute the solution to the directed diffusion equation on a 2D lattice \cite{bergamaschi_efficient_2003}.
However, to use the Faber series approximation, one first needs to find an ellipse which bounds the eigenvalues of the matrix in the complex plane, e.g., by computing the eigenvalues with the largest real and imaginary parts respectively, which, in our experience, becomes a computational bottleneck for large matrices.

Explicit computation of eigenvalues can be avoided by instead using a general bound.
Specifically, the eigenvalues of $\lapnorm$ for any graph lie in a unit circle in the complex plane centred at $1$ \cite{bauer_normalized_2012, Egerstedt}.
Applying this bound, Faber series approximation is reduced to Taylor series approximation \cite{geddes_78}.
Thus, while \textsc{Digraphwave} can be implemented using Faber series in general, we only consider truncated Taylor series approximation in this work.

\begin{figure}[htp]
\centering
\begin{subfigure}{0.5\linewidth}
    \centering
    \includegraphics[width=0.88\linewidth]{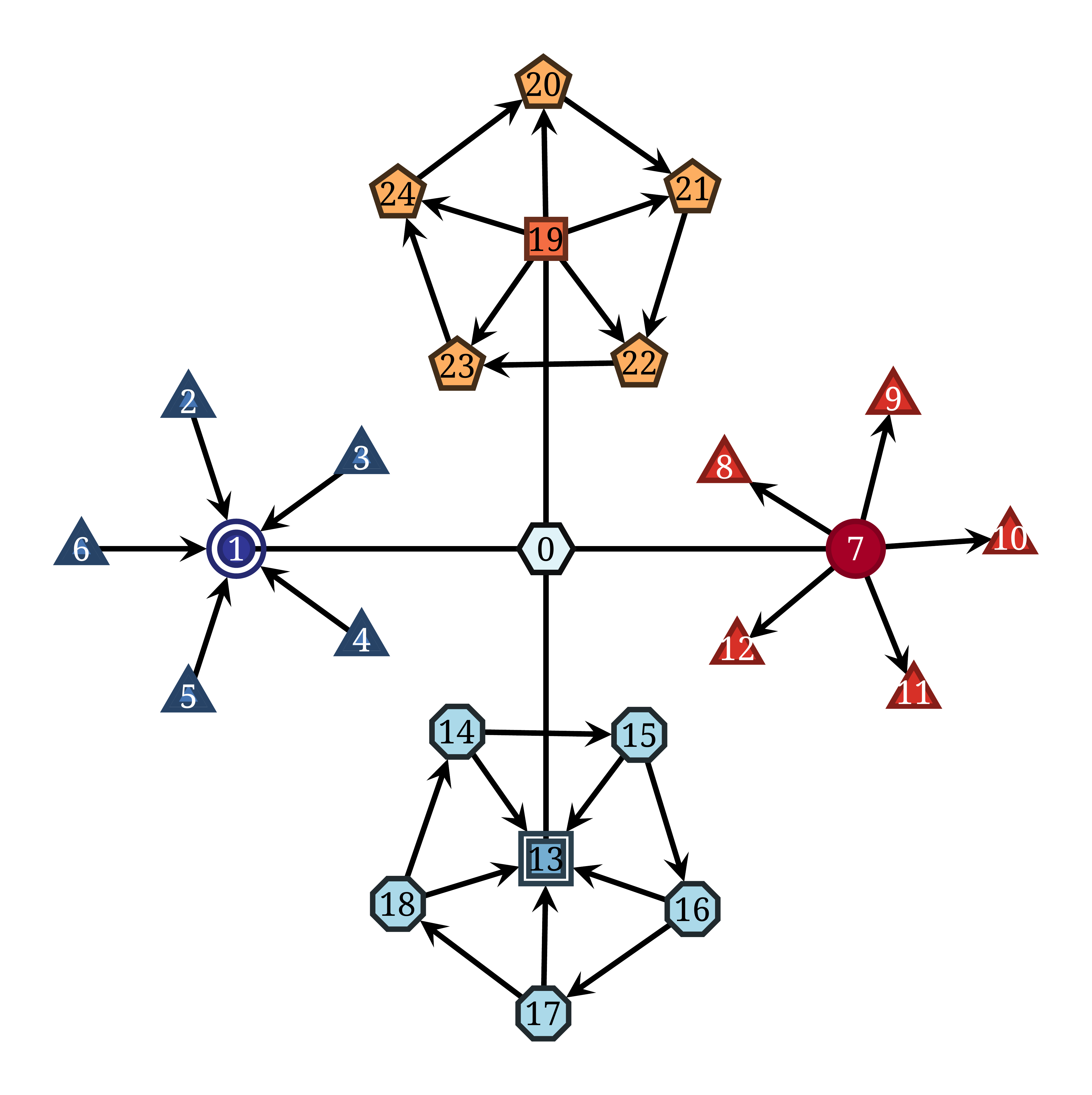}
    \caption{Example graph}
    \label{fig:example_graph}
\end{subfigure}
~
\begin{subfigure}{0.5\linewidth}
    \centering
    \includegraphics[width=0.88\linewidth]{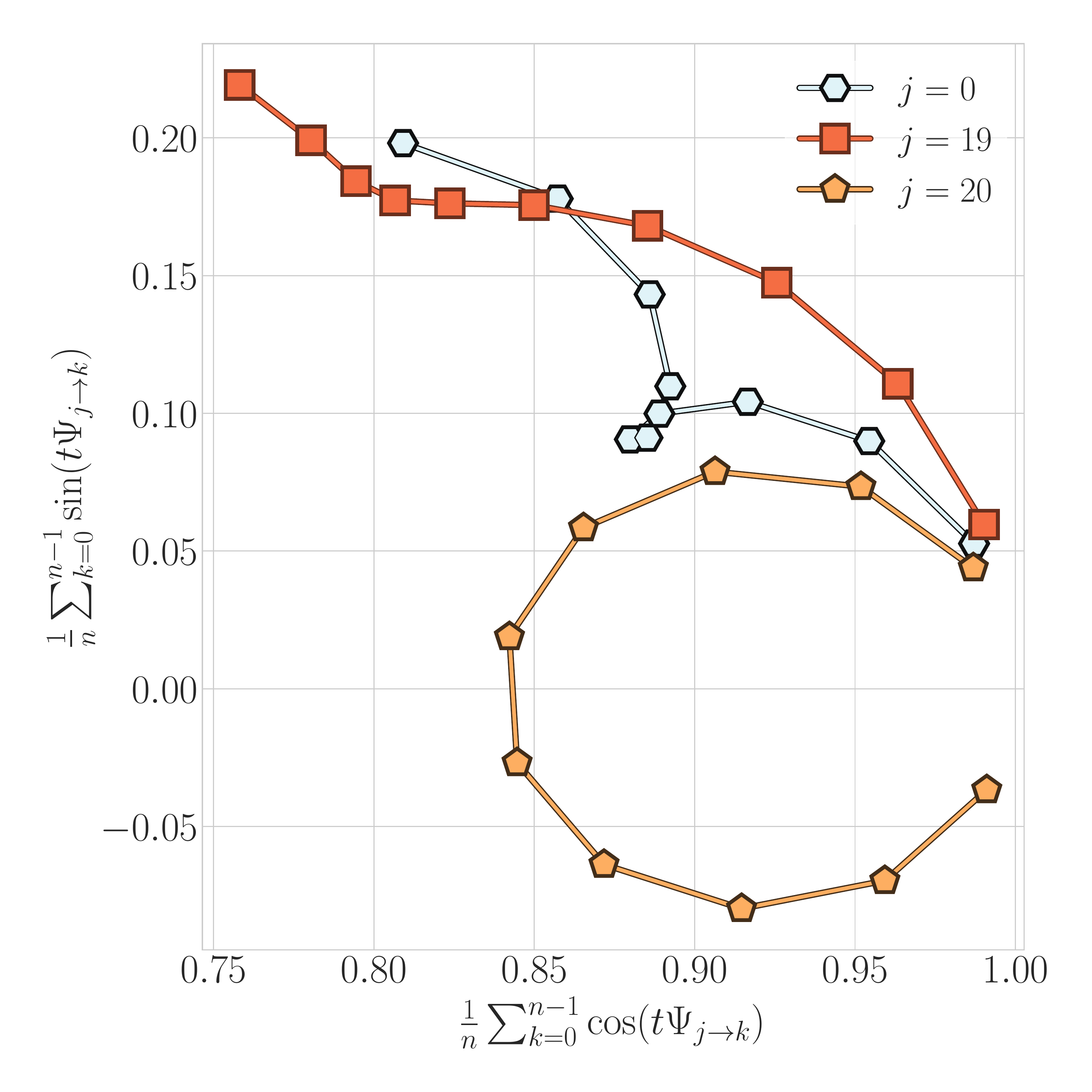}
    \caption{ECF values for node 0, 19 and 20}
    \label{fig:ecf_vis}
\end{subfigure}
\\
  \begin{subfigure}{0.32\linewidth}
    \centering
    \includegraphics[width=\linewidth]{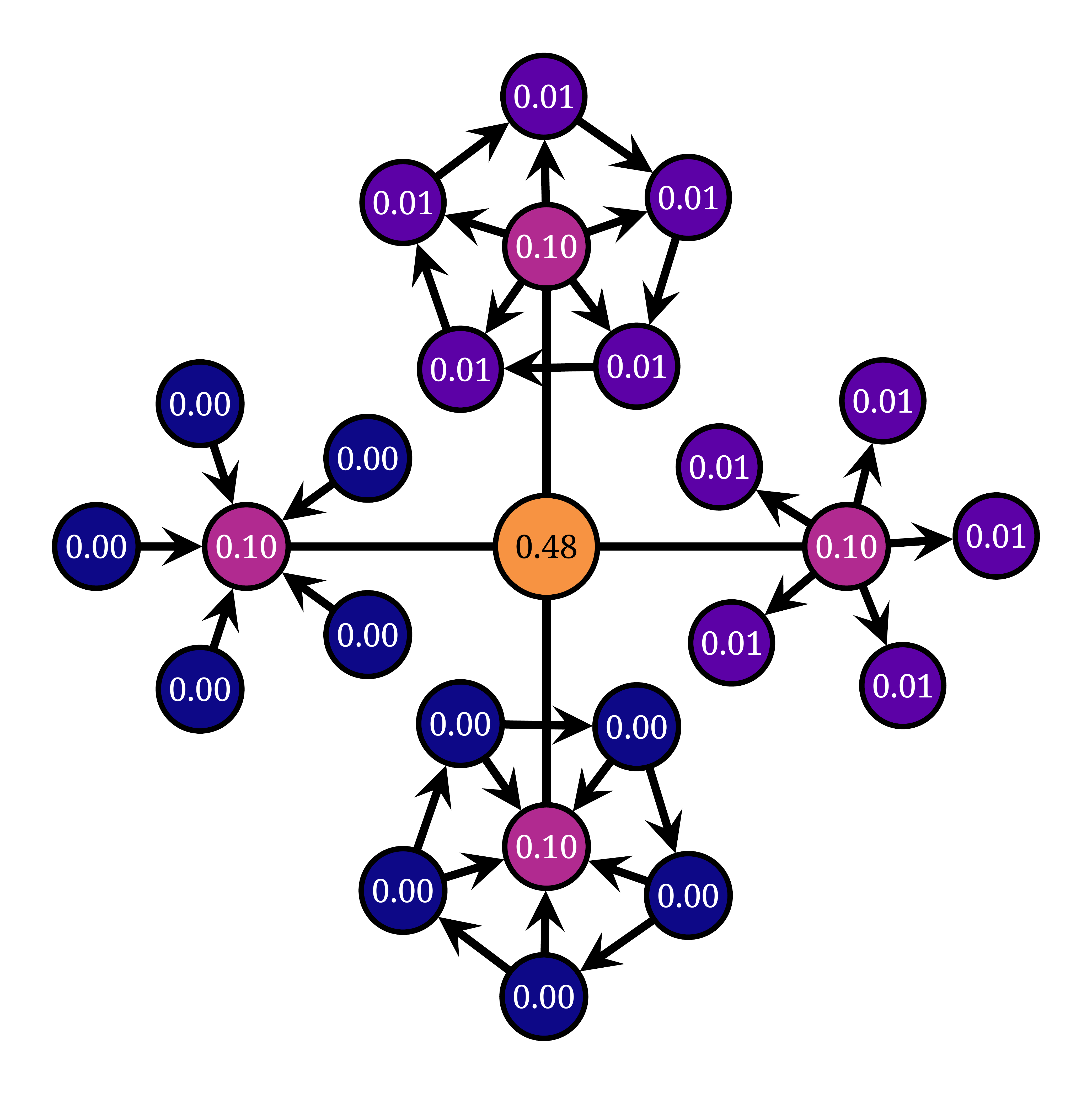}
    \caption{Node 0}
    \label{fig:illustrative_0_graph}
  \end{subfigure}
~
  \begin{subfigure}{0.32\linewidth}
    \centering
    \includegraphics[width=\linewidth]{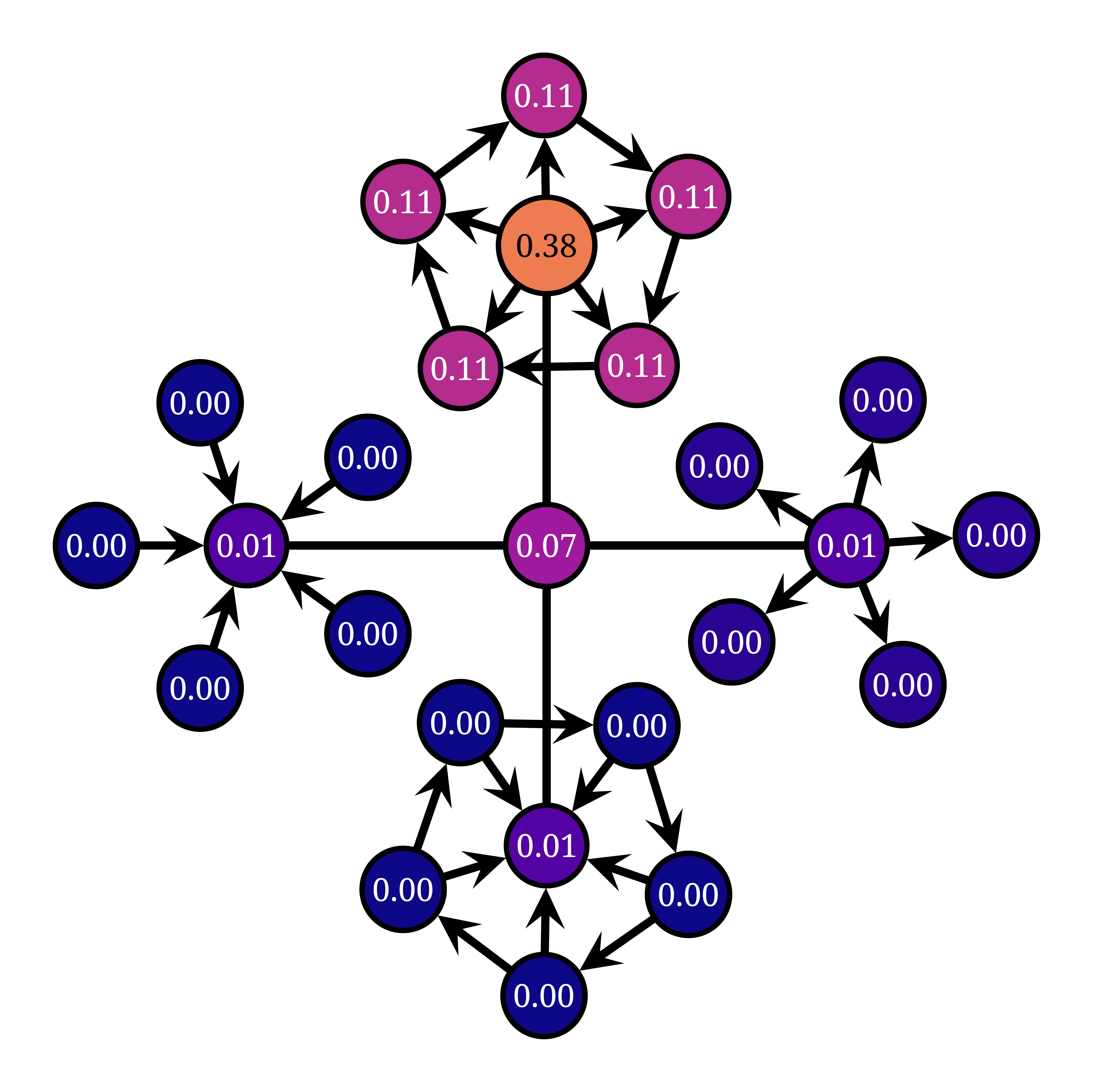}
    \caption{Node 19}
    \label{fig:illustrative_19_graph}
  \end{subfigure}
~
  \begin{subfigure}{0.32\linewidth}
    \centering
    \includegraphics[width=\linewidth]{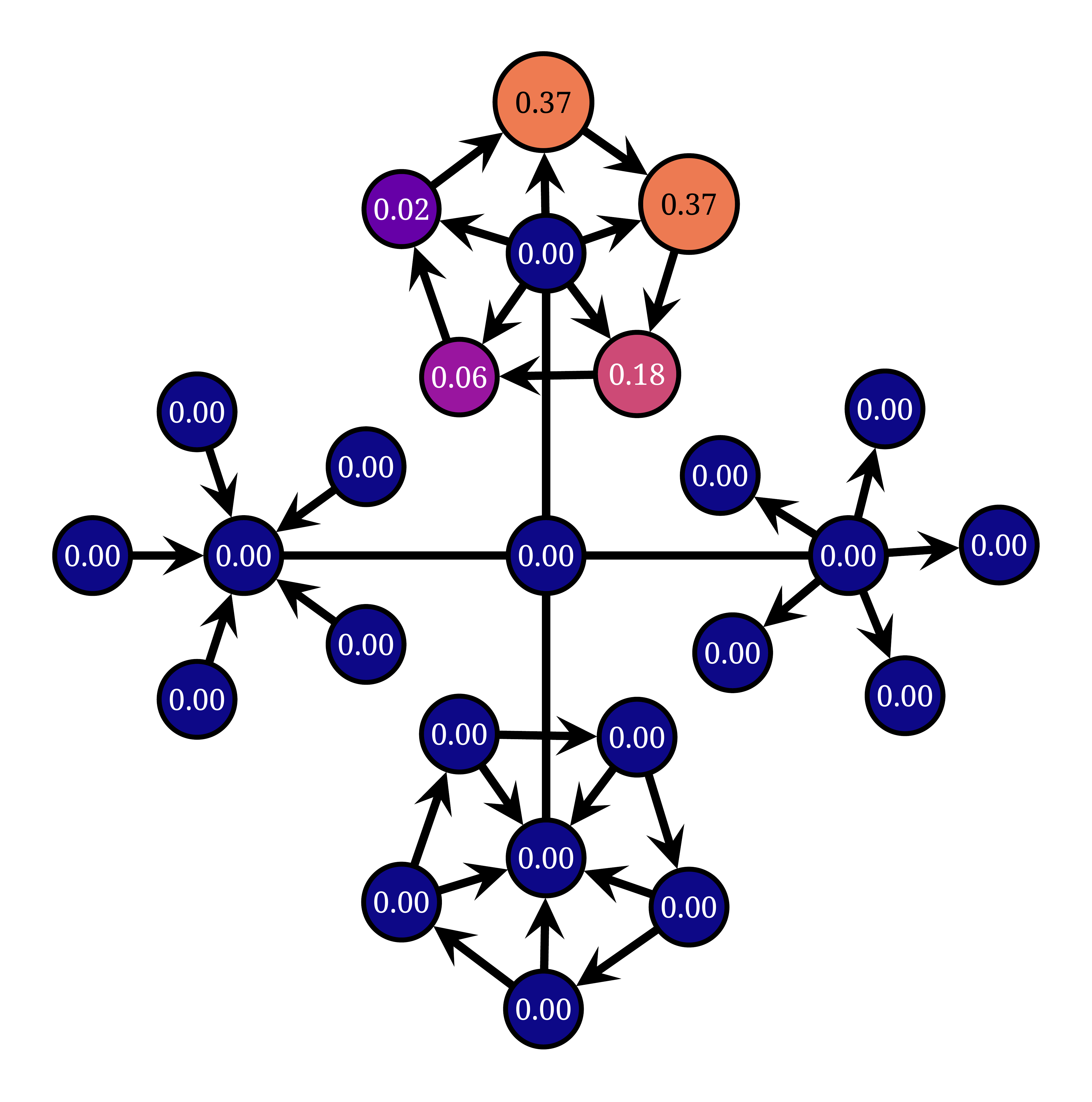}
    \caption{Node 20}
    \label{fig:illustrative_20_graph}
  \end{subfigure}
\\
  \begin{subfigure}{0.32\linewidth}
    \centering
    \includegraphics[width=\linewidth]{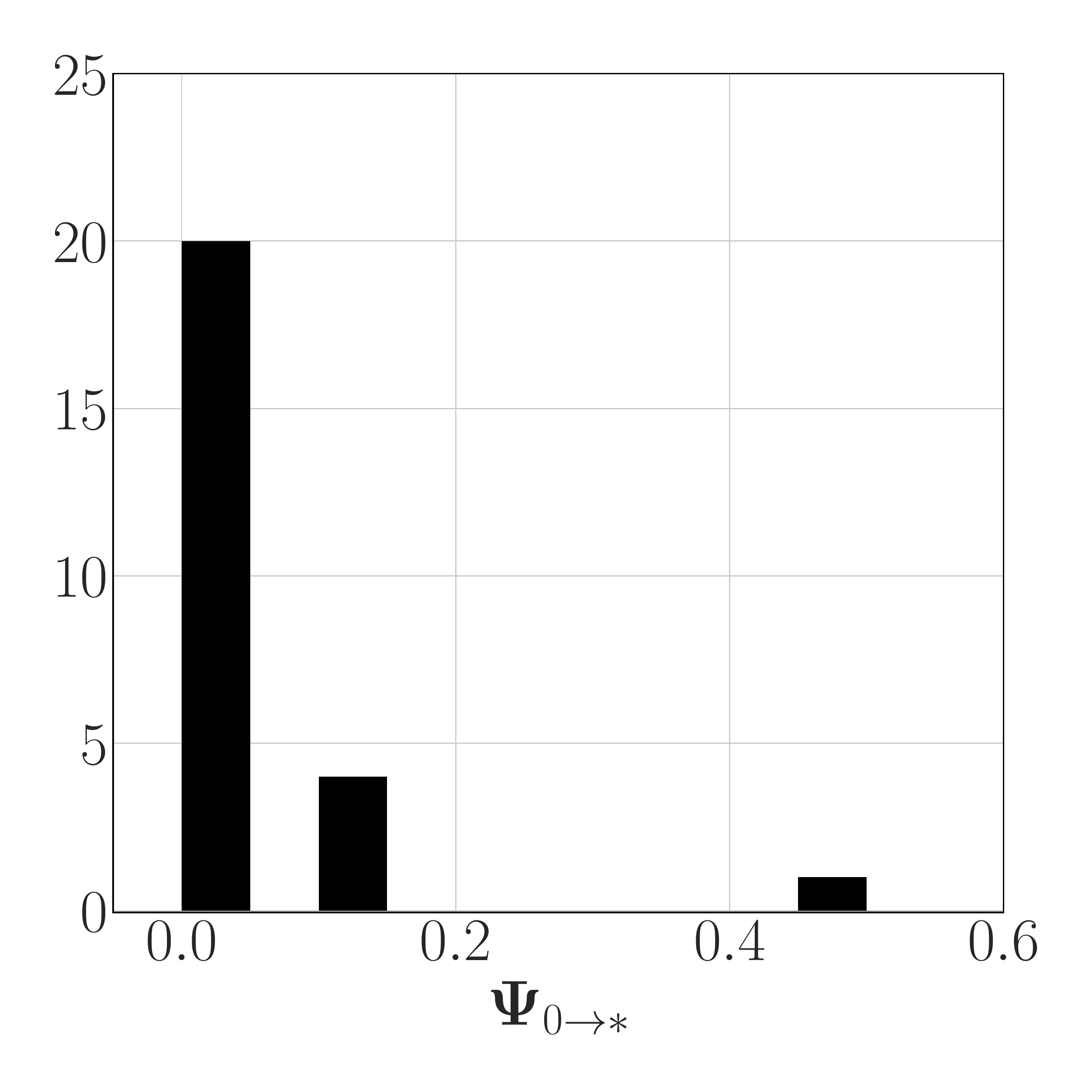}
    \caption{Node 0}
    \label{fig:illustrative_0_hist}
  \end{subfigure}
~
  \begin{subfigure}{0.32\linewidth}
    \centering
    \includegraphics[width=\linewidth]{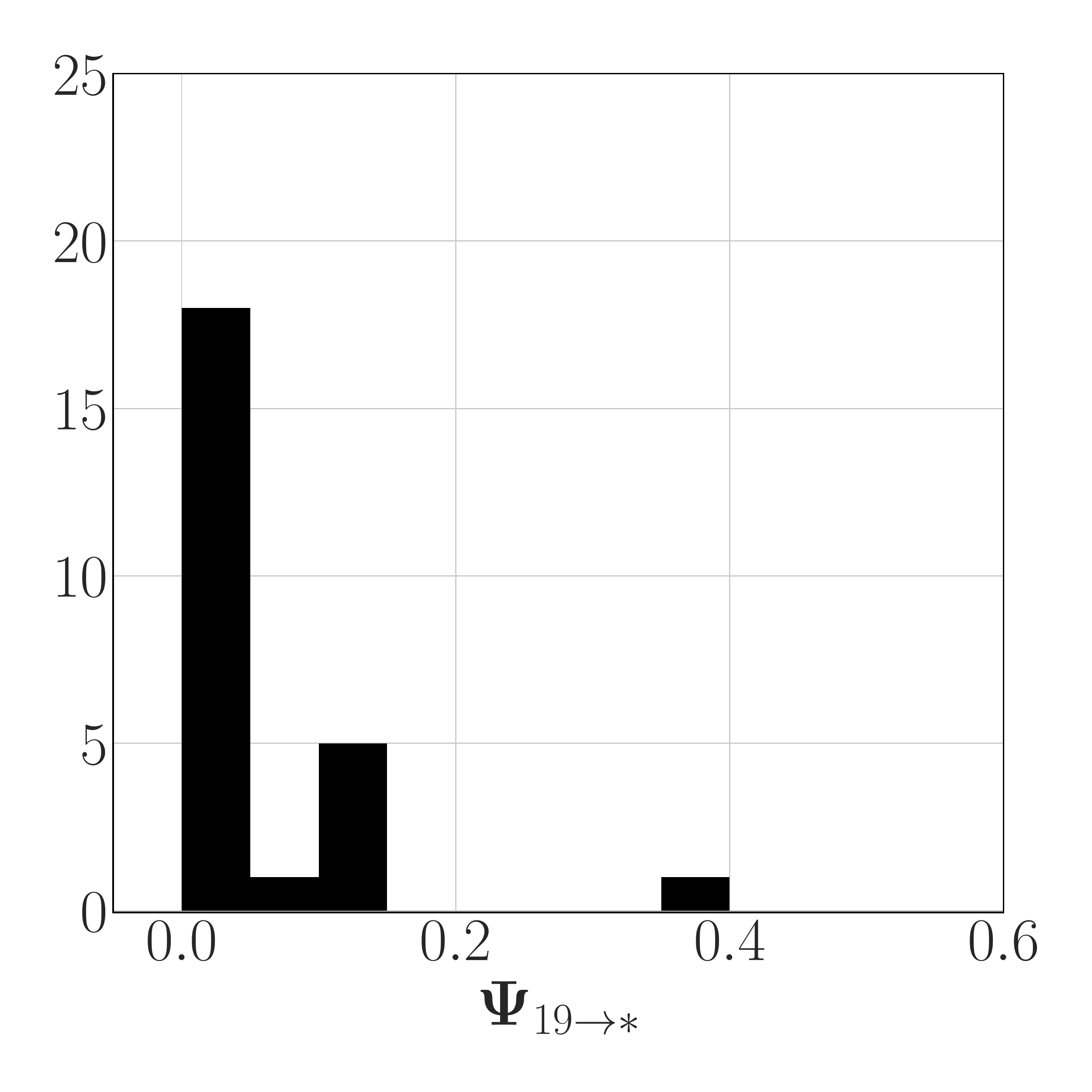}
    \caption{Node 19}
    \label{fig:illustrative_19_hist}
  \end{subfigure}
~
  \begin{subfigure}{0.32\linewidth}
    \centering
    \includegraphics[width=\linewidth]{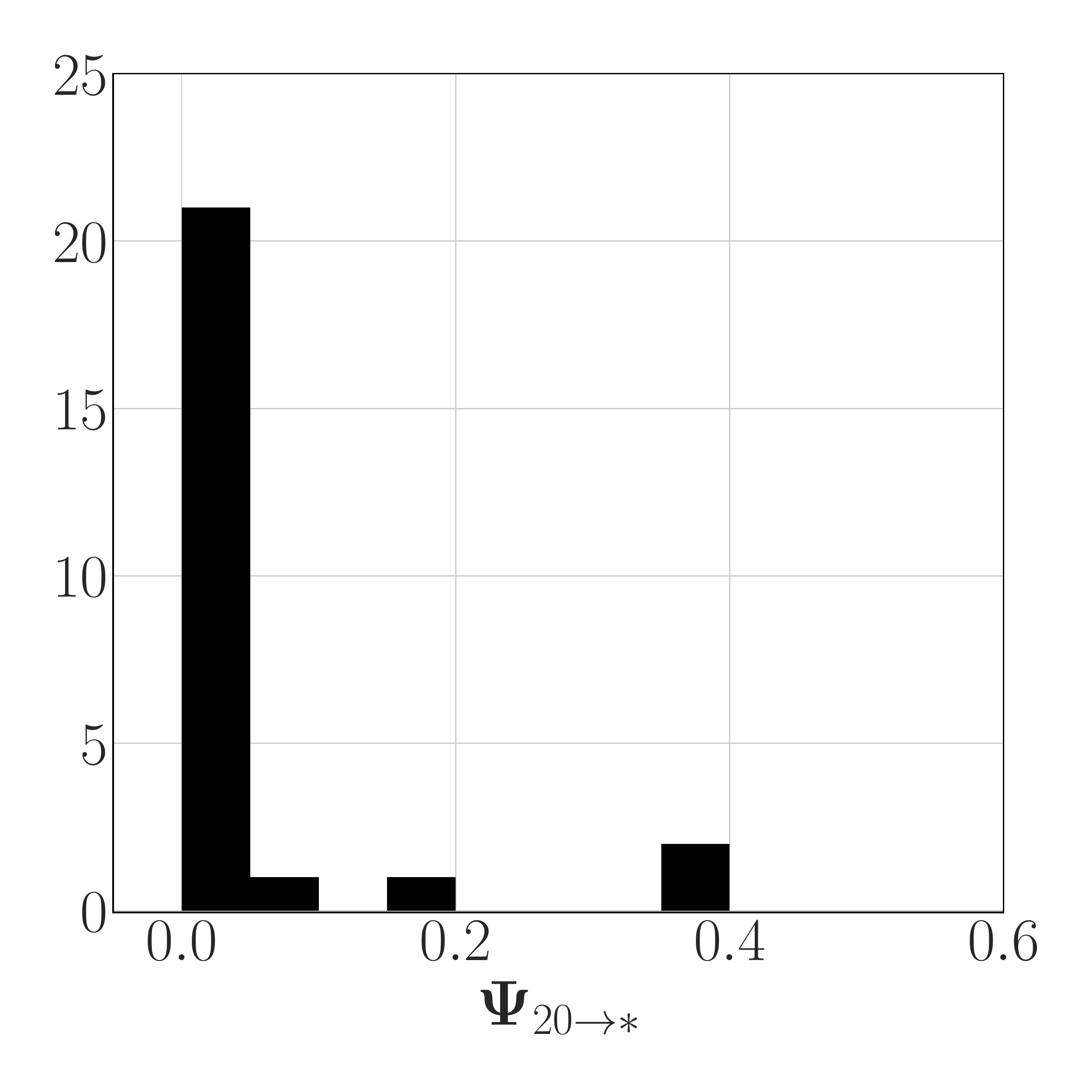}
    \caption{Node 20}
    \label{fig:illustrative_20_hist}
  \end{subfigure}
\caption{
\subref{fig:example_graph} shows an example graph with nodes labelled by their node index and with colours and shapes corresponding to their automorphic identity.
\subref{fig:illustrative_0_graph}—\subref{fig:illustrative_20_graph} show the reachability values for nodes 0, 19 and 20 for $\tau=1$ as heat distributions over the graph nodes.
\subref{fig:illustrative_0_hist}—\subref{fig:illustrative_20_hist} show the same reachability values as histograms.
\subref{fig:ecf_vis} shows a visualization of the reachability values after compression using the empirical characteristic function in the complex plane for $t \in \{0.5 \pi, \pi, \dots 5 \pi\}$. Here we use $\Delta t = 0.5 \pi $ instead of $\Delta t = \pi$ for visual clarity. See Section \ref{sec:digw:compress} for the ECF details.
}
\label{fig:illustrative_example}
\Description{Eight images composed in a grid. The first shows a graph with node labelled by their automorphic identities. The second shows a visualisation of the Digraphwave embeddings. The six lower images visualises the reachability values for three nodes in the graph, first as heat spead in the graph, then as histograms.}
\end{figure}

\section{\textsc{Digraphwave}} \label{sec:digraphwave}

In this section, the \textsc{Digraphwave} algorithm is described in detail.
The core algorithm consists of two parts: computation of reachability values for various timescales and compression of these into structural node embeddings using the empirical characteristic function.
To aid the explanation, we use an example graph shown in Figure \ref{fig:example_graph}.
Each node in the graph is labelled with its node index and the shape and colour indicates automorphic equivalence, i.e., nodes with the same shape and colour share all graph theoretic properties (e.g., in-/out-degree, centralities etc.) \cite{junchen_2021}.

In Section \ref{sec:digw:reachability_values}, we discuss the reachability values' relation to local connectivity patterns, as well as the timescale parameters and how to choose them appropriately. 
Then, in Section \ref{sec:digw:reachability_values_computation}, the scalable computation of the reachability values is described, followed by the compression into node embeddings using the empirical characteristic function, Section \ref{sec:digw:compress}.
The two enhancements to the core embeddings, transposition and aggregation, are motivated and detailed in Section \ref{sec:enhancements}.
Finally, in Section \ref{sec:digw:finalize}, the appropriate values for the remaining hyperparameters of \textsc{Digraphwave} are discussed, and the full algorithm is presented as pseudocode, see Algorithms \ref{alg:expm}-\ref{alg:hyperparameters}.

\subsection{Reachability Values and Timescales} \label{sec:digw:reachability_values}
In Section \ref{sec:diffusion_on_graphs}, we defined the matrix of reachability values $\bm{\Psi}(\tau)$ as the general solution to the directed diffusion equation \eqref{eq:advection_diffusion_diff_eq} using the matrix exponential
\begin{equation} \label{eq:psi_def}
    \bm{\Psi}(\tau) = \expmlap.
\end{equation}
In Figures \ref{fig:illustrative_0_graph}-\subref{fig:illustrative_20_graph}, we visualize the column vectors $\heatdist$ for nodes 0, 19 and 20 in the example graph, Figure \ref{fig:example_graph}.
Again, the values on each node can be interpreted either as the proportion of transferred heat, or as the probability of a random walk ending on that node.
Moreover, when treated as sets, each column vector $\heatdist$ provides a signature of the topology of the neighbourhood surrounding the node $j$, with $\tau$ determining the neighbourhood resolution.
The signatures for the nodes 0, 19 and 20 are visualized as histograms in Figures \ref{fig:illustrative_0_graph}-\subref{fig:illustrative_20_graph}.
It is clear that any nodes which are automorphically equivalent will share signatures $\heatdist$ for any timescale $\tau$.
Moreover, for a correctly chosen timescale, the signatures can also capture distinct local structures.
This raises the question of which timescale value $\tau$ to use.
Since real-world graphs often display a hierarchical structure \cite{Hierarchical_networks}, different patterns may appear at different neighbourhood resolutions. 
Therefore, to capture these different resolutions, we use $k_\tau$ evenly-spaced values for $\tau$ in an interval $[\taumin, \taumax]$, reminiscent of the heat kernel signature used for shape recognition in computer graphics \cite{heat_kernel_signature}.
The next question, then, is to find appropriate values for $\taumin$, $\taumax$ and $k_\tau$.

As can be discerned from the discussion about stationary solutions in Section \ref{sec:diffusion_intro}, the values of $\taumin$ and $\taumax$ are crucial for the reachability values to contain local structural information, as $\tau$ should be neither too small nor too large.
For large values of $\tau$, the signatures $\heatdist$ will converge to uninformative stationary distributions.
For undirected graph, this distribution will be identical for all nodes in the same weakly connected component, while the signal will eventually get trapped in sets of sink nodes for directed, see \cite{bauer_normalized_2012, veerman_primer_2020} for details.
Additionally, at $\tau=0$, $\bm{\Psi}$ will equal the $n \times n$ identity matrix, which is also uninformative.
The cases of too small and too large timescales are illustrated for a small graph on three nodes in Figure \ref{fig:triplet_heat}.
In the figure, the entire evolution of the reachability values over time $\tau$ for nodes A and B is shown.
For $\tau=0$ the signatures are identical and by $\tau=4$ the two heat distributions are close to the convergence to stationary solution and can no longer be distinguished.
Around $\tau=1$ the signatures are the most different.

From this example, we conclude that we desired guarantees on $\taumin$ and $\taumax$ such that $\tau \in [\taumin, \taumax ]$ are neither too small nor too large to produce informative signatures.
To provide such guarantees, we formulate and prove Theorem \ref{thm:main_theorem} in Section \ref{sec:theory_contribution}.
Using the heat diffusion interpretation of $\bm{\Psi}(\tau)$, the theorem says that the proportion of heat remaining at most $R$ hops away from the initialization node is bounded from below by the normalized upper incomplete gamma function, see \cite{NIST_DLMF8.2.E4} for its definition.
Mathematically, this is expressed as:
\begin{equation} \label{eq:heat_bound}
    \sum_{i \in \Ncore(j, R)} \Psi_{\edgeuv{j}{i}}(\tau) \geq  Q(R+1, \tau) = \frac{\Gamma(R+1, \tau)}{\Gamma(R+1)}.
\end{equation}
We plot $Q(R+1, \tau)$ as a function of $\tau$ for several values of $R$ in Figure \ref{fig:heat_containment}.
Since the function $Q(R+1, \tau)$ is invertible with respect to $\tau$, one may choose a desirable proportion of heat, $a$, to be contained in $R$-egonet of the initialization node, and use a $\taumax = Q^{-1}(R+1, a)$ to guarantee this.
For \textsc{Digraphwave}, we simplify the choice even further by using the bound $Q(R+1, R) > 0.5$ \cite{NIST_DLMF8.10.E13}, that is, by setting $\taumax = R$ we are always guaranteed that at least half of the total heat remains within $R$ hops from the initialization node.

For the lower bound of the interval, we use $\taumin=1$.
This has two motivations. 
First, Figure \ref{fig:triplet_heat} shows that $\tau=1$ provides a good signal to distinguish nodes $A$ and $B$, and already for $\tau=2$ the signal is much weaker.
Second, with the random-walk interpretation of $\heatdist$, see Appendix \ref{app:rw_intepretation}, $\tau$ marks the mean walk length, meaning that $\tau=1$ is the smallest timescales to use for a random walk to leave the initial node on average.

Thus, the timescale interval is $[1, R]$ and we have replaced the unintuitive timescale hyperparameter with the new hyperparameter $R$: the radius of the largest neighbourhood scale to consider.
It might be possible to choose $R$ automatically through consideration of different graph statistics, e.g., based on the characteristic path distance or diameter, bur we leave exploration of such schemes to future work.
The remaining hyperparameter, $k_\tau$, the number of timescales to use, is address in Section \ref{sec:digw:finalize}.

\begin{figure}[htp]
\centering
\begin{minipage}[t]{.49\textwidth}
    \centering
    \resizebox{0.95\linewidth}{!}{
\begin{tikzpicture}[
simplenode/.style={circle, draw, fill=black!30, thick, inner sep=1pt, minimum size=12pt},
]

\definecolor{chocolate217952}{RGB}{217,95,2}
\definecolor{darkcyan27158119}{RGB}{27,158,119}
\definecolor{darkslategray38}{RGB}{38,38,38}
\definecolor{lightgray204}{RGB}{204,204,204}

\begin{axis}[
axis line style={lightgray204},
legend cell align={left},
legend style={fill opacity=0.6, draw opacity=1, text opacity=1, draw=white},
tick align=outside,
tick pos=left,
x grid style={lightgray204},
xlabel=\textcolor{darkslategray38}{\Large \(\displaystyle \tau\)},
xmajorgrids,
xmin=-0.33, xmax=6.93,
xtick style={color=darkslategray38},
y grid style={lightgray204},
ylabel=\textcolor{darkslategray38}{\(\displaystyle \mathbf{\Psi}(\tau)\)},
ymajorgrids,
ymin=-0.05, ymax=1.05,
ytick style={color=darkslategray38}
]
\addplot [very thick, darkcyan27158119]
table {%
0 1
0.1 0.909365376538991
0.2 0.83516002301782
0.3 0.774405818047013
0.4 0.724664482058611
0.5 0.683939720585721
0.6 0.650597105956101
0.7 0.623298481970803
0.8 0.600948258997328
0.9 0.582649444110793
1 0.567667641618306
1.1 0.555401579181167
1.2 0.545358976644706
1.3 0.537136789107167
1.4 0.530405031312609
1.5 0.524893534183932
1.6 0.520381101989183
1.7 0.516686634980163
1.8 0.513661861223646
1.9 0.511185385928083
2 0.509157819444367
2.1 0.507497788410239
2.2 0.506138669951534
2.3 0.505025917872317
2.4 0.50411487352451
2.5 0.503368973499543
2.6 0.50275828221038
2.7 0.502258290471306
2.8 0.501848931858241
2.9 0.501513777372688
3 0.501239376088333
3.1 0.501014715318148
3.2 0.500830778636587
3.3 0.500680184018774
3.4 0.500556887573922
3.5 0.500455940982777
3.6 0.500373292904188
3.7 0.500305626380565
3.8 0.50025022571672
3.9 0.50020486748949
4 0.500167731313951
4.1 0.500137326784986
4.2 0.500112433662089
4.3 0.500092052896834
4.4 0.500075366537548
4.5 0.500061704902043
4.6 0.500050519700919
4.7 0.500041362032778
4.8 0.500033864368245
4.9 0.500027725799716
5 0.500022699964881
5.1 0.500018585159342
5.2 0.500015216241504
5.3 0.500012458004866
5.4 0.500010199751706
5.5 0.500008350850395
5.6 0.500006837098033
5.7 0.500005597742421
5.8 0.500004583043868
5.9 0.500003752278958
6 0.500003072106177
6.1 0.500002515227804
6.2 0.500002059294354
6.3 0.500001686007617
6.4 0.500001380386286
6.5 0.500001130164703
6.6 0.500000925300599
};
\addlegendentry{$\mathbf{\Psi_{A \rightarrow *}}(\tau)$}
\addplot [very thick, darkcyan27158119, forget plot]
table {%
0 0
0.1 0.0453173117305045
0.2 0.0824199884910902
0.3 0.112797090976493
0.4 0.137667758970695
0.5 0.158030139707139
0.6 0.17470144702195
0.7 0.188350759014598
0.8 0.199525870501336
0.9 0.208675277944603
1 0.216166179190847
1.1 0.222299210409417
1.2 0.227320511677647
1.3 0.231431605446417
1.4 0.234797484343696
1.5 0.237553232908034
1.6 0.239809449005408
1.7 0.241656682509918
1.8 0.243169069388177
1.9 0.244407307035959
2 0.245421090277816
2.1 0.246251105794881
2.2 0.246930665024233
2.3 0.247487041063842
2.4 0.247942563237745
2.5 0.248315513250229
2.6 0.24862085889481
2.7 0.248870854764347
2.8 0.249075534070879
2.9 0.249243111313656
3 0.249380311955833
3.1 0.249492642340926
3.2 0.249584610681707
3.3 0.249659907990613
3.4 0.249721556213039
3.5 0.249772029508611
3.6 0.249813353547906
3.7 0.249847186809718
3.8 0.24987488714164
3.9 0.249897566255255
4 0.249916134343024
4.1 0.249931336607507
4.2 0.249943783168955
4.3 0.249953973551583
4.4 0.249962316731226
4.5 0.249969147548978
4.6 0.249974740149541
4.7 0.249979318983611
4.8 0.249983067815877
4.9 0.249986137100142
5 0.249988650017559
5.1 0.249990707420329
5.2 0.249992391879248
5.3 0.249993770997567
5.4 0.249994900124147
5.5 0.249995824574802
5.6 0.249996581450984
5.7 0.249997201128789
5.8 0.249997708478066
5.9 0.249998123860521
6 0.249998463946912
6.1 0.249998742386098
6.2 0.249998970352823
6.3 0.249999156996191
6.4 0.249999309806857
6.5 0.249999434917648
6.6 0.249999537349701
};
\addplot [very thick, chocolate217952, dashed]
table {%
0 1
0.1 0.907101397287475
0.2 0.826945388047901
0.3 0.757612019364366
0.4 0.697492264047125
0.5 0.645235190149177
0.6 0.599704371025064
0.7 0.559941892881106
0.8 0.525138611557275
0.9 0.494609551925696
1 0.467773541394874
1.1 0.444136331439623
1.2 0.423276594278454
1.3 0.40483429107059
1.4 0.388500997627108
1.5 0.374011847166181
1.6 0.361138809991919
1.7 0.349685079516449
1.8 0.339480374722616
1.9 0.330377002575359
2 0.32224655134049
2.1 0.31497710833161
2.2 0.308470914156934
2.3 0.30264238079756
2.4 0.297416413406961
2.5 0.292726986061721
2.6 0.288515930212357
2.7 0.284731901605528
2.8 0.28132949724173
2.9 0.278268498714548
3 0.275513222228099
3.1 0.273031958855853
3.2 0.270796491307477
3.3 0.268781675710007
3.4 0.266965078767124
3.5 0.265326662202548
3.6 0.26384850767574
3.7 0.262514576425452
3.8 0.261310498786443
3.9 0.260223389467647
4 0.259241685101343
4.1 0.258355001093374
4.2 0.257554005241284
4.3 0.256830305954517
4.4 0.256176353220308
4.5 0.255585350720143
4.6 0.255051177722776
4.7 0.254568319567237
4.8 0.254131805708633
4.9 0.25373715443532
5 0.253380323481983
5.1 0.253057665862429
5.2 0.252765890331133
5.3 0.252502025955888
5.4 0.252263390347159
5.5 0.25204756114443
5.6 0.251852350407258
5.7 0.251675781599946
5.8 0.251516068894622
5.9 0.251371598548863
6 0.251240912141422
6.1 0.251122691473645
6.2 0.251015744965325
6.3 0.250918995392323
6.4 0.25083146882973
6.5 0.250752284678841
6.6 0.250680646669073
};
\addlegendentry{$\mathbf{\Psi_{B \rightarrow *}}(\tau)$}
\addplot [very thick, chocolate217952, dashed]
table {%
0 0
0.1 0.0906346234610091
0.2 0.16483997698218
0.3 0.225594181952987
0.4 0.275335517941389
0.5 0.316060279414279
0.6 0.349402894043899
0.7 0.376701518029197
0.8 0.399051741002672
0.9 0.417350555889207
1 0.432332358381694
1.1 0.444598420818833
1.2 0.454641023355294
1.3 0.462863210892833
1.4 0.469594968687391
1.5 0.475106465816068
1.6 0.479618898010817
1.7 0.483313365019837
1.8 0.486338138776354
1.9 0.488814614071917
2 0.490842180555633
2.1 0.492502211589761
2.2 0.493861330048466
2.3 0.494974082127683
2.4 0.49588512647549
2.5 0.496631026500457
2.6 0.49724171778962
2.7 0.497741709528694
2.8 0.498151068141759
2.9 0.498486222627312
3 0.498760623911667
3.1 0.498985284681852
3.2 0.499169221363413
3.3 0.499319815981226
3.4 0.499443112426078
3.5 0.499544059017223
3.6 0.499626707095812
3.7 0.499694373619435
3.8 0.49974977428328
3.9 0.49979513251051
4 0.499832268686049
4.1 0.499862673215014
4.2 0.499887566337911
4.3 0.499907947103166
4.4 0.499924633462452
4.5 0.499938295097957
4.6 0.499949480299081
4.7 0.499958637967222
4.8 0.499966135631755
4.9 0.499972274200284
5 0.499977300035119
5.1 0.499981414840658
5.2 0.499984783758496
5.3 0.499987541995134
5.4 0.499989800248294
5.5 0.499991649149605
5.6 0.499993162901967
5.7 0.499994402257579
5.8 0.499995416956132
5.9 0.499996247721042
6 0.499996927893823
6.1 0.499997484772196
6.2 0.499997940705646
6.3 0.499998313992383
6.4 0.499998619613714
6.5 0.499998869835297
6.6 0.499999074699401
};
\addplot [very thick, chocolate217952, dashed]
table {%
0 0
0.1 0.00226397925151567
0.2 0.00821463496991892
0.3 0.0167937986826477
0.4 0.0271722180114857
0.5 0.0387045304365439
0.6 0.0508927349310373
0.7 0.0633565890896969
0.8 0.0758096474400531
0.9 0.0880398921850971
1 0.099894100223432
1.1 0.111265247741544
1.2 0.122082382366252
1.3 0.132302498036577
1.4 0.141904033685501
1.5 0.150881687017751
1.6 0.159242291997264
1.7 0.167001555463714
1.8 0.17418148650103
1.9 0.180808383352724
2 0.186911268103877
2.1 0.192520680078628
2.2 0.1976677557946
2.3 0.202383537074757
2.4 0.206698460117549
2.5 0.210641987437822
2.6 0.214242351998023
2.7 0.217526388865778
2.8 0.220519434616512
2.9 0.22324527865814
3 0.225726153860235
3.1 0.227982756462295
3.2 0.23003428732911
3.3 0.231898508308767
3.4 0.233591808806798
3.5 0.235129278780229
3.6 0.236524785228448
3.7 0.237791049955113
3.8 0.238939726930277
3.9 0.239981478021843
4 0.240926046212609
4.1 0.241782325691612
4.2 0.242558428420806
4.3 0.243261746942316
4.4 0.24389901331724
4.5 0.244476354181901
4.6 0.244999341978142
4.7 0.245473042465541
4.8 0.245902058659613
4.9 0.246290571364396
5 0.246642376482898
5.1 0.246960919296913
5.2 0.247249325910372
5.3 0.247510432048978
5.4 0.247746809404546
5.5 0.247960789705965
5.6 0.248154486690775
5.7 0.248329816142475
5.8 0.248488514149246
5.9 0.248632153730095
6 0.248762159964755
6.1 0.248879823754159
6.2 0.248986314329029
6.3 0.249082690615294
6.4 0.249169911556556
6.5 0.249248845485863
6.6 0.249320278631525
};
\end{axis}

\fill[fill=white, opacity=0.8] (1.7,3.8) rectangle (3.9,5.3);
\coordinate (O) at (2.8,5);
\coordinate (A) at (2.8-0.85,5-0.85);
\coordinate (B) at (2.8+0.85,5-0.85);
\node[simplenode, draw=darkcyan27158119] (0) at (O) {\small A};
\node[simplenode, draw=chocolate217952] (1) at (A) {\small B};
\node[simplenode] (2) at (B) {};

\path [latex-latex, thick] (0) edge node {} (1);
\path [latex-latex, thick] (0) edge node {} (2);

\end{tikzpicture}
    }
    \Description{%
    A line plot and a small graph with three nodes. The center node is labelled as A and the left node is labelled as B. The plot shows the reachability values signatures for the two labelled nodes over time.%
    }
    \captionof{figure}{%
    The reachability values for nodes A and B in the three node graph. For $\tau=0$ the signatures are identical, and again for $\tau > 4$ the signatures are not distinguishable. The largest signature discrepancy is observed for $\tau=1$.%
    }
    \label{fig:triplet_heat}
\end{minipage}
\hfill
\centering
\begin{minipage}[t]{.49\textwidth}
    \centering
    \resizebox{0.95\linewidth}{!}{
    \input{figures/theory/heat_containment}
    }
    \Description{%
    The normalized upper incomplete gamma function ploted as eight decreasing and injective functions of the timescale tau for different R values. A dashed curve showing the case when tau equals R intersects the other eight functions approaches y equal to 0.5.%
    }
  \captionof{figure}{%
  Visualisation of the bound \eqref{eq:heat_bound}. Each curve shows the minimal amount of heat contained in the $R$-egonet of any directed graph under the diffusion process \eqref{eq:advection_diffusion_diff_eq}. The dashed curve shows the case when $R = \tau$ and never drops below $0.5$.%
  }
  \label{fig:heat_containment}
 \end{minipage}
\end{figure}

\begin{figure}[htp]
\centering
\begin{minipage}[t]{.49\textwidth}
    \centering
    \resizebox{0.95\linewidth}{!}{
    \input{figures/theory/taylor_error_bound}
    }
    \Description{%
    A plot depicting six error bound curves. The each curve is labelled with the order of its corresponding Taylor polynomial. The curves increase linearly with the timescale until they reach the point where the approximation error exceeds the round-off error. At these points, the error bound rapidly increases.
    }
  \captionof{figure}{%
  The error bound for the truncated Taylor series approximation of the matrix exponential $\expmlap$, with $K$ as the order of the Taylor polynomial. The bound includes both the Taylor approximation error and the rounding error due to double-precision floating point arithmetic. The point where the approximation error bound exceedes the rounding error bound is seen as a steep increase.%
  }
  \label{fig:taylor_error_bound}
 \end{minipage}
\hfill
\begin{minipage}[t]{.45\textwidth}
    \centering
    \includegraphics[width=0.9\linewidth]{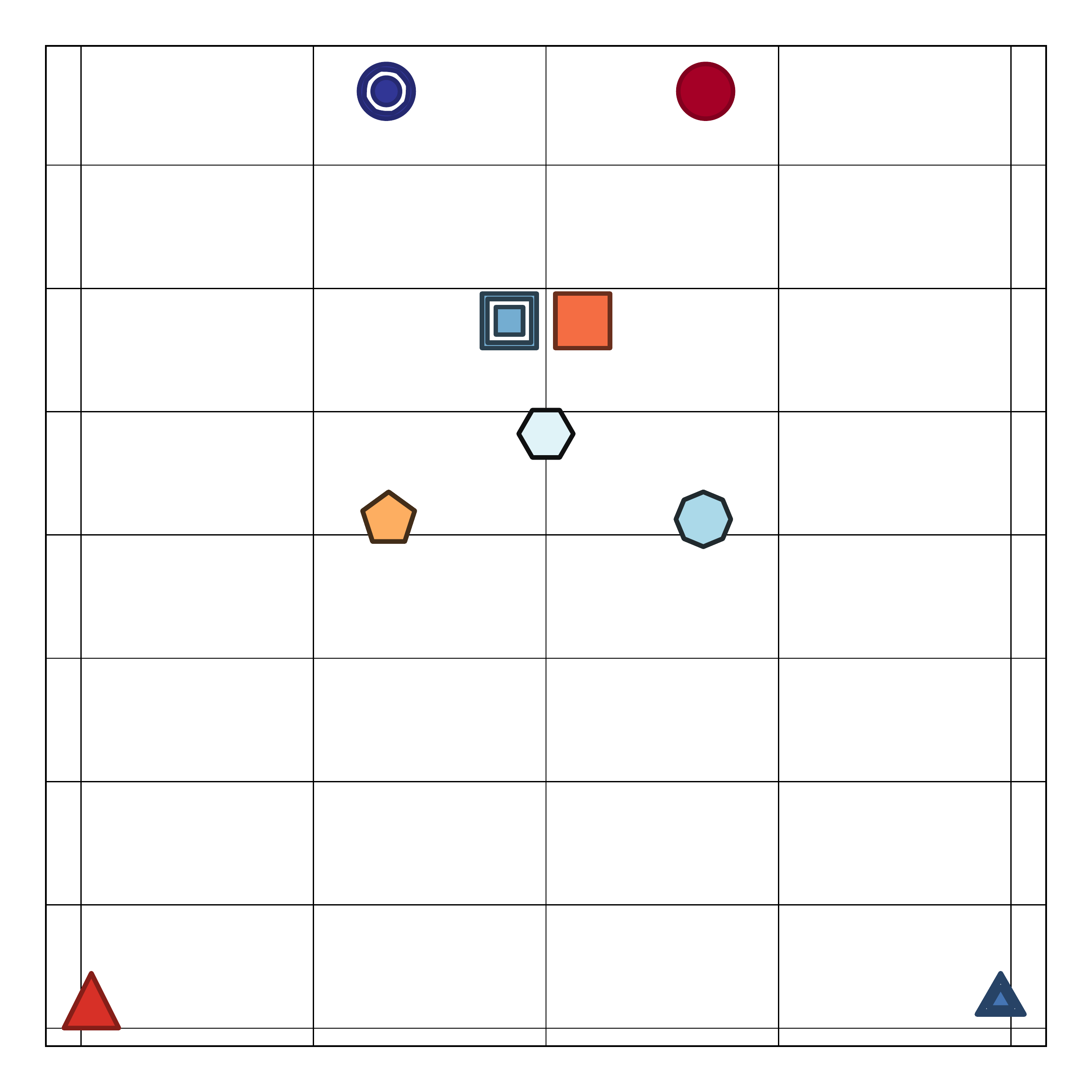}
    \Description{%
    A plot with unlabelled axis depicting nine different shapes, each correponding to a set of nodes with the same automorphic identities in a previous figure.%
    }
    \captionof{figure}{%
    Digraphwave embeddings for the example graph in Figure \ref{fig:example_graph}. The embeddings are projected into the 2D plane using PCA. The x-axis seem to capture the proportion of in- to out-degree for each node, while the y-axis correlates with the nodes' total degree. %
    }
    \label{fig:digw_embeddings_pca}
\end{minipage}
\end{figure}

\subsection{Scalable Computation of the Reachability Values} \label{sec:digw:reachability_values_computation}

Scalable computation of the reachability values requires the space complexity to be manageable.
As explained in Section \ref{sec:numerical_computation_of_mat_exp}, the matrix exponential $\expmlap$ is typically dense even if $\lapnorm$ is sparse, resulting in $\bigO(n^2)$ space complexity for naive implementations.
We address this issue by computing the columns of $\bm{\Psi}(\tau)$, i.e., the signatures $\heatdist$, in batches using truncated Taylor series approximation.
Since our implementation permits parallelism, the computation can be scaled and performed on large graphs.

The $K$th order truncated Taylor series approximation of $\bm{\Psi}(\tau)$ expanded in the identity matrix $\eye$ is 
\begin{align}
    \label{eq:taylor_approx}
    \mathbf{T}_{K}(\lapnorm, \tau) &= \sum_{k=0}^K a_k(\tau) \left(\lapnorm - \eye \right)^k, &&
    a_k(\tau) = \frac{(-\tau)^k e^{-\tau}}{k!}.
\end{align}
Given a list of $n_{\text{batch}}$ node indices $B=(v_0, \dots, v_{n_{\text{batch}} - 1})$, we can express the corresponding columns of $\bm{\Psi}(\tau)$ as
\begin{align}
    \label{eq:expm_batch}
    \heatdistbatchbx{*} = \expmlap \mathbf{B},
\end{align}
where $\mathbf{B}$ is a $n \times n_{\text{batch}}$ matrix of one-hot columns,
\begin{equation} \label{eq:b-batch}
    [\mathbf{B}]_{ij} = 
    \begin{cases}
        1 & \text{if $i = v_j \in B$ } \\
        0 & \text{otherwise}.
    \end{cases}
\end{equation}
Thus, $\heatdistbatchbx{*}$ is also a $n \times n_{\text{batch}}$ matrix, as is each term in the corresponding Taylor series:
\begin{align}
    \label{eq:taylor_approx_batch}
    \mathbf{T}_{K}(\lapnorm, \tau) \mathbf{B} &= \sum_{k=0}^K a_k(\tau) \left(\lapnorm - \eye \right)^k \mathbf{B}.
\end{align}
Computing a batch of $n_{\text{batch}}$ signatures thus only requires $\bigO(n n_{\text{batch}})$ memory.
By slightly increasing the space complexity of a batch to $\bigO ( k_\tau n n_{\text{batch}} )$, the monomials $\left(\lapnorm - \eye \right)^k \mathbf{B}$ can be reused for several values of $\tau$, since $\tau$ only affects the coefficients in \eqref{eq:taylor_approx_batch}.

The monomials $\left(\lapnorm - \eye \right)^k \mathbf{B}$ are computed using sparse-dense matrix multiplication, meaning that the time complexity for a batch of reachability values is $\bigO ( k_\tau K m n_{\text{batch}} )$.
The matrix multiplications are easily parallelized over both over $m$ and $k_\tau$ through a GPU supported implementation using \textsc{pytorch\_sparse}, part of the \textsc{PyG} framework \cite{Fey/Lenssen/2019}.
Furthermore, separate batches can be parallelized over multiple GPUs, further improving the scalability.
The computation of the Taylor series coefficients $a_k$ is implemented using the Fast Fourier transform.
Readers not familiar with this technique may consult the Appendix \ref{sec:app:fft} for a short derivation.
Algorithm \ref{alg:expm} shows pseudocode for the full computation of a batch of reachability values.

A remaining question is which order of the Taylor approximation, $K$, to use.
Generally, a low $K$ value is desirable for a lower computational cost.
However, a too low value will result in a high Taylor approximation error, and a too large error will render the embeddings senseless. 
To ensure that $K$ is large enough, we use the following the error bound, derived in Appendix \ref{sec:app:error_bound} based on Theorem 4.5 and Theorem 4.8 in \cite{Higham08},
\begin{align} \label{eq:lapexpm_error_bound}
    \| \expmlap - \mathbf{\hat{T}}_K (\lapnorm, \tau) \|_1
    &\leq
    \left(\frac{\exp(\tau)\tau^{K+1}}{(K+1)!} 
    + 
    \tilde{\gamma}_{Kn}
    \right)
    \exp(\tau), &&
    \tilde{\gamma}_{Kn} = cKnu / (1 - cKnu),
\end{align}
where $\mathbf{\hat{T}}_K (\lapnorm, \tau)$ is the approximation \eqref{eq:taylor_approx} with rounding errors from floating point arithmetic included, and $\| \cdot \|_1$ is the maximum column sum matrix norm,
\begin{align} \label{eq:matrix_norm}
    \| \mathbf{X} \|_1 = \max_{1\leq j \leq n} \sum_{i=1}^n | x_{ij} |.
\end{align}
The first term in the bound is the Taylor approximation error bound, and the second term a bound on the rounding error.
For the rounding error bound, $c$ is a small integer constant whose precise value does not matter, $u=2^{-53} \approx 1.11 \times 10^{-16}$ for IEEE double precision arithmetic and $u=2^{-24} \approx 5.96 \times 10^{-8}$ for IEEE single precision arithmetic. 

In Figure \ref{fig:taylor_error_bound}, the error bound is shown as a function of $\tau$ for several values of $K$ using $n=10^7$ and double precision.
Single precision is not shown, since $\tilde{\gamma}_{Kn} > 1$ for $n=10^7$.
The error generally grows with $\tau$, and the growth rate increases drastically once the Taylor approximation error surpasses the rounding error.
Thus, we want to choose $K$ so that the rounding error dominates for the largest timescale used, i.e., $\taumax=R$.
For this work, $K=40$ and double precision is used to guarantee small errors even for $R=8$.

It is worth noting that the error bound \eqref{eq:lapexpm_error_bound} is pessimistic in two ways:
the upper bound on the Taylor error is loose, and the rounding error term assumes dense-dense matrix multiplication.
Furthermore, we currently do not know how sensitive the quality of the final embeddings are to errors in the reachability values.
It is clear that the embeddings should become meaningless if the error is very large that, however, it is possible that errors which would be unacceptable in a numerical approximation setting will still produce useful embeddings in a machine learning or data mining context.
This means smaller values for $K$ and even single precision may be used in practice, and still produce meaningful embeddings.

\subsection{Compressing Reachability Values to Structural Node Embeddings} \label{sec:digw:compress}

Each computed batch of reachability values $\heatdistbatchbx{*}$ are compressed into node embeddings by evaluating the empirical characteristic function (ECF), 
\begin{align} \label{eq:digw:ecf}
    \phi_j (t, \tau) = \frac{1}{n} \sum_{k=0}^{n-1} e^{i t \Psi_{\edgeuv{j}{k}}(\tau)} 
\end{align}
for $k_\phi$ different $t$ values, $\mathbf{t} = (t_0, \dots, t_{k_{\phi} - 1})$.
The $k_\phi$ complex ECF values are then concatenated into $2 k_\phi$ dimensional embeddings,
\begin{equation}\label{eq:digw:core_emb_tau}
    \bm{\chi}_j (\tau) = \left[\Re ( \phi_j (t, \tau) ), \Im ( \phi_j (t, \tau)  )   \right]_{t\in \mathbf{t}}.
\end{equation}
for each timescale.
These are then further concatenated into the final core embeddings:
\begin{equation}\label{eq:digw:core_emb_new}
    \bm{\chi}_j = \left[ \bm{\chi}_j (\tau_s) \right]_{s \in \{0, \dots, k_\tau - 1 \}}.
\end{equation}

As far as we can tell, \cite{donnat_learning_2018} is the first work to use the ECF compression this way, with \cite{zhu_2021_proximity_is_all_you_need} following suit.
For the sampling points $\mathbf{t}$, both \cite{donnat_learning_2018} and \cite{zhu_2021_proximity_is_all_you_need} use $t_0=0$ and $t_{k_{\phi} - 1}=100$, and then sample equally spaced values with step size $\Delta t = 100 / (k_{\phi} - 1)$.
Since we have not been able to discern the motivation for these values, we instead propose to use principles from signal processing to choose the sampling points.
We observe that $\phi_j (t, \tau)$ is a band-limited signal in $t$, with maximum frequency $f_{\text{max}} = 1/2\pi$ since $\Psi_{\edgeuv{j}{k}}(\tau)) \in [0, 1]$.
Thus, Shannon's sampling theorem \cite{shannon} applies, and the signal is completely determined given a series of ordinates spaced $1/2f_{\text{max}}  = \pi$ apart.
Consequently, we use $\Delta t = \pi$.
Moreover, since $\phi_j (0, \tau) = 1$ irrespective of $\bm{\Psi}(\tau)$, we use $t_0=\Delta t = \pi$, and thus $t_{k_{\phi} - 1} = k_{\phi} \pi$.

In Figure \ref{fig:ecf_vis}, the core embeddings for nodes $0$, $19$ and $20$ in the example graph are visualized in the 2D plane by plotting the points $(\cos(t \Psi_{\edgeuv{j}{k}}(\tau), \sin(t \Psi_{\edgeuv{j}{k}}(\tau)))$ using $\tau=1$ and step size $\Delta t = \pi / 2$.
The three nodes produce distinct patterns reflecting their different local structure in the example graph.

Next, we look at the practical implementation of the ECF compression.
For large graphs, it is likely that many reachability values are very small and can well be approximated by zeros in \eqref{eq:digw:ecf}.
Doing so both removes small elements which may be dominated by numerical errors, see Section \ref{sec:digw:reachability_values_computation}, and reduces the time complexity of applying the ECF to a batch of reachability values.
For a dense $\heatdistbatchbx{*}$, the time complexity of the compression is $\bigO(k_\phi n n_{\text{batch}})$, which can be reduced $\bigO(k_\phi \texttt{nnz}(\heatdistbatchbx{*}))$ by using a sparse $\heatdistbatchbx{*}$. 
Thus, we apply thresholding to the reachability values as 
\begin{equation}
    \bar\Psi_{\edgeji}(\tau) = 
    \begin{cases}
       \Psi_{\edgeji}(\tau) & \text{if $\Psi_{\edgeji}(\tau) > \theta_{j}(R)$} \\
        0 & \text{otherwise,}
    \end{cases}
\end{equation}
where each $\theta_{j}(R)$ is a threshold specific per node.
By using node specific threshold values, $\heatdistbatchbx{*}$ can be made sparser compared to using a global threshold value.

To obtain threshold values, we use the theory developed in Section \ref{sec:theory_contribution}.
Specifically, we consider the extreme case defined in Lemma \ref{lemma:ss_analytical_expression} where initial heat spreads maximally over the graph, and set the threshold values to the smallest reachability value in an $R$-egonet surrounding node $j$.
That is, we set
\begin{align}
    \label{eq:threshold}
	\tilde\theta_{j}(R) = \min_{\tau \in [\taumin,  \taumax], l \in [0, R]}  [\exp (-\tau \lapnormstar(d_j, \beta_j, R+1))]_{\edgeuv{0}{l*}},
\end{align}
where $[\exp (-\tau \lapnormstar(d_j, \beta_j, R+1))]_{\edgeuv{0}{l*}}$ are the reachability values from node $0$ to nodes at distance $l$ away in a source-star graph -- see Section \ref{sec:source_star} -- constructed using the same out-degree $d_j$ as node $j$, and a corresponding branch factor 
\begin{align}
    \beta_j = \frac{1}{n-1} \sum_{k \in V, k \neq j} d_k. \label{eq:betaj}
\end{align}
Evaluation of \eqref{eq:threshold} using the $\taumin=1$ and $\taumax=R$ gives
\begin{align*}
    \tilde\theta_{j}(R) = \min \left( \exp(-R), \frac{R \exp(-R)}{d_j } , \frac{\exp(-1)}{d_j \beta_j^{R-1} R!} \right).
\end{align*}
Since $\tilde\theta_{j}(R)$ may be very small for nodes with large out-degrees, we also use a global threshold to remove elements judged to be dominated by numerical inaccuracies.
The final thresholds are thus given by $\theta_{j}(R)$
\begin{align*}
    \theta_{j}(R) = \max (\tilde\theta_{j}(R), 10^{-6}).
\end{align*}

Pseudocode for the compression of a reachability value batch into core embeddings is shown in Algorithm \ref{alg:embeddings}.
As input, the algorithm takes the reachability values batch and the number of ECF sample points $k_\phi$, which how to choose in Section \ref{sec:digw:finalize}.

\subsection{Enhancements to the Core Embeddings} \label{sec:enhancements}

The core \textsc{Digraphwave} embeddings, as described above, are able to capture several distinct local structure, as can be seen in Figure \ref{fig:illustrative_example}.
However, not all different structures can be captured.
The core embedding of a node $j$ only receives information about the nodes which can be reached from $j$, meaning that the structure of its in-neighbourhood may be lost.
A consequence is that all nodes with no out edges, e.g., nodes 8-12 in the example graph or any isolated node, will have the exact same core embeddings, regardless of the structure of their in-neighbourhood.

We call our solution to this issue the transposition enhancement.
This enhancement consists of a second set of core embeddings computed from the transposed graph, i.e., the graph with all edge directions reversed.
Mathematically, this is easily done by replacing $\lapnorm$ with is in-degree normalized counterpart in Section \ref{sec:digw:reachability_values_computation}, and programmatically, one simply inputs the transposed adjacency matrix to the core embedding extraction, Algorithm \ref{alg:digraphwave_core}.
The core embeddings for the normal and the transposed graph are then concatenated into a joint embedding space, visualized in Figure \ref{fig:digraphwave}.

Another observed issue with the core embeddings is that they may fail to capture different structures which share similar symmetries. 
This is the case for nodes 1 and 13 in the example graph, Figure \ref{fig:example_graph}, as well as nodes 7 and 19, which pairwise have the exact same reachability values, i.e., $\bm{\Psi}_{\edgeuv{1}{*}} = \bm{\Psi}_{\edgeuv{13}{*}}$ and $\bm{\Psi}_{\edgeuv{7}{*}} = \bm{\Psi}_{\edgeuv{19}{*}}$.

To resolve this issue, we take inspiration from \cite{henderson_its_2011}, namely the idea that the structural embedding of a node should be influenced by the structural embeddings of its neighbours.
That is, given the core embeddings for each node, possibly with the transposition enhancements, we calculate, for each node, the mean embedding among all the neighbours, both in and out, and concatenate this aggregation to the original embedding. 
Mathematically, if $\bm{\chi}_j$ is the core embedding vector for a node $j$, we create the aggregation embedding as
\begin{equation}
    \bm{\bar\chi}_j = \frac{1}{|\mathcal{N}(j)|}\sum_{k \in \mathcal{N}(j)} \bm{\chi}_k,
\end{equation}
and then create the enhanced embedding as
\begin{equation}
    \bm{\hat\chi}_j = [\bm{\chi}_j, \bm{\bar\chi}_j].
\end{equation}

We judge that the cases where the transposition and aggregation enhancements are insufficient to produce distinguishable embeddings for noes with distinct structural identities to be rare, since we have yet to encounter such a case empirically.
However, if one would encounter such a rare case, it is possible to repeat the aggregation enhancement iteratively until any symmetries are broken.
In practice, this would have to be combined with feature pruning, as done in \cite{henderson_its_2011}, to avoid the embedding dimension doubling for each iteration.
We leave exploration of such schemes to future work.

\subsection{Finalizing \textsc{Digraphwave} } \label{sec:digw:finalize}
\begin{wraptable}{r}{0.2\textwidth}
    \caption{Dimension values $k_{\tau}$ and $k_{\phi}$ obtained from a desired embedding dimension $k_{\text{emb}}$.}\label{tab:dim_vals}
    \vskip -0.2in
    \begin{tabular}{cccc}\\\toprule  
    $k_{\text{emb}}$ & $k_{\tau}$ & $k_{\phi}$ & $\tilde{k}_{\text{emb}}$\\\midrule
    32 & 1 & 4 & 32\\
    64 & 2 & 4 & 64\\ 
    128 & 2 & 8 & 128\\  
    256 & 3 & 10 & 240\\ 
    512 & 4 & 16 & 512\\ 
    \bottomrule
    \end{tabular}
\end{wraptable}
There are three input parameters left to address, $k_\tau$, $k_\phi$ and $n_{\text{batch}}$.
The batch size $n_{\text{batch}}$ does not affect the embeddings, only the time and space complexity, and can therefore be set based on the available computer memory.

To set $k_\tau$ and $k_\phi$, we assume that there is a desirable final embedding dimension $k_{\text{emb}}$.
Given $k_\tau$ and $k_\phi$, and two Boolean variables $f_{T}$ and $f_{A}$ indicating if the transposition and aggregation enhancements are used, the final \textsc{Digraphwave} embedding dimensionality is 
\begin{align*}
    \tilde{k}_{\text{emb}} &= k_f k_\tau  k_\phi, &&
    k_f = 2 \cdot 2^{f_{T}} \cdot 2^{f_{A}}.
\end{align*}
We set $k_\tau$ and $k_\phi$ so that $\tilde{k}_{\text{emb}}$ is as close as possible, but smaller than, the desired dimension $k_{\text{emb}}$.
Furthermore, we assume that $k_\phi$ should be larger than $k_\tau$ since in practice the range for $\tau$, i.e., $[1, R]$, will likely be small, e.g., $R=3$ and not many values for $\tau$ should therefore be needed. 
These assumptions are implemented as the heuristic formulae
\begin{align}
    \label{eq:digw_dims}
    k_\tau &= \lfloor (k_{\text{emb}} / k_f)^{1/3}  \rfloor, &&
    k_\phi = \lfloor k_{\text{emb}} /k_\tau  \rfloor.
\end{align}
In Table \ref{tab:dim_vals}, the resulting dimensions are shown for a few values of $k_{\text{emb}}$ using $f_{T} = 1$ and $f_{A} = 1$.
Future work will determine if a more principled approach can be used for setting $ k_\tau$ and $k_\phi $, but our experience is that this heuristic works well in practice.

With this, \textsc{Digraphwave} is fully specified. 
The pseudocode for \textsc{Digraphwave} is shown in Algorithms \ref{alg:expm}-\ref{alg:digraphwave}, and Algorithm \ref{alg:hyperparameters} shows how the hyperparameters are set.
Only $R$ and $k_{\text{emb}}$ are left to the practitioner to decide.
For our experiments, $R=3$ or $R=2$, and $k_{\text{emb}}=128$ have generally worked well.

\SetKwFunction{FFT}{FFT}%
\SetKwFunction{hstack}{hstack}%
\SetKwFunction{vstack}{vstack}%
\SetKwFunction{linspace}{linspace}%
\SetKwFunction{emptyinit}{empty}%
\SetKwFunction{expmmb}{expm\_batch}%
\begin{algorithm}
 \DontPrintSemicolon
\caption{Calculating a batch of reachability values }\label{alg:expm}
\Fn{\expmmb{$\lapnorm$, $B$, $\{ \tau_s \}_{s=0}^{k_\tau-1}$, $K=40$}}{
    
    $\mathbf{B} \gets $ Eq. \eqref{eq:b-batch} \Comment*[r]{Create batch matrix}
    $\mathbf{F} \gets \eye \mathbf{B}$\\
    
    \For{$s$ \KwTo \Range{$k_\tau$}}{
        $\{ a_{sk}(\tau) \}_{k=0}^{K} \gets \FFT \left( \left\{ \exp -\tau_s \left(1 +  e^{2 \pi i \frac{\kappa}{K}} \right) \right\}_{\kappa=0}^{K} \right)$ \Comment*[r]{Fast fourier transform}
        $\heatdistbatchb{*} \gets a_{0s} \mathbf{F}$
    }

    \For{$k$ \KwTo \Range{$1$, $K+1$}}{
        $\mathbf{F} \gets (\lapnorm - \eye) \mathbf{F}$\\
        \For{$s$ \KwTo \Range{$k_\tau$}}{
            $\heatdistbatchb{*} \gets \heatdistbatchb{*} + a_{sk} \mathbf{F}$ \\
        }
    }
    \Return $\{ \heatdistbatchb{*} \}_{s=0}^{k_\tau-1}$
}

\end{algorithm}

\SetKwFunction{normalise}{normalise}%
\SetKwFunction{numzeros}{numzeros}
\SetKwFunction{reachabilitytoembeddings}{reachability2embeddings}%
\begin{algorithm}
\caption{Compressing reachability values to core embeddings}\label{alg:embeddings}
\Fn{\reachabilitytoembeddings{$\{ \heatdistbatchb{*} \}_{s=0}^{k_\tau-1}$, $\{ t_k \}_{k=0}^{k_\phi -1}$}}{

    $\bm{\chi}_{B} \gets \emptyinit(\text{len}(B), 0 )$ \Comment*[r]{Initialise the core embeddings.}
    
    \For{$s$ \KwTo \Range{$k_\tau$}}{
        $\zeta_B \gets \numzeros(\heatdistbatchb{*}) / n$ \Comment*[r]{Fraction zeros for each node in batch}
        \tcc{The below sums are applied only to nonzero elements of $\heatdistbatchb{*}$}
        \For{$k$ \KwTo \Range{$k_\phi$}}{
            $\bm{\tilde\chi}_{B, 2k}(\tau_s) \gets \zeta_B + \sum_{j: \heatdistbatch{j} > 0} \cos( t_k  \heatdistbatch{j})$\\
            $\bm{\tilde\chi}_{B, 2k+1}(\tau_s)  \gets \sum_{j: \heatdistbatch{j} > 0} \sin( t_k  \heatdistbatch{j})$\\
        }
        $\bm{\chi}_{B} \gets \hstack( \bm{\chi}_{B} , \bm{\tilde\chi}_{B}(\tau_s) ) $\\
    
    }
    \Return $\bm{\chi}_{B}$
}
\end{algorithm}

\SetKwFunction{Digraphwave}{digraphwave}%
\SetKwFunction{Digraphwavecore}{digraphwave\_core}%
\SetKwFunction{threshold}{threshold}%

\begin{algorithm}
 \DontPrintSemicolon
\caption{\textsc{Digraphwave core}}\label{alg:digraphwave_core}
\Fn{\Digraphwavecore{$G$, $\{ \tau_s \}$, $\{ \theta_j \}$, $\{ t_k \}$, $n_{\text{batch}}$ }}{

    $\adj \gets \text{adj}(G)$ \Comment*[r]{Get the (weighted) adjacency matrix}
    $D_{jj} \gets \{ \sum_{i} A_{\edgeji} \}_j$ \Comment*[r]{Calculate out-degrees}
    
    $\lapnorm \gets \eyestarb - \adj \degDstarb^{-1}$\\
    $\bm{\chi} \gets \emptyinit(0, 2 k_\tau k_\phi )$ \Comment*[r]{Initialise the core embeddings.}
    
    \For{$j$ \KwTo \Range{$0$, $n$, $n_{\text{batch}}$}}{
        $B \gets \{j, j+1, \dots j+n_{\text{batch}} - 1\}$ \\
        $\{ \Psi_{*, B}(\tau_s) \}_{s=0}^{k_\tau-1} \gets$ \expmmb{$\lapnorm$, $B$, $\{ \tau_s \}_{s=0}^{k_\tau-1}$}\\ 
        $\{ \bar\Psi_{*, B}(\tau_s) \}_{s=0}^{k_\tau-1}\gets $ \threshold{$\{ \Psi_{*, B}(\tau_s) \}_{s=0}^{k_\tau-1}$, $\theta_j$}\\
        $\bm{\chi}_{B} \gets$ \reachabilitytoembeddings{$\{ \bar\Psi_{*, B}(\tau_s) \}_{s=0}^{k_\tau-1}$, $\{ t_k \}$} \\
        $\bm{\chi} \gets \vstack(\bm{\chi}, \bm{\chi}_{B}) $\\
    }
    \Return $\bm{\chi}$
}

\end{algorithm}

\begin{algorithm}
 \DontPrintSemicolon
\caption{\textsc{Digraphwave}}\label{alg:digraphwave}
\Fn{\Digraphwave{$G$, $\{ \tau_s \}$, $\{ \theta_{j} \}_j$,  $\{ t_k \}$, $n_{\text{batch}}$, $f_{T}$, $f_{A}$}}{
    $\bm{\chi} \gets$ \Digraphwavecore{$G$, $R$, $k_{\tau}$, $k_{\phi}$, $n_{\text{batch}}$}\\
    
    \tcc{Transposition enhancement}
    \If{$f_{T}$}{ 
        $\bm{\chi}^{(T)} \gets$ \Digraphwavecore{$G^T$, $R$, $k_{\tau}$, $k_{\phi}$, $n_{\text{batch}}$}\\
        $\bm{\chi} \gets \hstack(\bm{\chi}, \bm{\chi}^{(T)})$\\
    }
    \tcc{Aggregation enhancement}
    \If{$f_{A}$}{
        \For{$i$ \KwTo $V$}{
            $\bm{\chi}^{(A)}_i \gets \text{mean} ( \{\bm{\chi}_j \}_{j \in \mathcal{N}(i)} ) $\\
        }
        $\bm{\chi} \gets \hstack(\bm{\chi}, \bm{\chi}^{(A)})$\\
    }   
    \Return $\bm{\chi}$
}

\end{algorithm}

\SetKwFunction{Parameter}{set\_hyperparameters}%

\begin{algorithm}
 \DontPrintSemicolon
\caption{Setting hyperparameters}\label{alg:hyperparameters}
\Fn{\Parameter{$G$, $R$, $k_{\text{emb}}$, $f_{T}$, $f_{A}$}}{
    
    $\adj \gets \text{adj}(G)$ \Comment*[r]{Get the unweighted adjacency matrix}
    $\{ d_j \}_j \gets \{ \sum_{i} A_{\edgeji} \}_j$ \Comment*[r]{Calculate unweighted out-degrees}
    $\beta_j \gets $ Eq. \eqref{eq:betaj} \\

    $k_f \gets 2 \cdot 2^{f_{T}} \cdot 2^{f_{A}}$ \\
    $k_{\tau} \gets \lfloor (k_{\text{emb}} / k_f)^{1/3} \rfloor$\\
    $k_{\phi} \gets \lfloor (k_{\text{emb}} / k_f) / k_{\tau} \rfloor$\\
    
    $\{ \tau_s \} = \linspace(1, R, k_{\tau})$\\
    $\{ t_k \} = \linspace(\pi, k_\phi \pi, k_\phi)$ \\
    
    $\theta_{j} \gets \max \left( \min \left( \exp(-R), \frac{R \exp(-R)}{d_j } , \frac{\exp(-1)}{d_j \beta_i^{R-1} R!} \right), 10^{-6} \right)$\\
    
    $n_{\text{batch}} \gets$ Calculate based on available memory \\
    \Return $\{ \tau_s \}$, $\{ \theta_{j} \}_j$, $\{ t_k \}$, $n_{\text{batch}}$
}
\end{algorithm}

\section{Theoretical Analysis of Diffusion Timescales} \label{sec:theory_contribution}

In this section, the theoretical analysis used in Section \ref{sec:digraphwave} to determine appropriate values for $\taumin$ and $\taumax$ is produced.
Specifically, the analysis studies the speed at which the diffusion signal propagates away from the neighbourhood of the initialization node in a directed graph.
By considering one worst case scenario, in which the heat signal leaves the neighbourhood of the initialization node the fastest and thus lose the local structural information,  we find Theorem \ref{thm:main_theorem} which provides a lower bound on the total heat contained in the $R$-egonet of the initialization node at time $\tau$.
Several additional lemmas are also presented as part of the proof of Theorem \ref{thm:main_theorem}.
Lemma \ref{lemma:ss_analytical_expression} is particularly useful since it provides analytical expressions for the heat coefficients of all nodes in the used worst case graph.
Throughout this section, the heat diffusion terminology will be used to help concretize the theory.

\subsection{Preliminaries}

Throughout this section, the differential equation \eqref{eq:advection_diffusion_diff_eq} with a unit of heat initialized at a node $j$ in a  directed graph $G=(V,E)$ is studied: 
\begin{align}
    \dot{\ubold}(\tau) &= \frac{\text{d}\ubold}{\text{d} \tau} (\tau) = -\lapnorm \ubold (\tau) && 
    u_i(\tau=0)= b_i =
    \begin{cases}
        1 & \text{if $i=j$}\\
        0 & \text{otherwise.}
    \end{cases} \label{eq:heat_transport_diff_eq}
\end{align}
As before, $\tau \in \Rreal_{+}$ is the time variable, and $\Rreal_{+}$ is the set of real and nonnegative numbers.
The heat variable $\ubold(\tau) \in \Rreal_{+}^n$ is a compact notation for the reachability values $\heatdist$, which was used in Section \ref{sec:digraphwave}.
Furthermore, the dot notation $\dot{\ubold}$ is used for the derivative with respect to the diffusion time $\tau$.

We here consider the edge set $E$ as a variable over which functions can be maximized or minimized, and we use $\Alphab$ -- the normalized adjacency matrix, see \eqref{eq:alpha_values} -- to represent the variable edges.
That is, a positive value of $\alpha_{\edgeji} = [\Alphab]_{\edgeji} > 0$ implies $\edgeji \in E$, and conversely, $\alpha_{\edgeji} = 0$ implies $\edgeji \notin E$. 
Note that $\ubold$ is also a function of $\Alphab$, $\ubold(\tau, \Alphab) $, though we suppress dependence and use $\ubold(\tau)$ for brevity.

Moreover, we define the $(j, r)$-ball in $G$ as the set of nodes which can be reached with $r$ or fewer hops from $j$, and denote it mathematically as $\Ncore(j, r) = \{i | \spdist(j, i) \leq r\}$, where $\spdist(j, i)$ is the directed shortest path distance from $j$ to $i$. 
Conversely, we denote the set of nodes which cannot be reached with $r$-hops from $j$ as $\Nperiph(j, r) = V -  \Ncore(v, r)$, and refer this set as the $(j, r)$-periphery in $G$.
Furthermore, the notation $\mathcal{C}(j, r) = \{j | \spdist(j, j) = r\} = \Ncore(j, r) - \Ncore(j, r-1)$ is used for the set of nodes exactly $r$-hops away from node $i$.
These three definitions are visualized in Figure \ref{fig:core periphery}, where nodes belonging to the $(0, 2)$-ball are encircled and marked with green, while nodes in the $(0, 2)$-periphery reside outside the dashed circle and are marked with orange.
The nodes in $\mathcal{C}(0, 1)$ are labelled with $1$, the nodes $\mathcal{C}(0, 2)$ with $2$, etc.

Using these definitions, we further define $u_b (j, r, \tau)$ as the total heat contained in a $(j, r)$-ball as time $\tau$, and likewise $u_p (j, r, \tau)$ as the total heat contained in the $(j, r)$-periphery:
\begin{align*}
    u_b (j, r, \tau) = \sum_{i \in \Ncore(j, r)} u_i(\tau), &  && u_p (j, r, \tau)= \sum_{i \in \Nperiph(j, r)} u_i(\tau).
\end{align*}
It is clear from these definitions and the heat conservation in \eqref{eq:heat_conservation} that
\begin{equation} \label{eq:ball_plus_periphery}
    u_b (j, r, \tau) + u_p (j, r, \tau)= \sum_{i \in V} u_i(\tau) = 1.
\end{equation}

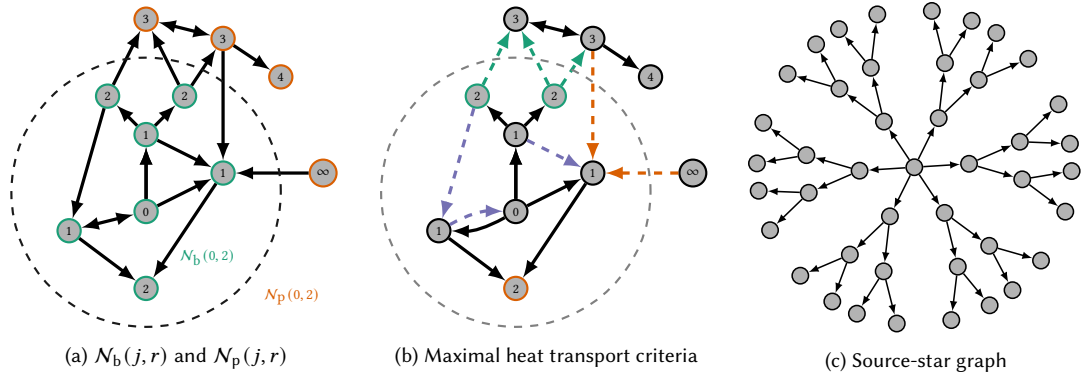
\begin{figure}[tp]
    \centering
    \begin{subfigure}{0.32\textwidth}
        \centering
        \resizebox{0.95\linewidth}{!}{
        \begin{tikzpicture}[
simplenode/.style={circle, draw, fill=black!30, thick},
c2cedge/.style={color=black, very thick},
c2pedge/.style={color=black, very thick},
p2cedge/.style={color=black, very thick},
p2pedge/.style={color=black, very thick},
neurtaledge/.style={color=black, very thick},
goodedge/.style={color=black, very thick},
badedge/.style={color=black, very thick},
]
\tiny
\def \ma {4}
\def \mb {6}
\def \mc {10}
\def \md {7}
\def \radiusa {1}
\def \radiusb {2}
\def \radiusc {5}
\node[simplenode, draw=cb1] (0) at (0,0) {0};
\node[simplenode, draw=cb1] (1) at (-1,-0.25) {1};
\node[simplenode, draw=cb1] (2) at (0,1) {1};
\node[simplenode, draw=cb1] (3) at (1,0.5) {1};

\node[simplenode, draw=cb1] (4) at (-0.5,1.5) {2};
\node[simplenode, draw=cb1] (5) at (0.5,1.5) {2};
\node[simplenode, draw=cb1] (6) at (0,-1) {2};

\node[simplenode, draw=cb2] (7) at (0,2.5) {3};
\node[simplenode, draw=cb2] (8) at (1,2.25) {3};
\node[simplenode, draw=cb2] (9) at (1.75,1.75) {4};
\node[simplenode, draw=cb2] (10) at (2.3,0.5) {$\infty$};

\draw[latex-latex, neurtaledge] (1) to (0);
\draw[-latex, neurtaledge] (0) -- (2);
\draw[-latex, neurtaledge] (0) -- (2);
\draw[-latex, neurtaledge] (0) -- (3);
\draw[-latex, neurtaledge] (2) -- (3);
\draw[-latex, neurtaledge] (1) -- (6);

\draw[-latex, neurtaledge] (2) -- (4);
\draw[-latex, neurtaledge] (2) -- (5);

\draw[-latex, goodedge] (4) -- (7);
\draw[-latex, goodedge] (5) -- (7);
\draw[-latex, goodedge] (5) -- (8);
\draw[-latex, badedge] (4) -- (1);
\draw[-latex, badedge] (8) -- (3);

\draw[latex-latex, neurtaledge] (7) -- (8);
\draw[-latex, neurtaledge] (8) -- (9);
\draw[-latex, badedge] (10) -- (3);

\draw[-latex, neurtaledge] (3) -- (6);

\draw[color=black!90, dashed, thick](0,0.25) circle (1.75);

\node[text=cb1] () at (0.8, -0.6) {$\Ncore(0, 2)$};
\node[text=cb2] () at (1.9, -1.1) {$\Nperiph(0, 2)$};

\end{tikzpicture}
        }
        \caption{$\Ncore(j, r)$ and $\Nperiph(j, r)$}\label{fig:core periphery}
        \Description{A graph with ten nodes labelled by their directed shortest path distance from node zero. Nodes at distance at most two are encompassed by a circle. The nodes within the circle are labelled as belonging to the core of node zero. The nodes outside the circle are labelled as belonging to the periphery of node zero.}
    \end{subfigure}
    ~
    \begin{subfigure}{0.32\textwidth}
        \centering
        \resizebox{0.95\linewidth}{!}{
        \begin{tikzpicture}[
simplenode/.style={circle, draw, fill=black!30, thick},
c2cedge/.style={color=black, very thick},
c2pedge/.style={color=cb1, very thick},
p2cedge/.style={color=cb3, very thick},
p2pedge/.style={color=black, very thick},
neurtaledge/.style={color=black, very thick},
goodedge/.style={color=cb1, very thick},
badedge/.style={color=cb2, very thick},
badedgeintra/.style={color=cb3, very thick},
]
\tiny
\def \ma {4}
\def \mb {6}
\def \mc {10}
\def \md {7}
\def \radiusa {1}
\def \radiusb {2}
\def \radiusc {5}
\node[simplenode] (0) at (0,0) {0};
\node[simplenode] (1) at (-1,-0.25) {1};
\node[simplenode] (2) at (0,1) {1};
\node[simplenode] (3) at (1,0.5) {1};

\node[simplenode, draw=cb1] (4) at (-0.5,1.5) {2};
\node[simplenode, draw=cb1] (5) at (0.5,1.5) {2};
\node[simplenode, draw=cb2] (6) at (0,-1) {2};

\node[simplenode] (7) at (0,2.5) {3};
\node[simplenode] (8) at (1,2.25) {3};
\node[simplenode] (9) at (1.75,1.75) {4};
\node[simplenode] (10) at (2.3,0.5) {$\infty$};

\path [-latex,badedgeintra,dashed,bend left=15] (1) edge node {} (0);
\path [-latex,neurtaledge,bend left=15] (0) edge node {} (1);
\draw[-latex, neurtaledge] (0) -- (2);
\draw[-latex, neurtaledge] (0) -- (2);
\draw[-latex, neurtaledge] (0) -- (3);
\draw[-latex, dashed, badedgeintra] (2) -- (3);
\draw[-latex, neurtaledge] (1) -- (6);

\draw[-latex, neurtaledge] (2) -- (4);
\draw[-latex, neurtaledge] (2) -- (5);

\draw[-latex, dashed, goodedge] (4) -- (7);
\draw[-latex, dashed, goodedge] (5) -- (7);
\draw[-latex, dashed, goodedge] (5) -- (8);
\draw[-latex, dashed, badedgeintra] (4) -- (1);
\draw[-latex, dashed, badedge] (8) -- (3);

\draw[latex-latex, neurtaledge] (7) -- (8);
\draw[-latex, neurtaledge] (8) -- (9);
\draw[-latex, dashed, badedge] (10) -- (3);

\draw[-latex, neurtaledge] (3) -- (6);

\draw[color=black!50, dashed, thick](0,0.25) circle (1.75);

\end{tikzpicture}
        }
        \caption{Maximal heat transport criteria}\label{fig:heat source maximal}
        \Description{A graph with ten nodes labelled by their directed shortest path distance from node zero. Nodes at distance at most two are encompassed by a circle. Three graph edges going from nodes at distance two are marked with green. Two edges are marked with orange; one going from a node at distance two to a node at distance one, the other doing from a node at distance three to a node at distance one. The one node at distance two without connections to nodes at distance three is marked with orange.}
    \end{subfigure}
    ~
    \begin{subfigure}{0.32\textwidth}
        \centering
        \resizebox{0.95\linewidth}{!}{
        \begin{tikzpicture}[
simplenode/.style={circle, draw, fill=black!30, thick, inner sep=3pt},
invisnode/.style={},
dashednode/.style={draw, circle, dashed, thick, fill=black!15, inner sep=3pt},
]

\def \mf {5}
\def \md {1}
\def \radius {1.0}
\def \margin {8}
\coordinate (O) at (0,0);
\node[simplenode] (0) at (0,0) {};
\pgfmathsetmacro\m{int(\mf + \md)};
\foreach \s in {1,...,\mf}
{
  \pgfmathsetmacro\i{floor((\s - 1)/ \m)};
  \coordinate (A) at ({360/\m * (2 + \i\m - \s)}:{\radius} );
  \coordinate (B0) at ({360/\m * (2 + \i\m - \s) - 10}:{\radius*1.98} );
  \coordinate (C00) at ({360/\m * (2 + \i\m - \s) - 5}:{\radius*2.87} );
  \coordinate (C01) at ({360/\m * (2 + \i\m - \s) - 20}:{\radius*2.87} );
  \coordinate (B1) at ({360/\m * (2 + \i\m - \s) + 10}:{\radius*1.98} );
  \coordinate (C10) at ({360/\m * (2 + \i\m - \s) + 20}:{\radius*2.87} );
  \coordinate (C11) at ({360/\m * (2 + \i\m - \s) +5}:{\radius*2.87} );

  \coordinate (C) at ({360/\m * (2 + \i\m - \s)}:{\radius*3} );
  \node[simplenode] (\s) at (A) {};
  \draw[-latex, thick] (0) -- (\s);
  \node[simplenode] (\s * \m) at (B0) {};
  \draw[-latex, thick] (\s) -- (\s * \m);
  
  \node[simplenode] (2 * \s * \m) at (C00) {};
  \draw[-latex, thick] (\s * \m) -- (2 * \s * \m);
  \node[simplenode] (2 * \s * \m) at (C01) {};
  \draw[-latex, thick] (\s * \m) -- (2 * \s * \m);
  
  \node[simplenode] (\s * \m) at (B1) {};
  \draw[-latex, thick] (\s) -- (\s * \m); 
  \node[simplenode] (2 * \s * \m) at (C10) {};
  \draw[-latex, thick] (\s * \m) -- (2 * \s * \m);
  \node[simplenode] (2 * \s * \m) at (C11) {};
  \draw[-latex, thick] (\s * \m) -- (2 * \s * \m);

}
\foreach \t in {1, ..., \md}
{
  \pgfmathsetmacro\s{int(\t + \mf)};
  \pgfmathsetmacro\i{floor((\s - 1)/ \m)};
  \coordinate (A) at ({360/\m * (2 + \i\m - \s)}:{\radius} );
  \coordinate (B0) at ({360/\m * (2 + \i\m - \s) - 10}:{\radius*1.98} );
  \coordinate (C00) at ({360/\m * (2 + \i\m - \s) - 5}:{\radius*2.87} );
  \coordinate (C01) at ({360/\m * (2 + \i\m - \s) - 20}:{\radius*2.87} );
  \coordinate (B1) at ({360/\m * (2 + \i\m - \s) + 10}:{\radius*1.98} );
  \coordinate (C10) at ({360/\m * (2 + \i\m - \s) + 20}:{\radius*2.87} );
  \coordinate (C11) at ({360/\m * (2 + \i\m - \s) +5}:{\radius*2.87} );
  
  \node[simplenode] (\s) at (A) {};
  \draw[-latex, thick] (0) -- (\s);
  \node[simplenode] (\s * \m) at (B0) {};
  \draw[-latex, thick] (\s) -- (\s * \m);
  
  \node[simplenode] (2 * \s * \m) at (C00) {};
  \draw[-latex, thick] (\s * \m) -- (2 * \s * \m);
  \node[simplenode] (2 * \s * \m) at (C01) {};
  \draw[-latex, thick] (\s * \m) -- (2 * \s * \m);
  
  \node[simplenode] (\s * \m) at (B1) {};
  \draw[-latex, thick] (\s) -- (\s * \m); 
  \node[simplenode] (2 * \s * \m) at (C10) {};
  \draw[-latex, thick] (\s * \m) -- (2 * \s * \m);
  \node[simplenode] (2 * \s * \m) at (C11) {};
  \draw[-latex, thick] (\s * \m) -- (2 * \s * \m);

}
    
\end{tikzpicture}
        }
        \caption{Source-star graph}\label{fig:source-star}
        \Description{A star graph  with all edges directed away from the centre node. Each branch of the star splits into three at every level away from the centre node.}
    \end{subfigure}
    \label{fig:theory_figures}
    \caption{
    \subref{fig:core periphery} Visualization of the 2-hop ball and periphery of node 0.
    \subref{fig:heat source maximal} Conceptualization of a graph fulfilling the first three criteria of Lemma \ref{lemma:heat_transport_maximal} for begin $(0,2)$-heat transport maximal. The three criteria are fulfilled if: 1) The green dashed edges crossing the boundary outwards exist. 2) Neither the dashed orange edges going from $\Nperiph(0,2)$ to $\Ncore(0,2)$ nor the dashed purple edges exist. 3) The bottom node marked with orange is connected to a node outside the dashed circular boundary.
    \subref{fig:source-star} Source-star graph with $d=6$, $\beta=2$ and $\ell=3$.
    }
\end{figure}

\subsection{Heat Containment Theorem}

The main result of the theoretical analysis is stated in the following theorem:
\begin{theorem} \label{thm:main_theorem}
For any directed graph $G=(V,E)$ with 
a heat diffusion process as defined by \eqref{eq:heat_transport_diff_eq},
the total heat in the $(j, R)$-ball $\Ncore(j, R)$ at time $\tau$ is greater or equal to $ Q(R+1, \tau)$, the normalized upper incomplete gamma function \cite{NIST_DLMF8.2.E4}:
\begin{align*}
    u_b (j, R, \tau) &= \sum_{i \in \Ncore(j, R)} u_i(\tau) \geq  Q(R+1, \tau) = \frac{\Gamma(R+1, \tau)}{\Gamma(R+1)}.
\end{align*}

\end{theorem}
\begin{proof}
In the trivial case where the out-degree of $j$ is $0$, $d_j = 0$, we have $u_b (j, R, \tau) = 1$ for all $R$, and since $1 \geq Q(R+1, \tau)$ the statement holds.

For the case $d_j \geq 1$, Lemma \ref{lemma:least_contained_heat} combined with Corollary \ref{col:50percent_heat} completes the proof.
\end{proof}

In other words, the normalized upper incomplete gamma function $Q(R+1, \tau)$ provides a lower bound for the amount of heat which can leave the $R$-hop neighbourhood of the heat initialization node. 
In Figure \ref{fig:heat_containment}, $Q(R+1, \tau)$ is shown as a function over $\tau$ for a set of $R$ values. 
The behaviour is as expected, the bound for the contained heat eventually reaches zero as $\tau$ grows and the proportion of contained heat for a given $\tau$ grows with $R$.

The rest of this section is devoted to completing the proof of Theorem \ref{thm:main_theorem}, in particular by establishing Lemma \ref{lemma:least_contained_heat} and \ref{col:50percent_heat}.
To summarize, this is done in three steps. 
First, three criteria on the graph edges for when the heat leaves the $(j, R)$-ball the fastest are established, Section \ref{sec:heat_transport_maximal}.
These criteria are then used to create the source-star graph for which the heat contained in $(j, R)$-ball is minimal, Section \ref{sec:source_star}.
Finally, we find that $Q(R+1, \tau)$ provides the lower bound by calculating the analytical expression of the matrix exponential for all nodes in  the source-star graph, and the proof of Theorem \ref{thm:main_theorem} is completed.

\subsection{Heat Transport Maximality} \label{sec:heat_transport_maximal}

\begin{definition}
Given a node $j \in V$ and a radius $R \in \{ 0, 1, 2, \dots \} $, we call a set of directed edges $\Alphab$, and the corresponding directed graph $G = (V, E)$, \emph{$(j, R)$-heat transport maximal} over the node set $V$ if the function 
\begin{align*}
    f(\Alphab, \tau) &= \int_0^\tau \frac{\mathrm{d}}{\mathrm{d} \tau} \left( u_p (j, R, s) - u_b (j, R, s) \right)  \text{d} s
\end{align*}
is maximal for every time $\tau$.
That is, $f(\Alphab, \tau) \geq f(\tilde\Alphab, \tau) $
for every $\tau$ and every set of edges $\tilde\Alphab$.
\end{definition}
Such a set of heat transport maximal edges $\Alphab$ is guaranteed to exist since the time derivatives are bounded, which follows from the differential equation \eqref{eq:heat_transport_diff_eq}, and the fact that the set of all $\Alphab$ for a given $V$ is closed.
Furthermore, it is easy to show that $(j,R$-heat transport maximality is equivalent to $u_b(j, R, \tau)$ being minimal:
\begin{lemma} \label{lemma:heat_transport_equivalence}
A directed graph is $(j, R)$-heat transport maximal iff $u_b(j, R, \tau)$ is minimal.
\end{lemma}
\begin{proof}
Using $u_b (j, R, 0)=1$ and $u_p (j, R, 0)=0$ and \eqref{eq:ball_plus_periphery},
\begin{align*}
    f(\Alphab, \tau) 
    &= \int_0^\tau  \dot{u}_p (j, r, s) - \dot{u}_b (j, r, s)   \text{d} s 
    = u_p (j, R, \tau) - u_b (j, R, \tau) + 1
    = 1 -  u_b (j, R, \tau)  - u_b (j, R, \tau) + 1
    = 2 -  2 u_b (j, R, \tau).
\end{align*}
So when $f(\Alphab, \tau)$ is maximal, $u_b (j, R, \tau)$ is minimal, and vice-versa.
\end{proof}

Intuitively, the set of edges which are heat transport maximal should all point away from the initialization node, as edges pointing towards it only serve to contain the heat.
This intuition is formalized through the three criteria in the following lemma, and visualized in Figure \ref{fig:heat source maximal}.
\begin{lemma} \label{lemma:heat_transport_maximal}
A graph for which the following three criteria are fulfilled for every $r \leq R$ is $(j, R)$-heat transport maximal.
\begin{enumerate}
    \item There exist no edges going from $\Nperiph(j, r)$ to $\Ncore(j, r)$.
    \item All nodes in $\mathcal{C}(j, r)$ have at least one outgoing edge.
    \item All outgoing edges from nodes in $\mathcal{C}(j, r)$ go to $\Nperiph(j, r)$.
\end{enumerate}
\end{lemma}
\begin{proof}
See Appendix \ref{app:proof_heat_transport_maximal}.
\end{proof}

\subsection{The Source-Star Graph} \label{sec:source_star}

The three criteria in Lemma \ref{lemma:heat_transport_equivalence} can be used to construct a set of graphs which are heat transport maximal for all radii for a specific node.
We these graphs \emph{source-star graphs}, and an example is visualized in Figure \ref{fig:source-star}. 
A source-star graph is defined by our nonnegative integer parameters, $\Gss(d, \beta, \ell, n_{\text{isolated}})$.
The parameter $d \geq 1$ is the out-degree of the centre node, which has node index $0$.
Each of these $d$ mutually disconnected neighbour nodes is the root node of a $\beta$-tree, where $\beta \geq 1$ is the tree branch factor. 
The leaf nodes of each branch lie $\ell \geq 0$ steps away from the centre node.
The final parameter $n_{\text{isolated}} \geq 0$ specifies a number of isolated nodes which can be included in the graph.
These are needed merely for proof technical reasons, and no isolated nodes are used unless otherwise specified, i.e., $\Gss(d, \beta, \ell) = \Gss(d, \beta, \ell, 0)$.
All the edges of $\Gss$ are directed away from the centre node, and thus the three criteria in Lemma \ref{lemma:heat_transport_maximal} are easy to verify for $j=0$ and each $R \leq \ell - 1$.

The number of nodes in $\Gss(d, \beta, \ell, n_{\text{isolated}})$ is $n = n_{\text{isolated}} + 1 + d \sum_{k=0}^{\ell-1} \beta^k$, and $| \Ncore(0, R) | = 1 + d \sum_{k=0}^{R-1} \beta^k$.
Moreover, the non-isolated nodes of $\Gss$ can be indexed for $i > 0$ as 
\begin{align}
    i = \nu_l +  d \sum_{k=0}^{l-2} \beta^k, \label{eq:ss_index}
\end{align}
where $\nu_l \in \{1, \dots, d \beta^{l-1} \}$ is an index in each layer $\mathcal{C}(0, l)$.
The inversion formula for \eqref{eq:ss_index}, i.e., going from $i$ to $l$ and $\nu_l$, can be found in the Appendix \ref{app:source_star} together with a block-matrix representation of $\lapnorm$ for $\Gss(d, \beta, \ell)$.

The next question to answer is how the heat diffusion on a source-star graph can bed compared with the diffusion on another graph.
The answer is provided in this next lemma.
\begin{lemma} \label{lemma:least_contained_heat}
For any directed graph $G=(V, E)$ and node $j \in V$, with out-degree $d_j \geq 1$ and with $\ell$ being the distance to the furthest away, but reachable, node from $j$, i.e.\ $\ell = \max_{k \in V} \{\spdist(j, k) | \spdist(j, k) < \infty \}$, there exists a source-star graph $\Gss(d_j, \beta, \ell, n_{\text{isolated}}) = (V^{(\text{ss})}, E^{(\text{ss})})$, with $\Ncore^{(\text{ss})}(0, R) \subseteq V^{(\text{ss})}$ such that $|\Ncore^{(\text{ss})}(0, R)| \geq |\Ncore(j, R)|$ and $u_b (j, R, \tau)  \geq u_b^{(\text{ss})} (0, R, \tau)$, i.e.\
\begin{align*}
    \sum_{i \in \Ncore(j, R)} u_i(\tau) \geq  \sum_{i \in \Ncore^{(\text{ss})}(0, R)} u_i(\tau),
\end{align*}
for any $R \leq \ell - 1$ and all $\tau \in \Rreal_{+}$.
\end{lemma}
\begin{proof}
See Appendix \ref{app:proof_least_contained_heat}.
\end{proof}

This lemma completes the first part of the proof for Theorem \ref{thm:main_theorem}. i.e., the source-star graph provides a lower bound for the heat contained in the neighbourhood of the initialization node for all graphs.
It is left to show that this bound is given by $Q(R+1, \tau)$.
To do so, we expand the first column of the matrix exponential to express the heat coefficients for all nodes in a source-star graph when the heat is initialized at the centre node:
\begin{lemma} \label{lemma:ss_analytical_expression}
Given a source-star graph $\Gss(d, \beta, \ell)$ with out-degree normalized Laplacian $\lapnormstar$ as in \eqref{eq:laplace:source-star} and a diffusion process as defined by \eqref{eq:heat_transport_diff_eq} with $j=0$, we have
\begin{equation}
    u^{(ss)}_i(\tau) =
    [\exp (-\tau \lapnormstar ) \mathbf{b}]_i = 
    \begin{cases}
        \exp(-\tau) & \text{if $i = 0$} \\
        \displaystyle \frac{d^{-1}}{\beta^{l-1}}\frac{ \tau^l}{l!} \exp(- \tau) & \text{if $i \in \{1, \dots, n - d \beta^{\ell -1}  - 1 \}$} \\
        \displaystyle  \frac{d^{-1}}{\beta^{\ell-1}} P(\ell, \tau) & \text{if $i \in \{n - d \beta^{\ell -1}, \dots, n - 1\} $},
    \end{cases}
\end{equation}

where $P(\ell, \tau)$ is the normalized version of the lower incomplete gamma function \cite{NIST_DLMF8.2.E4}.

\end{lemma}
\begin{proof}
See Appendix \ref{app:ss_analytical_expression}.
\end{proof}

Using Lemma \ref{lemma:ss_analytical_expression}, it is straightforward to show the following corollary, which completes the proof of Theorem \ref{thm:main_theorem}.
\begin{corollary}\label{col:50percent_heat}
 For any source-star graph $G_{\text{source-star}}(d, \beta, \ell)$, at time $\tau$, the total heat within in the $(0,R)$-ball $\Ncore(0, R)$, where $R < \ell$, is $\Gamma(R+1, \tau) / \Gamma(R+1)$, i.e.\ 
 \begin{equation}
  u_b^{(\text{ss})} (\tau, R, 0) = \sum_{i \in \Ncore^{(\text{ss})}(0, R)} u_i(\tau) = 
     \frac{\Gamma(R+1, \tau)}{\Gamma(R+1)}.
\end{equation}
\end{corollary}
\begin{proof}
See Appendix \ref{app:ss_analytical_expression}.
\end{proof}

\section{Related Works}\label{sec:related_works}

\subsection{Relation to \textsc{Graphwave}} \label{sec:related_works_gw}

This work is markedly inspired by the \textsc{Graphwave} algorithm for extracting structural node embeddings \cite{donnat_learning_2018}, and we indeed recognize this by naming our algorithm \textsc{Digraphwave}.
Nevertheless, there are considerable differences between \textsc{Digraphwave} and \textsc{Graphwave} which motivates as an independent work.

Foremost, \textsc{Digraphwave} is applicable to directed graphs, unlike \textsc{Graphwave}.
\textsc{Graphwave} is formulated using the eigenvalue decomposition of the Laplacian, which, as explained in Section \ref{sec:numerical_computation_of_mat_exp}, does not generalize to directed graphs.
We instead use the diffusion equation \eqref{eq:advection_diffusion_diff_eq} and its matrix exponential solution to formulate \textsc{Digraphwave}, which directly generalize to directed graphs.
Practically, \textsc{Graphwave} uses Chebyshev series approximation to compute its reachability values, which, again, does not generalize to directed graphs, see Section \ref{sec:numerical_computation_of_mat_exp}. 
As detailed in Section \ref{sec:digw:reachability_values_computation}, \textsc{Digraphwave} uses Taylor series approximation with a provable error bounds for any directed and weighted graph.

Furthermore, disparate methods for setting the hyperparameters are used. 
While the intuition for setting timescales is shared between this work and \cite{donnat_learning_2018}, i.e., $\taumin$ should be large for the diffusion signal to differ from the initialization, and $\taumax$ small enough to avoid reaching an uninformative stationary distribution, the proposed expressions are not similar.
Through theoretical motivation, \textsc{Graphwave} proposes $\taumin = -\log{0.95} / \sqrt{\lambda_1 \lambda_{n-1}}$ and $\taumax = -\log{0.85} / \sqrt{\lambda_1 \lambda_{n-1}}$, where $\lambda_1$ is the smallest non-zero eigenvalue of the normalized Laplacian and $\lambda_{n-1}$ is its largest eigenvalue.
In practice, however, the official \textsc{Graphwave} implementation uses the approximations $\lambda_1 = 1/n$ and $\lambda_{n-1} = 2$, likely due to the eigenvalues being expensive to compute.
However, with these approximations, $\taumin$ and $\taumax$ grow with $\sqrt{n}$, resulting in $\taumin \approx 3.6$ and $\taumax \approx 11.5$ for a graph with 10k nodes, and $\taumin \approx 36$ and $\taumax \approx 115$ for a graph with 1M nodes.
At such large timescales, not only will the heat diffusion have reached an uninformative stationary distribution, but it is also likely that the polynomial approximation of the matrix exponential has large errors, see Figure \ref{fig:taylor_error_bound}. 
For \textsc{Digraphwave}, $\taumin=1$ and $\taumax=R$ are used, where $R$, in our experiments, is no larger than $3$ even for graphs with $100k$ nodes.

Another minor hyperparameter difference is the evaluation of the empirical characteristic function, with sample points $\texttt{linspace}(0, 100, k_\phi)$ being used for \textsc{Graphwave} and $\texttt{linspace}(\pi, \pi k_\phi, k_\phi)$ for \textsc{Digraphwave}, as motivated in Section \ref{sec:digw:compress}.

In this work, we also present the transposition and aggregation embedding enhancements, Section \ref{sec:enhancements}.
The transposition enhancement addresses an issue specific for directed graphs while the aggregation enhancement is equally useful for undirected graphs, meaning that \textsc{Graphwave} would also benefit from it.

Finally, with the use of batching and parallelism, our implementation of \textsc{Digraphwave} scales to much larger graphs than the official implementation of \textsc{Graphwave}, which runs into memory bottlenecks even for medium-sized graphs, see Figure \ref{fig:scale:results}.

\subsection{Other Structural Embedding Extraction Methods for Directed Graphs} \label{sec:related_works_other}


The \textsc{ReFeX} algorithm \cite{henderson_its_2011} is a seminal work on structural node embedding extraction, and which is applicable to both directed and weighted graphs.
In the algorithm, each node is initialized with a set of handcrafted features, i.e., node degrees and counts of the incoming/outgoing edges to each node's egonet, which are then enhanced through recursive concatenation of aggregations of the features of each node's neighbours. 
The resulting exponential growth in embedding dimension is remedied by feature pruning.
In \cite{henderson_its_2011}, the proposed feature pruning is a combination of logarithmic binning and correlation measurements.
However, no official implementation of \textsc{ReFeX} is available online, and various custom implementations all use different versions of the pruning.
In our experiments, our own implementation of \textsc{ReFeX}, which performs pruning via QR-factorization \cite{refex_reimplementation22}, is used as a baseline. 
The main hyperparameter of this implementation is the pruning tolerance. 


\textsc{EMBER} is a recently proposed method for extracting structural embeddings \cite{ember_19}, and an extension to the \textsc{xNetMF} algorithm \cite{regal_xnetmf_18} for applicability to large and directed graphs.
\textsc{xNetMF} produces embedding in two steps.
First, structural signatures are extracted as sets of weighted node degrees for in all nodes in the $K$-egonet each node.
These sets are made fixed length using logarithmic binning. 
In the second step, these signatures are compressed to embeddings using matrix factorization.
In \cite{ember_19}, \textsc{EMBER} is applied for node classification on large email datasets, and in \cite{regal_xnetmf_18} \textsc{xNetMF} is used for network alignment.

Though both \textsc{xNetMF} and \textsc{Digraphwave} consist of signature extraction and compression, each of these steps are significantly different.
In  \textsc{xNetMF}, the neighbourhood used to extract the signatures is defined by the hyperparameter $K$. 
Unlike $R$ in \textsc{Digraphwave}, $K$ is a hard cut-off, meaning that no nodes further away than $K$ are considered for the signature. 
In practice, $K=2$ is used \cite{regal_xnetmf_18, ember_19}, meaning that only small neighbourhoods are considered.
In \textsc{Digraphwave} $R$ is a soft constrained, meaning that even if $R=2$ is used, all reachability values in the entire network are contained in the signature.
The number of terms of the Taylor series approximation, indecently also $K$, does constitute a hard threshold for the neighbourhood size considered. However, we set it to $40$, which includes the entire graph for real-world graphs with small diameters \cite{watts_collective_1998}.
For the signature compression, \textsc{xNetMF} relies on matrix factorization, which is an iterative procedure learned per graph.
Consequently, projection matrices need to be stored to repeat the same embedding extraction on a new graph.
\textsc{Digraphwave} instead uses the empirical characteristic function, which is deterministic and thus does not require any information to be stored in order to repeat the embedding extraction on a new graph.

In \cite{furutani_graph_2020}, Furutani et al.\ propose to apply  \textsc{Graphwave} to the hermitian Laplacian of a graph -- a matrix with complex elements where the connection strengths and directions of edges are separately encoded using the magnitude and phase of the matrix elements -- to extract structural embeddings.
The advantage of this representation is that the eigenvalue decomposition is guaranteed to exist for the hermitian Laplacian, which enables theoretical analysis using tools from the signal processing literature.
In addition to the  \textsc{Graphwave} hyperparameters, this approach has one hyperparameter associated with the hermitian Laplacian matrix, $q$, which the authors set to $0.02$ for their experiments.

Both \cite{ember_19} and \cite{furutani_graph_2020} also applied \textsc{Graphwave} to directed graphs as a baseline in their experiments, but without any modifications to the algorithm.

\section{Experiments} \label{sec:experiments}

In this section, \textsc{Digraphwave} is evaluated in terms of embedding quality and scalability and compared to other structural node embedding methods applicable to directed graphs.
The baseline methods and a description of the general experiment settings is found in Section \ref{sec:baselines_and_compute}.
The scalability experiment is detailed, and the results are presented in Section \ref{sec:scalability}.

We evaluate quality of the \textsc{Digraphwave} embeddings both on synthetic and real graph datasets.
The synthetic dataset is an adopted and modified version of a previously suggested benchmark dataset for undirected graphs \cite{junchen_2021}, and it is described in Section \ref{sec:synthetic_dataset}.
Using the synthetic dataset, an ablation study of the transposition and aggregation enhancements is performed, Section \ref{sec:ablation}, and a comparison to other embedding methods, Section \ref{sec:synthetic_performance}.
On real graph datasets, the embeddings are evaluated using node classification, Section \ref{sec:enron_data}, and network alignment Section \ref{sec:network_alignment}.

\subsection{Specification of Embedding Methods and Computational Environment} \label{sec:baselines_and_compute}

We compare \textsc{Digraphwave} to the other structural embedding methods discussed in Section \ref{sec:related_works}: \textsc{ReFeX} \cite{henderson_its_2011, refex_reimplementation22}, \textsc{EMBER} \cite{ember_19}, \textsc{Graphwave} from 2018 \cite{donnat_learning_2018} and \textsc{Graphwave} using the Hermitian Laplacian by Furutani et al.\ \cite{furutani_graph_2020}.
Additionally, a trivial baseline is included where the embeddings consist of the in and out node degrees, both unweighted and weighted if edge weights are available.

The embedding dimensions are set to be as similar as possible for each method.
For the scalability experiments 64 dimensions are used, while 128 dimensions are used otherwise.
\textsc{ReFeX} automatically determines the embedding dimension from the data, and only a maximum embedding dimension may be specified.
By setting the maximum dimension to twice that specified for the other methods, and the prune tolerance to $10^{-6}$, the automatically determined \textsc{ReFeX} dimensions were generally close to those of the other methods.

For \textsc{EMBER}, we use the default hyperparameters $K=2$, $\gamma=1$ and discount factor $0.1$, and likewise for \textsc{Graphwave}, the default hyperparameters described in Section \ref{sec:related_works_gw} are used with $k_\tau=2$.
Furutani et al.\ appear to use $3620$ dimensions for their embeddings \cite{furutani_graph_2020}, which we reduce to match the other methods by using fewer timescales while keeping the other hyperparameter values.
Importantly, once extracted, all embeddings are standardized, zero mean and unit standard deviation, before being applied in the experiment tasks.
Forgoing standardization may result in significantly worse performance in our experience.

All experiments were performed using Intel Xeon Gold 6230 CPUs @ 2.10 GHz, with 8 cores for the scalability experiment and synthetic data, and 16 cores for the real data.
For \textsc{Digraphwave}, the only method with GPU support, one Nvidia RTX 2080Ti 11 GB was used in the scalability experiments, and three for the experiments on real data.
The memory requirements are very different between the methods, and some experimentation was required to avoid out-of-memory issues.
In the end,  64 GB was used for the scalability experiments, 32 GB for the synthetic data experiments and 148 GB for the experiments using real data.
it is worth noting that for \textsc{Digraphwave}, 32 GB was enough to have memory left-to-sparse in all experiments.

\begin{figure}[htp]
    \centering
    \begin{subfigure}{0.49\textwidth}
    \centering
    \includegraphics[width=\linewidth]{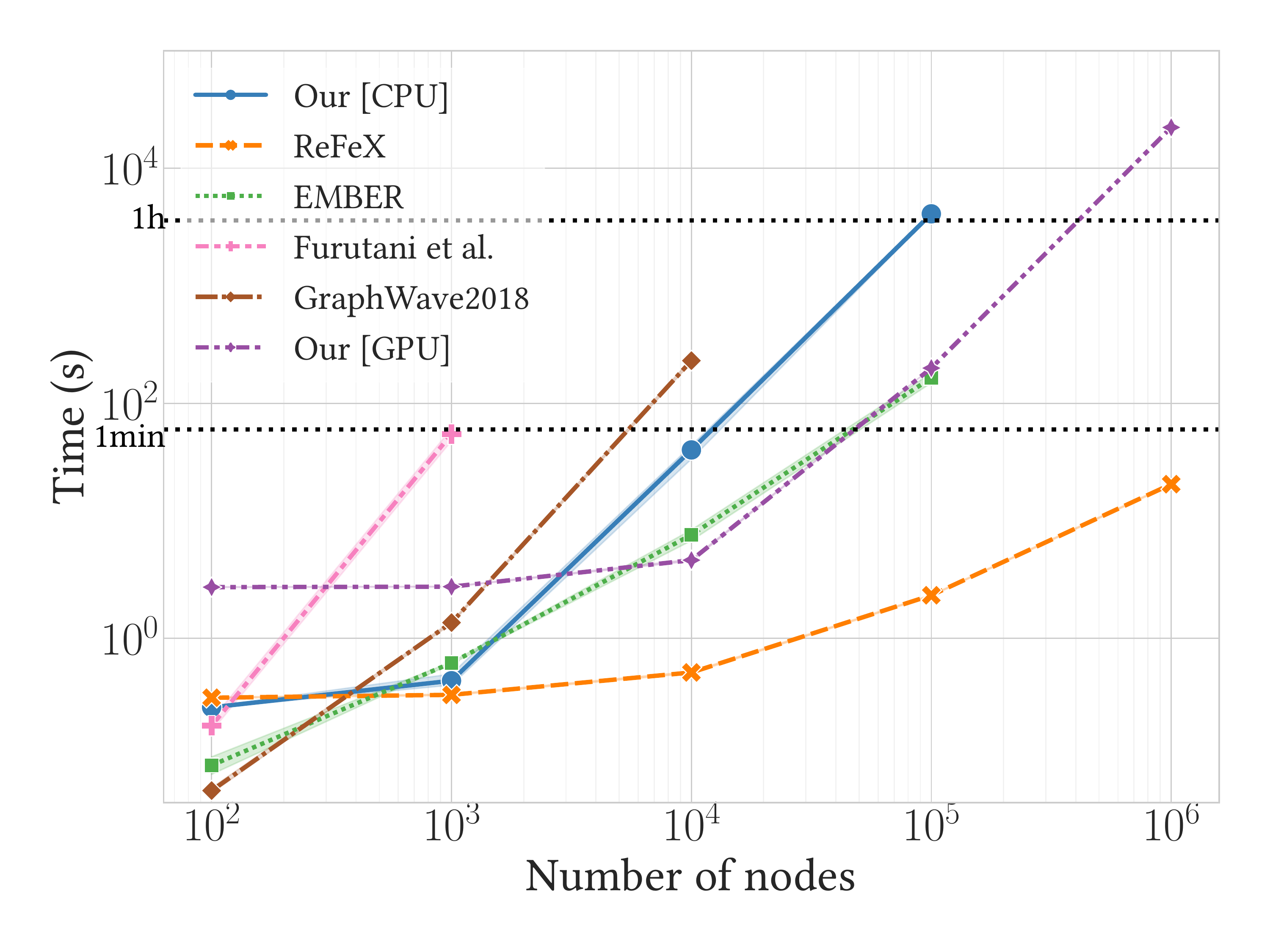}
    \caption{1 edge per node}
    \end{subfigure}
    ~
    \begin{subfigure}{0.49\textwidth}
    \centering
    \includegraphics[width=\linewidth]{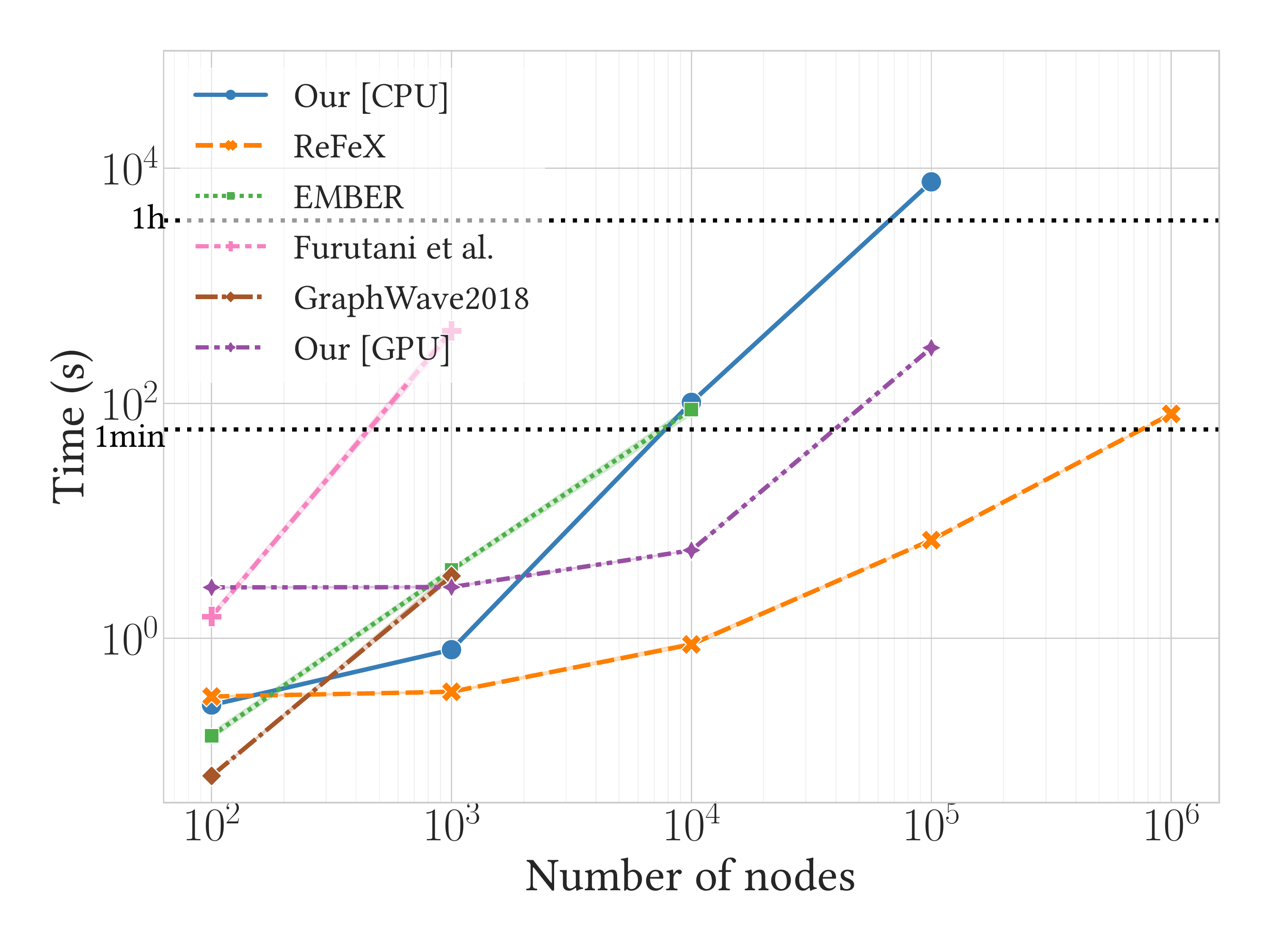}
    \caption{10 edges per node}
    \end{subfigure}
    \Description{%
    Two plots showing the scalability performance of six different embedding methods. The first plot shows the performance when one edge per node in the Barabási-Albert graph is used, the second for ten edges per node. The y axes show the mean computation time in seconds, and the x axes the number of nodes in the graphs.
    }
    \caption{%
    The scalability experiment results. The lines show the mean runtime in second vs number of nodes in the Barabási-Albert graph. Standard deviations are shown as shaded areas around each line, though these are hard to distinguish due to being insignificantly small.%
    }
    \label{fig:scale:results}
\end{figure}

\subsection{Scalability} \label{sec:scalability}

To measure the scalability, random graphs of successively increasing size were generated and the embedding computation time for each graph was measured.
The Barabási-Albert preferential attachment model \cite{barabasi_graph} was used for the generation, chosen since it produces graphs exhibiting power law degree distributions, similar to many real graphs.
The experiment was performed using 1, 5 and 10 edges per node\footnote{More specifically, these numbers are the degrees of nodes added to the graph during the generation process.}, and for each graph size, 5 different graphs were sampled over which the mean and standard deviation were calculated.
For methods where an out-of-memory error was raised, or when the timeout limit of 8 hours was reached, were aborted.
\textsc{Digraphwave} used $R=3$ and both transposition and aggregation enhancements. 

In Figure \ref{fig:scale:results}, the mean computation time in seconds for 1 and 10 edges per node are shown on log-log axes, with standard deviations barely distinguishable  as shaded areas around the lines. See Appendix \ref{app:scalability} 5 edges per node.
\textsc{ReFeX} is by far the fastest method with close to linear time complexity.
The intercept observed for small graph sizes is caused by JIT compilation, and could be removed by instead using an implementation with ahead-of-time compilation.
Similarly, for \textsc{Digraphwave}, reported as “Our”, the intercept of the CPU implementation is caused by JIT compilation, while CUDA initialization overhead is the main cause when GPU is used. 

The performance of CPU-\textsc{Digraphwave} and \textsc{EMBER} are similar, and on large graphs both are outperformed by GPU-\textsc{Digraphwave}, which is the only method, in addition to \textsc{ReFeX}, able to complete the embedding extraction for a graph with 1M nodes.
The main bottleneck for \textsc{EMBER} is its memory consumption, with out-of-memory errors terminating the experiment after 100K nodes for 1 and 5 edges per node, and already after 10K nodes when 10 edges per node is used.
Neither \textsc{Graphwave} nor the method by Furutani et al.\ are scalable with their current implementations, and the experiments were terminated even for moderately sized graph.

\begin{figure}[htp]
\centering
\includegraphics[width=\linewidth]{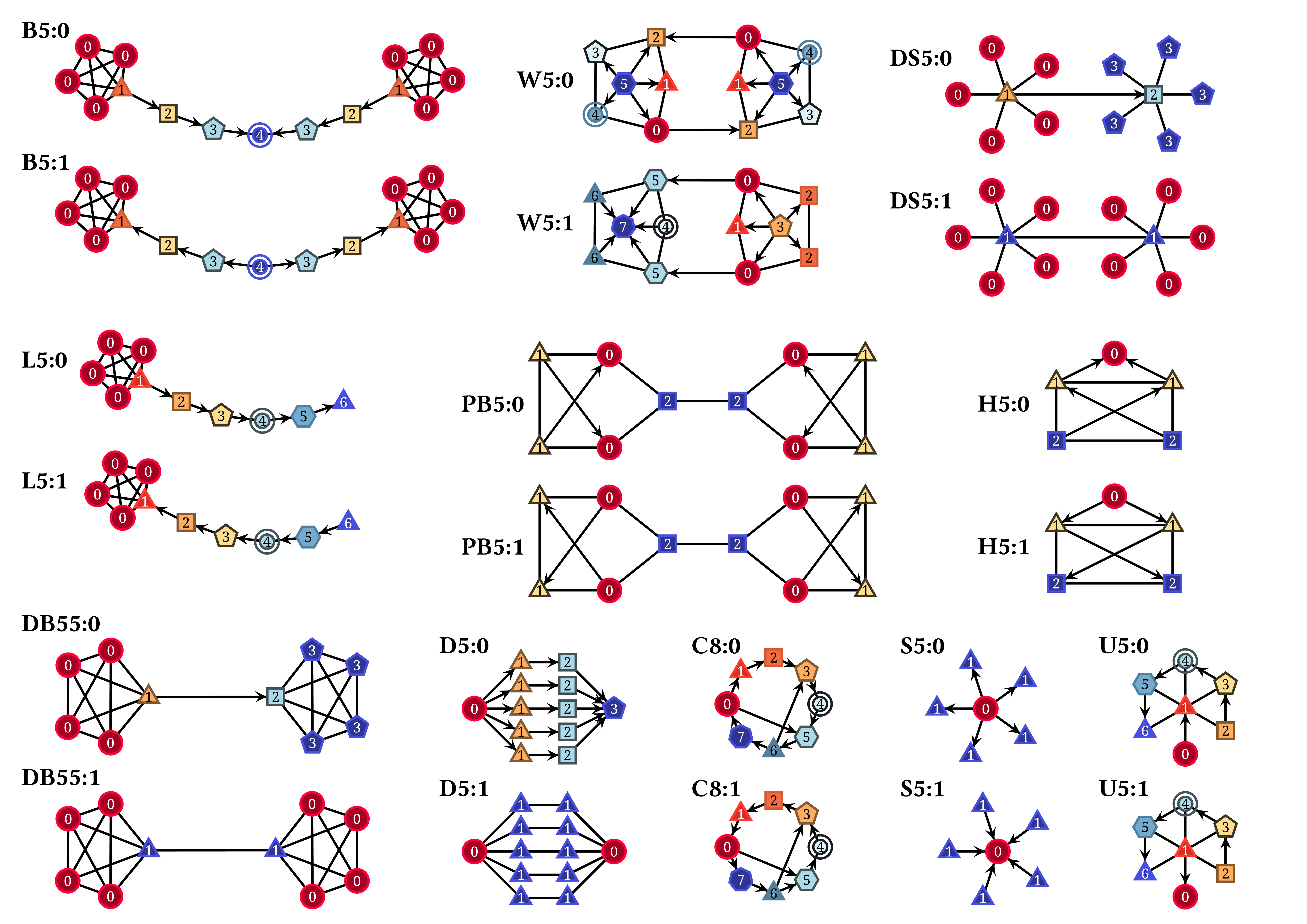}
\Description{%
22 different graphs arranged in a grid and labelled with names. All nodes are marked using numbers, colours and shapes to indicate their automorphic identitiy.%
}
\caption{%
Synthetic dataset of directed graphs based on the undirected dataset proposed in \cite{junchen_2021}. The numbers, colours, and shapes of the nodes show the automorphic identity of the node within its component. No automorphic identity is shared between disconnected components, with a total of 102 distinct identities present. Edges where both directions are present are shown as undirected edges.%
}
\label{fig:synth_dataset}
\end{figure}

\subsection{The Synthetic Dataset and Experiments Setup} \label{sec:synthetic_dataset}

In \cite{junchen_2021}, a synthetic dataset of undirected graphs with known automorphic equivalence identities was proposed for evaluation of structural node embedding methods, and we here adopt this dataset as well as their extrinsic evaluation experiment.
That is, the node equivalence identities are used as node labels to train and test classification and clustering models with the embeddings as input, and the classification and clustering scores are used as measures of the embedding quality.
Since our focus is directed graphs, we modify the dataset accordingly, manually defining two directed graphs for each undirected graphs in the original dataset.
Importantly, the directed graphs were defined so that no automorphic identity is \emph{between} the different graphs.
Moreover, we introduce a new graph named D5, while the larger composed graphs and the graphs for regular equivalence testing were excluded. 
The full set of directed graphs, along with their names, is shown in Figure \ref{fig:synth_dataset}.
The colour, shape, and number of each node indicates it automorphic identity in the context of its component. 
Note again that no automorphic identity is shared between the different graphs.

We also modify experiment setup used in \cite{junchen_2021}.
Instead of only treating the synthetic dataset graphs as disconnected components, we also evaluate the embeddings' robustness to noise in the local structures by randomly connecting them to a common circular graph.
The number of connecting noise edges per graph is gradually increased from zero to five.
An example using three connecting edges for the C8 and W5 graphs is shown in Figure \ref{fig:synth_dataset_composed}.
In the experiments, all graphs are included and repeated 10 times, each connected with different edges to the common circular graph. 
For the classification evaluation, five of these are included in the train set and the other five in the test set.
The number of nodes in the circular graph is scaled so that each inserted graph receives its own segment to which it is connected.
In Figure \ref{fig:synth_dataset_composed}, the C8 and W5 graphs are only repeated two times for visual clarity.
Moreover, for each number of connecting edges, five different composed graphs are created to measure the means and standard deviations of the classification and clustering performance.
\begin{wrapfigure}{r}{0.4\textwidth}
  \centering
    \includegraphics[width=1\linewidth]{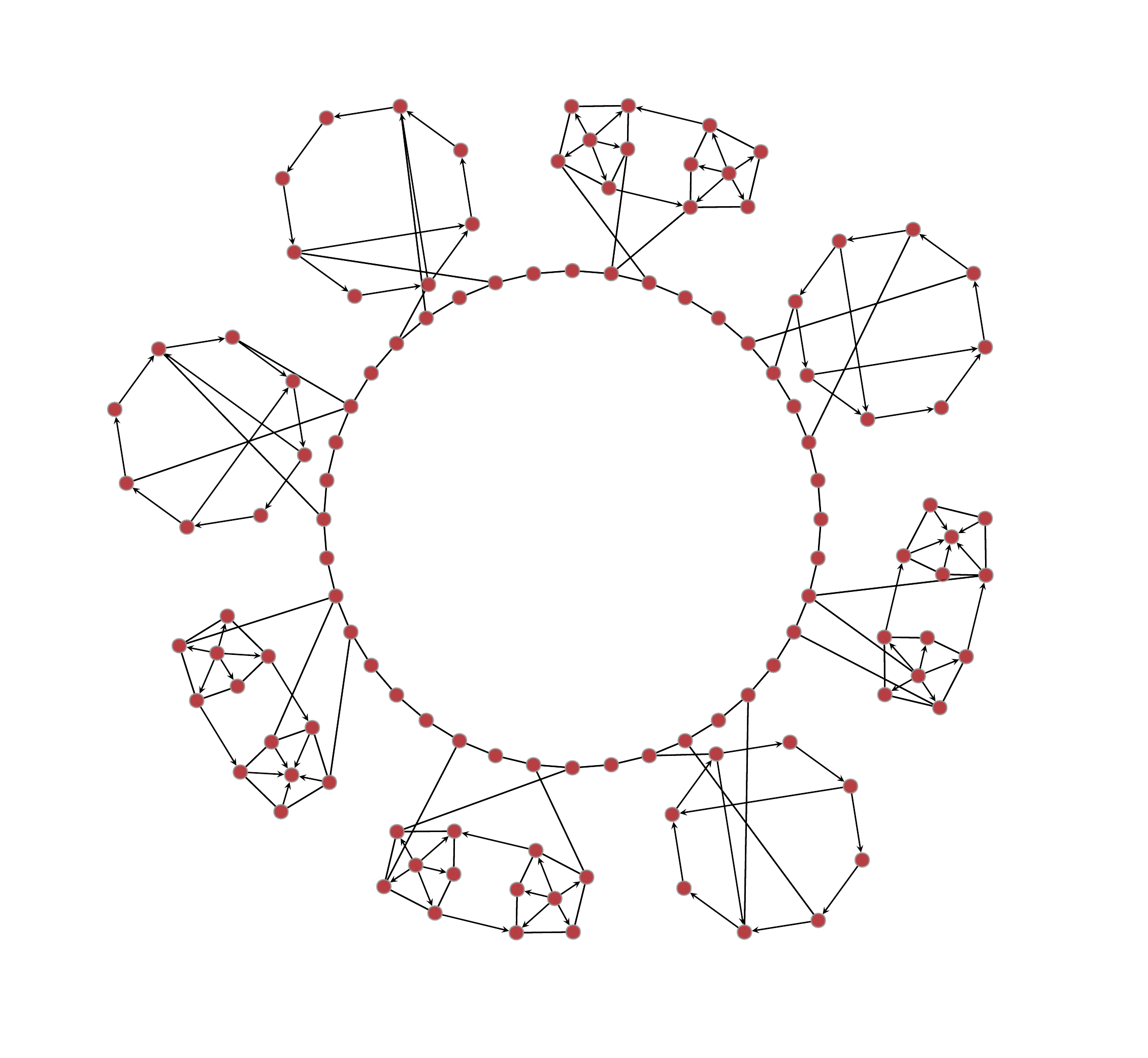}
    \Description{A circualar graph surrounded by eight subgraphs, each conencted with three edges to the circular graph.}
    \caption{A graph composed using the C8 and W5 synthetic dataset graphs repeated two times. Each dataset graph is connected to separate segments of the circular graph with three randomly sampled edges.}
    \label{fig:synth_dataset_composed}
    \vspace{-20pt}
\end{wrapfigure}
Each of these graphs consists of 3240 nodes, out of which 1100 are in the circular graph and not included in the classification or clustering, with 102 distinct automorphic identity classes. 
The whole dataset is available through the project page\footnote{\url{https://ciwanceylan.github.io/digraphwave-project/}}.

Macro F1 score is used to measure the classification performance, and adjusted mutual information for the clustering.
The classification model used for all embedding methods is a multinomial logistic regression model with $\ell_2$ penalty $C=1.0$ and class weighted loss function, and $k$-means with 102 clusters, i.e., the number of  ground truth automorphic identities, is used for the clustering, which is the same models as used in \cite{junchen_2021}.

\subsection{Ablation Study} \label{sec:ablation}

Our first usage of the synthetic dataset is an ablation study of the two enhancements, transposition (T) and aggregation (A).
We set $R=3$ and $k_{\text{emb}}=128$, meaning that $k_\tau=4$ and $k_\phi = 16$ when no enhancements were applied, $k_\tau=3$ and $k_\phi = 10$ when one enhancement was applied, and $k_\tau=2$ and $k_\phi = 8$ when both were applied.
The results are shown in Figure \ref{fig:synth:ablation}.
For all noise levels, the aggregation enhancement provides large classification performance gains, $10$-$20$\% higher F1, with even higher gains from the transposition enhancement, $20$-$30$ higher F1.
Using both yields the best results.
The result is similar for the clustering, but with the performance gain from aggregation being even smaller than that from transposition.
Regardless, it is easy to conclude that the benefits of using the two enhancements outweigh any detriments from using lower values of  $k_\tau$ and $k_\phi$.
Consequently, both enhancements are used for all later experiments. 

By analysing the confusion matrices from the classification experiment in the noise free case, we see that there is a large overlap between the two enhancements for which automorphic identities the classification accuracy improves, e.g., both improve the distinguishability of identity 3 in B5:1 and identity 3 in L5:1.
We also find that only two pairs of automorphism identities require both transposition and aggregation to be distinguishable: identity 0 in B5:0 and identity 0 in DB55:0, and identity 0 in B5:1 and identity 3 in DB55:0.

As a point of reference, the result where the graphs are treated as undirected is also included in Figure \ref{fig:synth:ablation}.
Some automorphic identities are distinguishable without consideration of the edges directions, but overall performance is poor, as expected.
The small increase in classification performance after the addition of a few noise edges is a consequence of how the logistic regression model handles overlapping data points with disparate class labels.

\begin{figure}[tp]
    \centering
    \begin{subfigure}{0.49\textwidth}
    \centering
    \includegraphics[width=\linewidth]{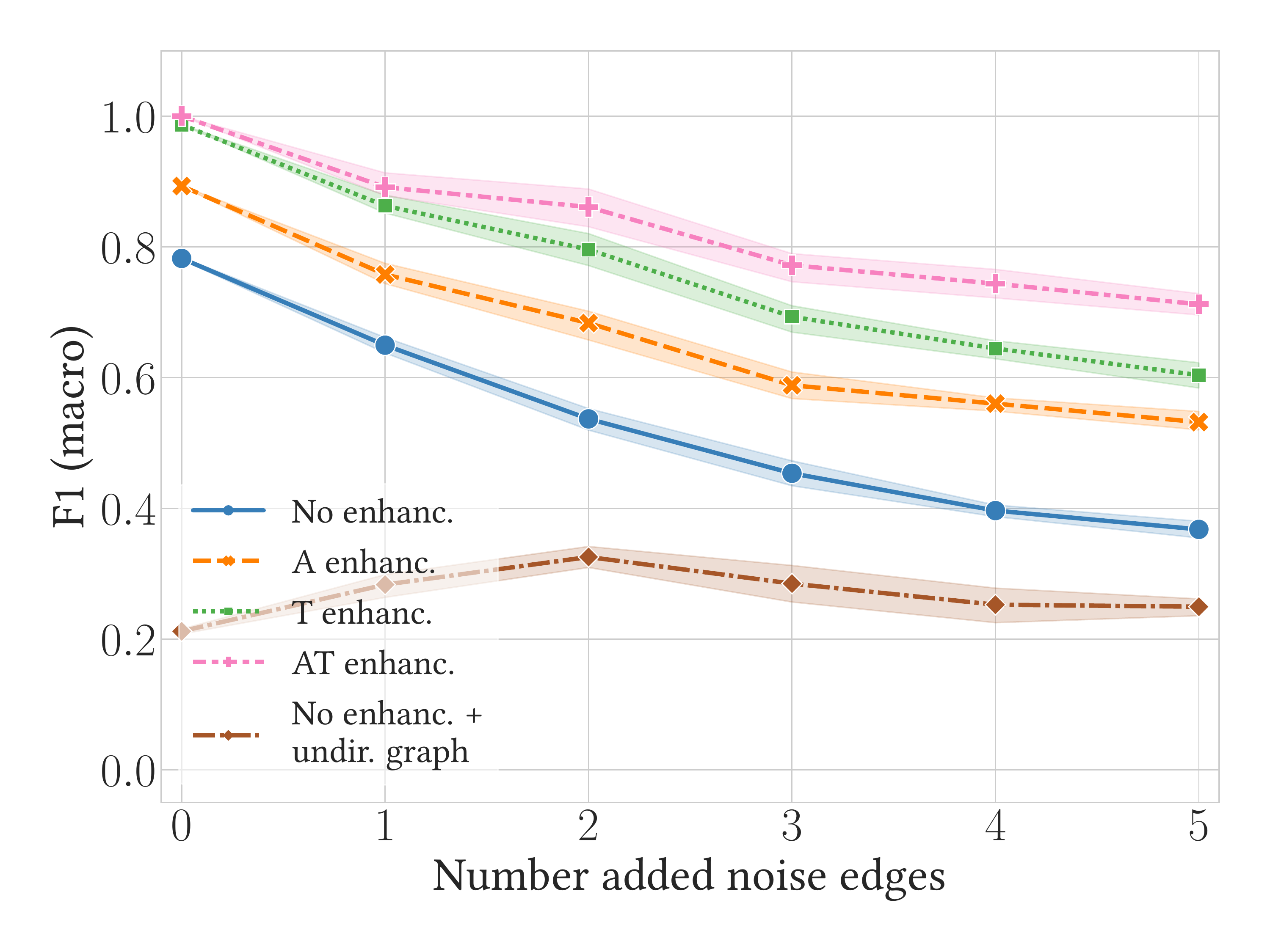}
    \caption{Classification}
    \end{subfigure}
    ~
    \begin{subfigure}{0.49\textwidth}
    \centering
    \includegraphics[width=\linewidth]{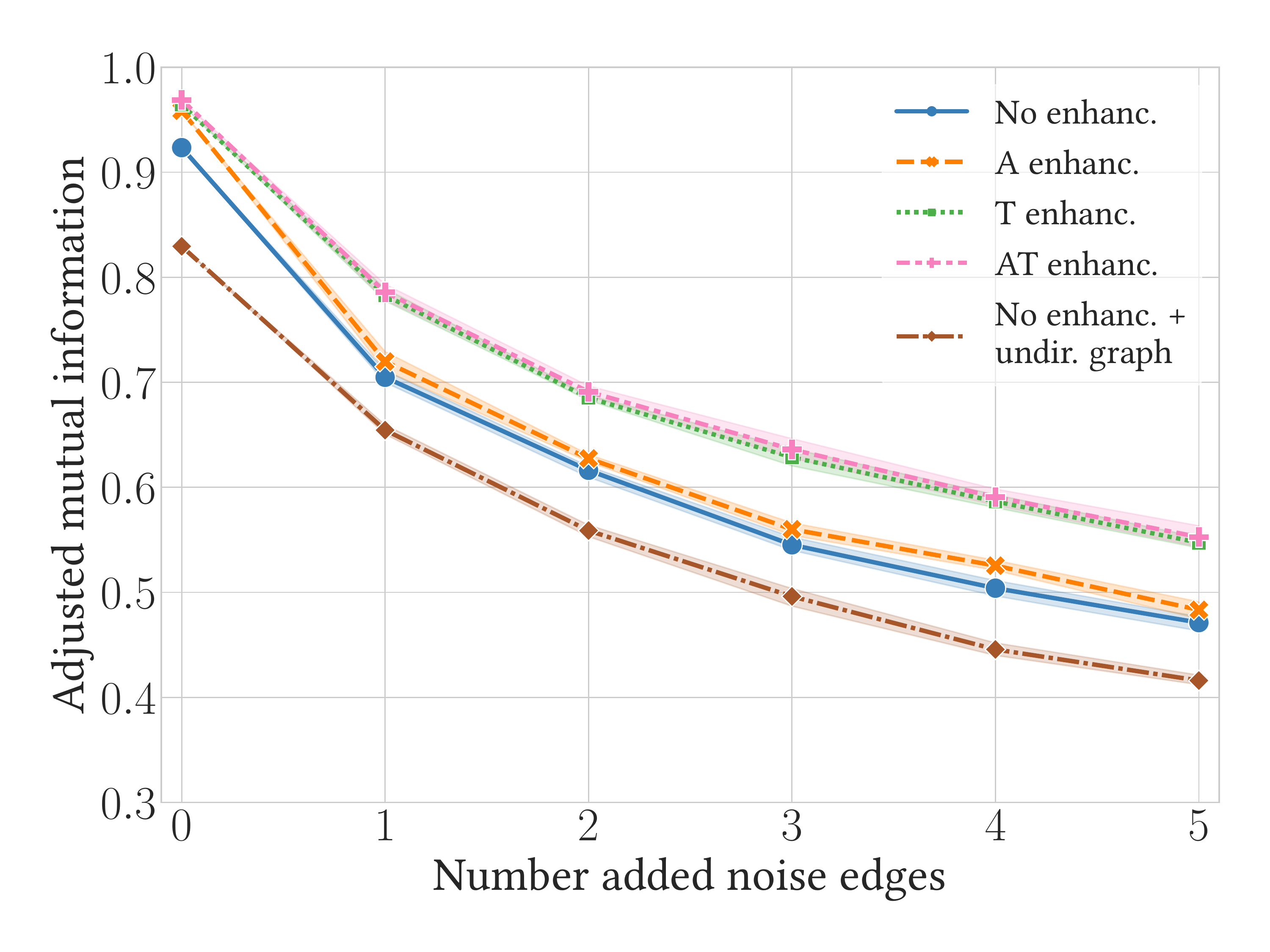}
    \caption{Clustering}
    \end{subfigure}
    \Description{Two plots side-by-side showing the results of the ablation study on the synthetic dataset. The left plot shows the classification performance,and the right plot the clustering performance, of Digraphwave with and without enhancements, and also without considering edge directions.}
    \caption{Results of the ablation study of the transposition (T) and aggregation (A) enhancements on the synthetic dataset.}
    \label{fig:synth:ablation}
\end{figure}

\subsection{Embedding Quality Comparison} \label{sec:synthetic_performance}

In the second experiment on the synthetic dataset, the quality of different structural node embedding algorithms are compared with \textsc{Digraphwave} using $R=3$ and both enhancements.
The results are shown in Figure \ref{fig:synth:alg_comp}.
When no noise edges are added, both \textsc{Digraphwave} and the method by Furutani et al.\ achieve perfect classification performance, meaning that all different automorphic identities can be represented by the embeddings, which, conversely, is not the case for the other methods, with 0.88 macro F1 for \textsc{ReFeX}, 0.76 for \textsc{EMBER} and 0.7 for \textsc{Graphwave}.
Once noise edges are added, \textsc{Digraphwave} performs similarly to \textsc{ReFeX}, whose classification performance is the most robust to the noise edges, while the performance of the method by Fututani et al.\ drops fast.

For noise levels 1-5, \textsc{Graphwave} performs similarly to \textsc{Digraphwave} without enhancements in the ablation study, see Figure \ref{fig:synth:ablation}, but receive $8$\% lower F1 score in the noise free case.
The similar performance is expected, since the two methods are, in this case, very similar.
The synthetic graphs are not large enough for the numerical imprecision resulting from the Chebyshev approximation in \textsc{Graphwave} to make a difference.
Thus, the main difference is the used timescales and ECF sample points: $\taumin \approx 2$, $\taumax \approx 6.5$,  $k_\tau=2$ and $k_\phi=32$ for \textsc{Graphwave}, and $\taumin = 1$, $\taumax = 3$, $k_\tau=4$ and $k_\phi=16$ for \textsc{Digraphwave}.
Since the range of the timescales quite similar, the difference in number of timescales used is likely the explanation for the observed performance discrepancy in the noise free case.

The most surprising results are \textsc{ReFeX}'s robustness to noise and the poor performance of \textsc{EMBER}.
Though both methods use aggregations of node degrees as their core features, classification performance only drops around 10\% for \textsc{ReFeX} while the performance of \textsc{EMBER} is reduced to close to trivial performance once the maximum amount of noise edges have been added.
Conversely, a similarity between \textsc{ReFeX} and \textsc{EMBER} is revealed by analysing the confusion matrices.
Both struggle to distinguish the identities in the C8 graphs, while these not presenting a challenge for the other embedding methods.
The reason is likely that \textsc{ReFeX} and \textsc{EMBER} focus on node degrees as signatures for their embeddings, which are homogenous between C8 graphs.
By instead extracting signatures via diffusion, the cyclic patterns of C8 graphs are revealed and can be used to distinguish the different nodes.

The results for the clustering are again similar in terms of method ranking.
All methods, except \textsc{Graphwave}, perform similar in the noise free case, and the performance of \textsc{EMBER} and the method by Furutani et al. drops fast as noise edges are added.
Here, \textsc{Digraphwave} is the most robust to the noise, and it is the only method which achieves a higher score than the trivial node degree baseline in the presence of all five noise edges.

\begin{figure}[tp]
    \centering
    \begin{subfigure}{0.49\textwidth}
    \centering
    \includegraphics[width=\linewidth]{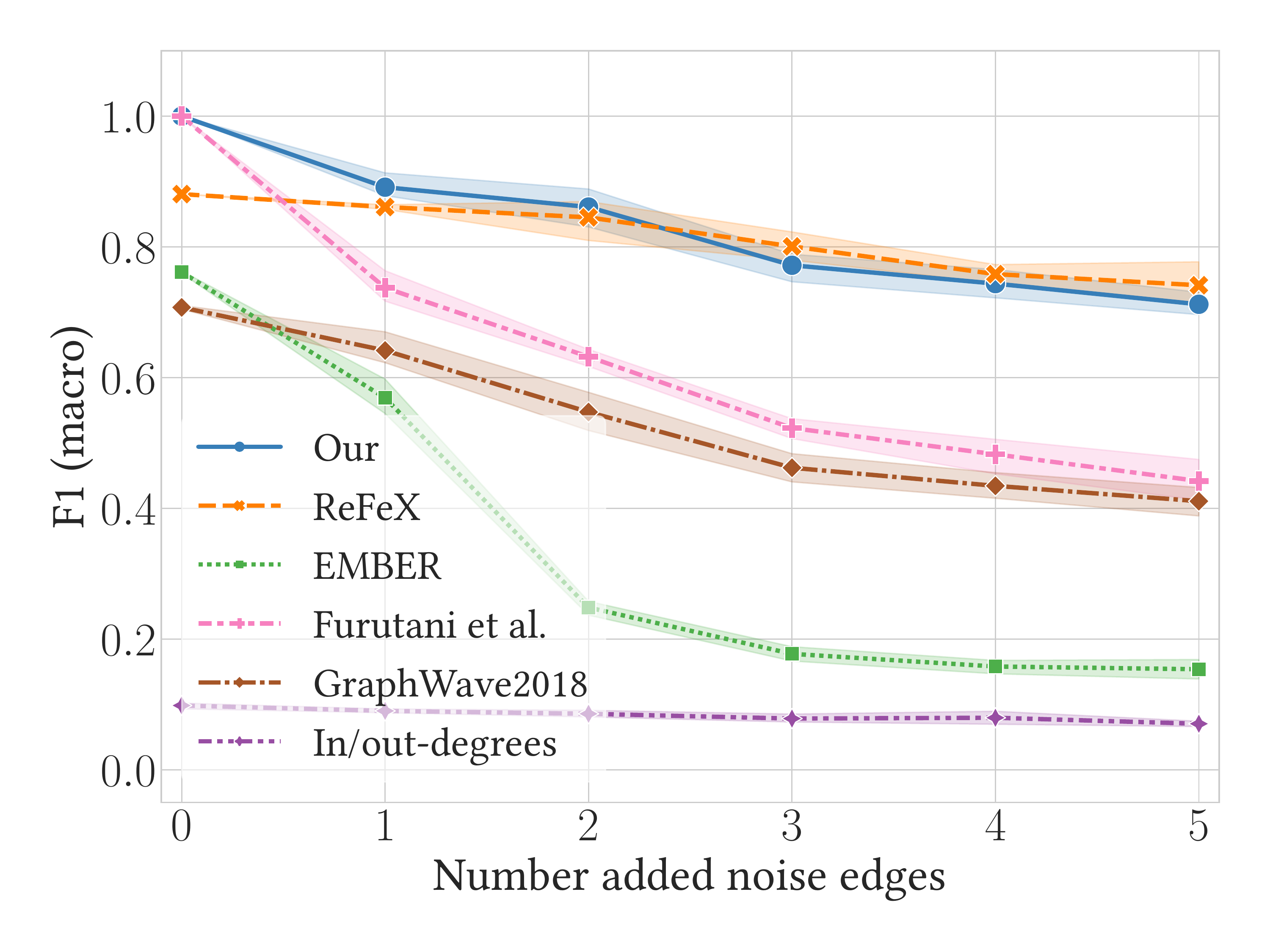}
    \caption{Classification}
    \end{subfigure}
    ~
    \begin{subfigure}{0.49\textwidth}
    \centering
    \includegraphics[width=\linewidth]{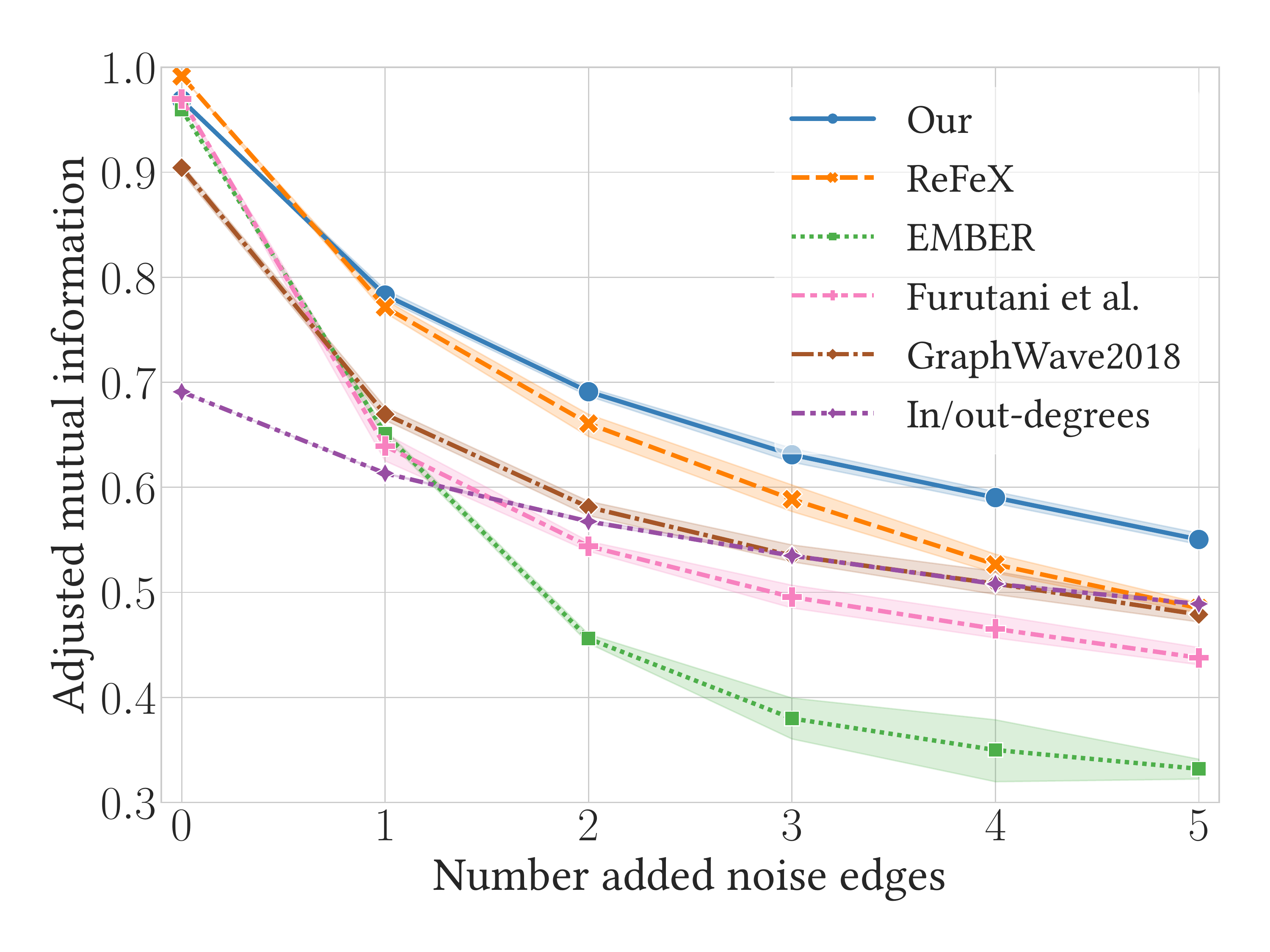}
    \caption{Clustering}
    \end{subfigure}
    \Description{Two plots side-by-side showing the results of the embedding method comparision experiment on the synthetic dataset. The left plot shows the classification performance, and the right plot the clustering performance.}
    \caption{Results of the embedding quality comparison experiment on the synthetic dataset.}
    \label{fig:synth:alg_comp}
\end{figure}

\subsection{Enron Node Classification} \label{sec:enron_data}

To evaluate how the embeddings perform on a node classification task on a real directed graph, a dataset with ground truth node labels which correlates with the nodes' local structures is required.
To this end, we use the Enron email communication corpus \cite{enron}, a well-known dataset which contains all the email communication from the Enron corporation, and which was made public as a result of legal action.
An email communication graph can be extracted from the corpus by considering each email address a node, and with edges between node $i$ and $j$ meaning that at least one email has been sent from address $i$ to address $j$. 

While many graphs extracted from the Enron corpus exist online, we have been unable to locate a version which includes node labels, or even the email address of each node.
However, a file specifying nine different company roles -- employee, manager, vice president etc.\ -- for a few email addresses is available online\footnote{\url{https://dsg.tuwien.ac.at/team/dschall/email/enron-employees.txt}}.
Thus, by extracting the communication graph from the corpus, and thereby matching email addresses to nodes, this ground truth role information can be used to define a node classification task.
Our extracted graph consists of 91685 nodes and 367995 edges, and edge weights are the number of emails sent between the respective nodes. 
The graph includes nodes of both Enron and non-Enron emails addresses, but not any edges between non-Enron nodes.

Unfortunately, the obtained file only contains 130 node labels and the nine classes are unbalanced.
To make good use of these few labels, and to avoid the classification result being dominated by the performance on the very rare classes, we create a simplified and more balanced set of labels: Employee (53 nodes), Manager (17 nodes) and Upper-management (60 nodes).
Moreover, we create another, larger, set of node labels based on the type of each email address: Enron (33561 nodes), External (51558 nodes) and Bot (6566 nodes).
These labels are extracted by first applying a heuristic to find any addresses which are likely controlled by an algorithm.
The remaining addresses are then assigned the Enron label if an Enron domain is used, and otherwise they are considered to be external.

Thus, we have two node label sets and two corresponding node classification tasks: company role classification and email type classification.
The experiment is performed both with and without the edge directions, and with and without node weights.
As for the synthetic dataset, a multinomial logistic regression model with $\ell_2$ penalty $C=1.0$ and class weighted loss function is used on top of the embeddings, and the performance is measured using the macro F1 score.
To obtain means and standard deviation estimates, 3-fold stratified cross validation repeated 5 times is used.
Due to their scalability issues, see Section \ref{sec:scalability}, \textsc{Graphwave} and the method of Furutani et al. are not used in this experiment.
For \textsc{Digraphwave}, we use $R=2$, as it works better than $R=3$.
This is likely due to the short average path length for the Enron graph, $\sim 4.5$.

The results are presented in Table \ref{tab:enron_results}.
\textsc{Digraphwave} and \textsc{ReFeX} achieve similar classification performance, with \textsc{Digraphwave} coming out slightly ahead on the company roles and vice-versa for the email types.
Both methods perform better when edge directions are considered, but including the weights does not necessarily improve the results.

\textsc{EMBER} performs similar to, or even worse, than the trivial baseline, and when edge weights are used it seems to completely fail.
We have not found any mention of required weight normalization in either \cite{ember_19} or the associated repository, and are unsure of the cause of this failure.

\begin{table}[t]
    \centering
    \caption{The node classification results on the Enron email communication graph. The models are sorted by classification performance and the best performing model is highlighted, as are all models within one standard deviation of the same result.}
    \label{tab:enron_results}
    \begin{subtable}[h]{0.48\linewidth}
    \centering
      \caption{Company role}
      \label{tab:enron:roles}
      \footnotesize{
      \begin{tabular}{ccp{15mm}c}
\toprule
Directed & Weights & Embedding algorithm & F1 (macro)\\
\midrule
Yes & No & Our & $\mathbf{0.49 \pm 0.07}$\\
Yes & No & ReFeX & $\mathbf{0.46 \pm 0.06}$\\
Yes & Yes & ReFeX & $\mathbf{0.46 \pm 0.08}$\\
No & Yes & Our & $\mathbf{0.45 \pm 0.06}$\\
No & Yes & ReFeX & $\mathbf{0.45 \pm 0.07}$\\
No & No & ReFeX & $\mathbf{0.45 \pm 0.06}$\\
Yes & No & In/out-degrees & $0.41 \pm 0.07$\\
No & No & Our & $0.41 \pm 0.05$\\
Yes & Yes & Our & $0.41 \pm 0.04$\\
Yes & No & EMBER & $0.40 \pm 0.06$\\
Yes & Yes & In/out-degrees & $0.40 \pm 0.06$\\
No & Yes & In/out-degrees & $0.40 \pm 0.06$\\
No & No & EMBER & $0.33 \pm 0.06$\\
No & No & In/out-degrees & $0.31 \pm 0.09$\\
No & Yes & EMBER & $0.11 \pm 0.06$\\
Yes & Yes & EMBER & $0.09 \pm 0.04$\\
\bottomrule
\end{tabular}
      }
    \end{subtable}
    ~
    \begin{subtable}[h]{0.48\linewidth}
    \centering
      \caption{Email type}
      \label{tab:enron:email_type}
      \footnotesize{
      \begin{tabular}{ccp{15mm}c}
\toprule
Directed & Weights & Embedding algorithm & F1 (macro)\\
\midrule
Yes & No & ReFeX & $\mathbf{0.50 \pm 0.00}$\\
Yes & Yes & ReFeX & $0.48 \pm 0.00$\\
Yes & No & Our & $0.48 \pm 0.00$\\
Yes & Yes & Our & $0.47 \pm 0.00$\\
No & No & ReFeX & $0.46 \pm 0.00$\\
No & Yes & ReFeX & $0.46 \pm 0.00$\\
No & Yes & Our & $0.45 \pm 0.00$\\
No & No & Our & $0.44 \pm 0.01$\\
Yes & No & In/out-degrees & $0.28 \pm 0.00$\\
Yes & Yes & In/out-degrees & $0.28 \pm 0.00$\\
No & Yes & In/out-degrees & $0.22 \pm 0.00$\\
Yes & No & EMBER & $0.22 \pm 0.00$\\
No & No & EMBER & $0.22 \pm 0.01$\\
No & No & In/out-degrees & $0.22 \pm 0.00$\\
No & Yes & EMBER & $0.04 \pm 0.00$\\
Yes & Yes & EMBER & $0.04 \pm 0.00$\\
\bottomrule
\end{tabular}
      }
    \end{subtable}
\end{table}

\subsection{Network Alignment} \label{sec:network_alignment}

Network alignment is the task of finding a matching correspondence between the nodes of two graphs, and solutions to this task has many applications within biology, computer vision and social network analysis \cite{net_alginment, regal_xnetmf_18}.
It has been shown in previous work \cite{regal_xnetmf_18} that structural node embeddings can be used to partially address this task.
Specifically, \cite{regal_xnetmf_18} proposes to extract structural node embeddings for the two graphs and then use greedy matching, i.e., each node gets matched with the node in the other graph with the most similar embedding, measured using the Euclidean distance.
While this approach ignore constrains which may improve performance, e.g., that two nodes should not be matched to the same node, it is exceedingly fast using KD-trees.
The results in \cite{regal_xnetmf_18} also indicate that this greedy approach can outperform more computationally expensive approaches which use message-passing and linear programming. 

We thus adopt the experiment setup of \cite{regal_xnetmf_18} to evaluate the embedding methods.
Specifically, given two graphs, $G_1 = (V_1, E_1)$ and $G_2 = (V_2 E_2)$, for which we know the node correspondences, and assuming $|V_1| \geq |V_2|$, we create a KD-tree of from the structural node embeddings of $G_1$ and then compute the $k$ closest embedding in the tree for each embedding extracted from $G_2$.
The top-$k$ accuracy is then the proportion of nodes in $G_2$ which contain the correct node correspondence among these $k$ most similar embeddings.
Additionally, the embeddings' robustness to noise is tested by repeating the above procedure for eight different noise levels, $p \in \{ 0, 0.01, 0.02, 0.03, 0.04, 0.05, 0.075, 0.1 \}$.
For each noise level, $p (|E_1| + |E_2|) $ noise edges are inserted into the graphs at random.
Five different sets of sampled edges per noise level are used to obtain mean and standard deviation estimates.


As datasets, one graph from \cite{regal_xnetmf_18} -- an undirected collaboration graph of the arXiv astrophysics category with 18772 nodes and 396100 edges \cite{leskovec_07} -- is used with the same setup, i.e., $G_1$ is the original graph and $G_2$ is the same graph but with all node indices randomly permuted.
Since all graphs used in \cite{regal_xnetmf_18} are undirected, we additionally perform the experiment on the Enron graph.
However, instead of simply randomly permuting the node indices, we use the email type node labels from Section \ref{sec:enron_data} and create $G_2$ as a subgraph of $G_1$.
Specifically, we take $G_2$ to be the induced subgraph of the 33561 Enron internal nodes.
Unlike the arXiv graph, the Enron graph has edge weights, and the experiment is performed both with and without theses.
To associate weights with the noise edges, the empirical distribution of the observed weights is used for sampling.
The embedding algorithms are set up as before, with $R=2$ for \textsc{Digraphwave} for both datasets, since, similarly to the Enron graph, the average path distance in the arXiv graph is short, $\sim 4$.

Top-1 accuracy results are shown in Figure \ref{fig:alignment_enron_all_internalk1}, and top-10 accuracy results can be found in the Appendix \ref{app:alignment}.
For the undirected and unweighted arXiv graph, the three embedding methods all reach close to $80$\% accuracy in the noise free case, with \textsc{EMBER} being less robust to added noise, with $2.6$\% accuracy for noise level $p=0.1$, compared to around $8.7$\% for \textsc{Digraphwave} and \textsc{ReFeX}.

For the more difficult setting using the unweighted Enron dataset, \textsc{Digraphwave} outperforms the other methods for all noise levels, with $40$\% top-1 accuracy in the noise free case, compared to $15$\% for \textsc{ReFeX} and $2.7$\% for \textsc{EMBER}, and $8.5$\% at noise level $p=0.1$ with $6.5$\% for \textsc{ReFeX} and $0.7$\% for \textsc{EMBER}.
When edge weights are used, \textsc{ReFeX} has a distinct advantage over \textsc{Digraphwave} since its embeddings includes features from both the weighted and unweighted graph, while the current version of \textsc{Digraphwave} uses either the weighted or unweighted graph.
The same is true for the in and out-degree baseline.
Thus, we see large improvements for these methods, with accuracies more than doubling in Figure \ref{fig:alignment_enron_dir_weighted} compared to Figure \ref{fig:alignment_enron_dir_unweighted}.
The accuracies for \textsc{Digraphwave} also improves, around $30-50$\% relative to the accuracies for the unweighted graphs, but the improvement is likely to be larger if embeddings from both the weighted and unweighted graph would be combined, as for \textsc{ReFeX}.
Similar to the experiment in Section \ref{sec:enron_data}, \textsc{EMBER} seem to completely fail when edge weights are used.

\begin{figure}[tp]
     \begin{subfigure}{0.32\textwidth}
    \centering
    \includegraphics[width=\linewidth]{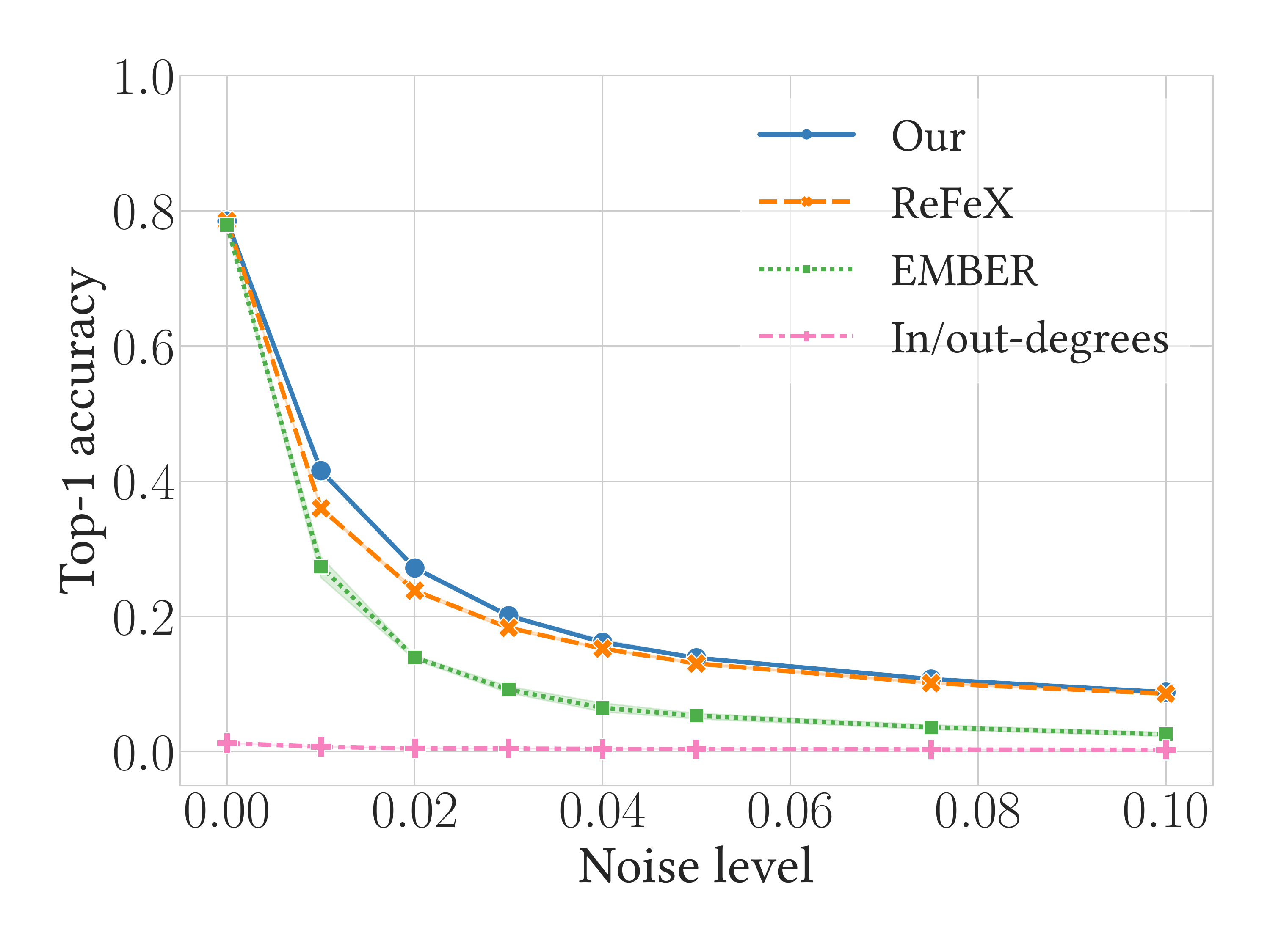}
    \caption{\footnotesize{arXiv, undirected and unweighted}}
    \end{subfigure}
    ~
    \begin{subfigure}{0.32\textwidth}
    \centering
    \includegraphics[width=\linewidth]{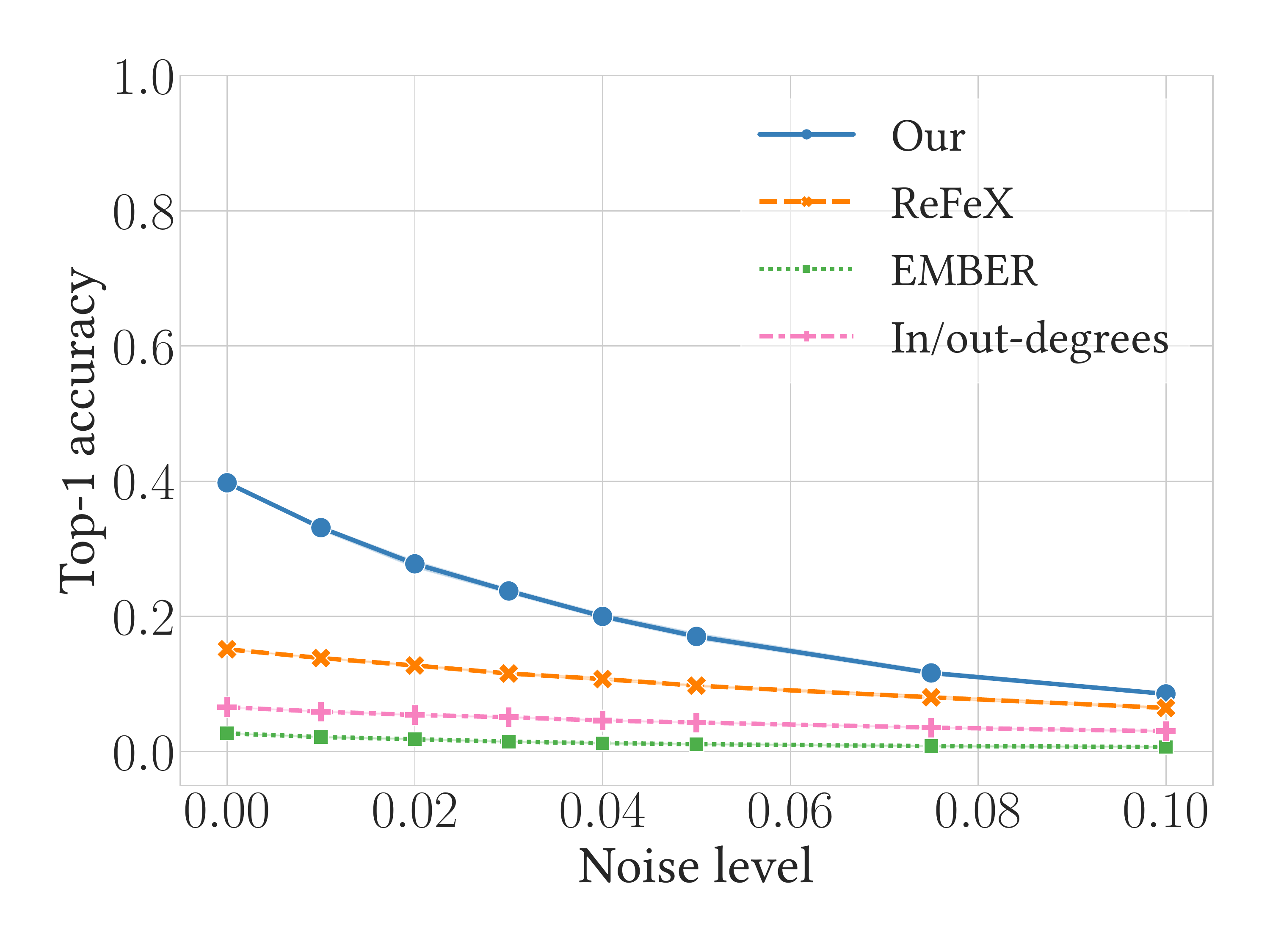}
    \caption{\footnotesize{Enron, directed and unweighted}}
    \label{fig:alignment_enron_dir_unweighted}
    \end{subfigure}
    ~
    \begin{subfigure}{0.32\textwidth}
    \centering
    \includegraphics[width=\linewidth]{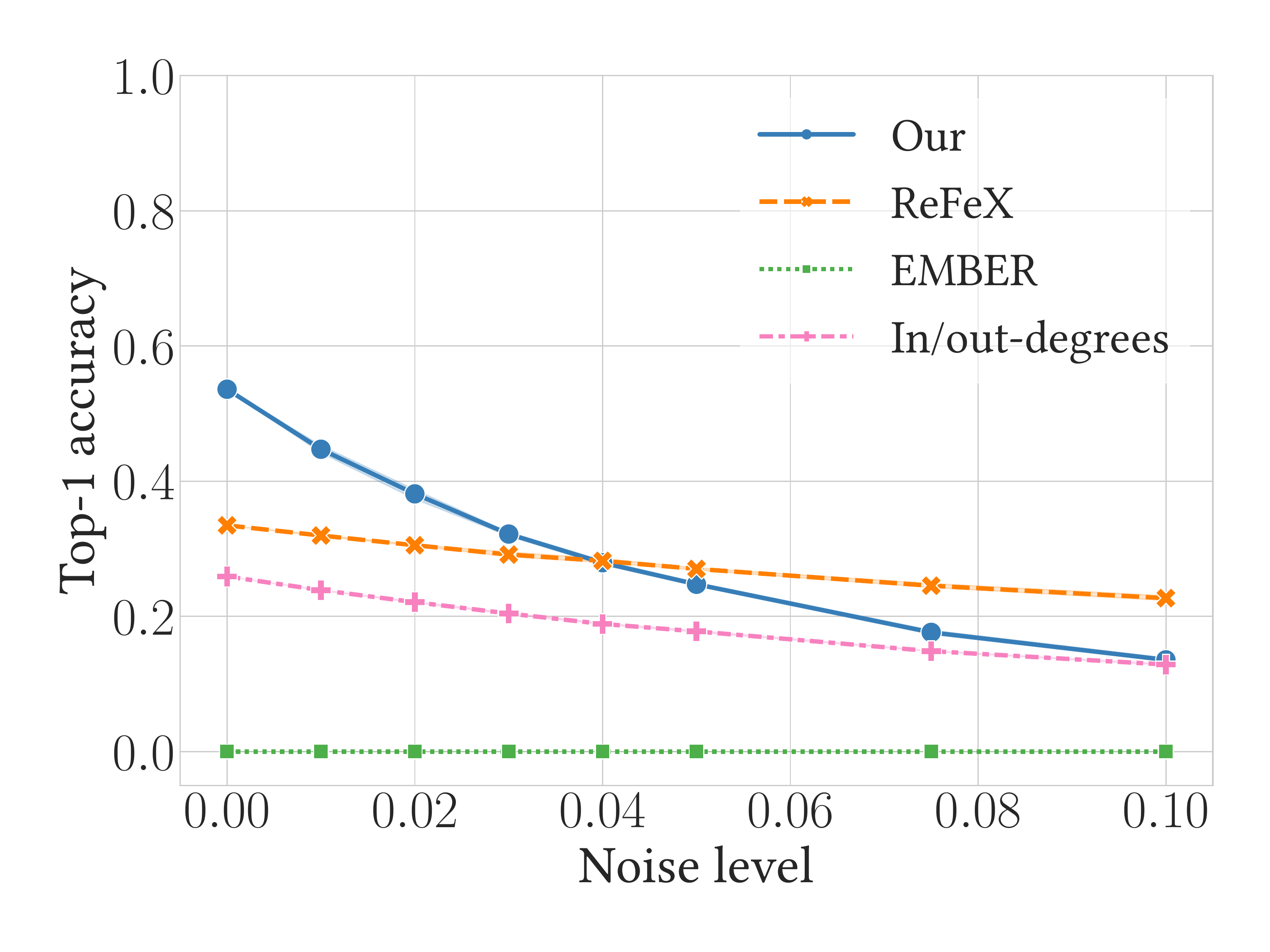}
    \caption{\footnotesize{Enron, directed and weighted}}
    \label{fig:alignment_enron_dir_weighted}
    \end{subfigure}
    \Description{Three plots placed  side-by-side containing the top-1 accuracy results of the network alignment experiments. }
    \caption{Results of the network alignment experiments. The noise level is the fraction of randomly inserted noise edges, and the lines show the mean result over five different sets of noise edges. Standard deviations are shown as bands, but these are barely visible due to their small magnitude, $10^{-4} - 10^{-3}$. The top-1 accuracy is the proportion of nodes in $G_2$ having themselves in $G_1$ as closest neighbour in the embedding space. Random change would score $\sim5 \cdot 10^{-5}$ on the arXiv graph, and  $\sim10^{-5}$ on the Enron graph.}
    \label{fig:alignment_enron_all_internalk1}
\end{figure}

\section{Conclusion and Outlook} \label{sec:conclusion_and_outlook}

In this work, we present \textsc{Digraphwave}, a scalable structural node embedding extraction algorithm based on the diffusion equation, and which is applicable to both directed and undirected graphs, with or without edge weights.
We also prove a lower bound on the heat contained in the local vicinity of a diffusion initialization node, and use this bound to determine appropriate hyperparameters for \textsc{Digraphwave}.
Our experiments show that \textsc{Digraphwave} vastly outperforms recent structural node embedding methods, both for identifying automorphic identities, and for node classification and network alignment tasks on the Enron email dataset.
We also discover that an early, yet often overlooked, structural embedding method. \textsc{ReFeX} \cite{henderson_its_2011, refex_reimplementation22}, an older structural embedding method, performs surprisingly well on our benchmark tasks while being very scalable.
\textsc{Digraphwave} does achieve $13$\% better macro F1 score in our automorphism identification task, and more than twice the accuracy in the network alignment task on the unweighted Enron graph.
Still, we encourage readers working in this field to include \textsc{ReFeX} in their future benchmarks, given its performance and good scalability.

Despite its excellent performance on our benchmarks, we have ideas which we believe would further improve \textsc{Digraphwave}, both in terms of embedding quality and scalability.
For the quality, more distinct structural signatures can be created by considering the tuples $(\Psi_{\edgeji}, \spdist(j, i))$, instead of only the reachability values $\Psi_{\edgeji}$, and then using the multidimensional generalization of the ECF for compression.
Such signature enrichment would be particularly beneficial for unsupervised applications, e.g., clustering, anomaly detection and network alignment, where the performance is largely determined by the node representations, and an increased computational cost could be warranted.

In this work, the matrix exponential was computed using double precision and $K=40$ for the Taylor approximation based on the error bound in \eqref{eq:lapexpm_error_bound}.
However, it is likely that the numerical precision can be reduced without much affecting the quality of the embeddings.
Preliminary tests indicate that lowering $K$ to $20$ and using single precision would result in more than $4$x speed up, though further experiments are needed to establish the effect on the embedding quality.

An interesting machine learning research avenue is to exploit that \textsc{Digraphwave} is differentiable with respect to the elements of the adjacency matrix.
Learning edge weights to improve downstream task performance is one possible use, and learned graph generation another.
Similarly, differentiation with respect to the Taylor coefficients is also possible, meaning that these could also be optimized for downstream task performance, as previously been done for certain convolutional graph neural networks on undirected graphs \cite{Defferrard16}.

It would also be interesting to see how the diffusion equation and matrix exponential could be generalized to extract different structural signatures.
For example, a nonlinear term $g(\tau, \ubold(\tau))$ may be added to the diffusion equation
and the solution can still be formulated using the matrix exponential \cite{al-mohy_computing_2011}, meaning that the Taylor series approximation can still be used to compute the solution.
Moreover, several node properties, e.g., subgraph centrality, communicability and betweenness, can be expressed as matrix functions of the adjacency matrix \citep[Sec. 3]{computing_matrix_functions_2010}, suggesting that applying the same matrix functions to the Laplacian or normalized adjacency matrix $\Alphab$ may yield other interpretable node properties.

Finally, recent work has suggested a unifying framework for structural and positional node embeddings \cite{zhu_2021_proximity_is_all_you_need}, where both embedding types can be extracted by first creating a proximity matrix, and applying the ECF for structural embeddings, and matrix factorization for positional embeddings.
Since this framework currently only considers undirected graphs, it will be interesting to see how it can be generalized to directed graphs and if the \textsc{Digraphwave} reachability values also can be used to extract positional embeddings.

\begin{acks}
This work was partially supported by the Wallenberg AI, Autonomous Systems and Software Program (WASP) funded by the Knut and Alice Wallenberg Foundation.
\end{acks}
\bibliographystyle{ACM-Reference-Format}
\bibliography{references}
\appendix

\section{Additional \textsc{Digraphwave} Details} \label{app:additional_details}

\subsection{Random-walk Interpretation} \label{app:rw_intepretation}
We here show how the matrix exponential $\expmlap$ can be interpreted from a random walk perspective.
First, we define $\mathbf{P}$ as
\begin{align*}
    \mathbf{P} 
    = \eye - \lapnorm 
    = \eye - \eyestarb  + \adj \degDstarb^{-1},
\end{align*}
and note that it is a column-stochastic transition matrix. 
This means that $P(i|j, k) = [\mathbf{P}^k]_{ij}$ is the probability of transitioning from $j$ to $i$ in $k$ steps.
Second, the Poisson probability distribution with parameter $\tau$ is given by $p(k) = \tau^k e^{-\tau}/ k!$.

Now we consider the Taylor series of $\expmlap$ expanded in $\eye$:
\begin{align*}
    \expmlap 
    = \sum_{k=0}^\infty \frac{(-\tau)^k e^{-\tau}}{k!} (\lapnorm - \eye )^k
    = \sum_{k=0}^\infty  \frac{(-\tau)^k e^{-\tau}}{k!} (-\mathbf{P})^k 
    = \sum_{k=0}^\infty \frac{(\tau)^k e^{-\tau}}{k!} \mathbf{P}^k 
    = \sum_{k=0}^\infty p(k) \mathbf{P}^k.
\end{align*}
From the last step, we see that $\expmlap$ is the marginalization of the joint distribution $ P(i, k|j)= p(k) \mathbf{P}^k $ over the number of walk steps $k$.
This has the interpretation of transitioning from node $j$ to node $i$ under a random walk with the number of steps being sampled from $\mathrm{Pois}(\tau)$.

\subsection{Stationary Solution to Diffusion Equation for an Undirected Graph} \label{app:stationary}


Here, we prove that $u_j(\tau)= c d_j$ is a stationary solution to the diffusion equation \eqref{eq:advection_diffusion_diff_eq} for an undirected graph, and where $d_j = \sum_{k \in V} A_{\edgeuv{k}{j}} = \sum_{k \in V} A_{\edgeuv{j}{k}}$ is the weighted degree of the node $j$ and $c$ is a normalization constant.
Stating from \eqref{eq:advection_diffusion_diff_eq}, we have
\begin{align*}
\frac{\text{d}u_i}{\text{d} \tau} (\tau) 
    &= -[\lapnorm \ubold(\tau)]_i 
    = \sum_{j \in V} \mathcal{L}_{\edgeji} u_j(\tau) 
    = \sum_{j \in V} (\alpha_{\edgeji} - \eyestar_{ij}) u_j(\tau).
\end{align*}
Inserting $u_j(\tau) = c d_j$ on the RHS,
\begin{align*}
    c \sum_{j \in V} (\alpha_{\edgeji} - \eyestar_{ij}) d_j 
    &= c \sum_{j \in V} \alpha_{\edgeji} d_j - c \sum_{j \in V} \eyestar_{ij} d_j
    = c \sum_{j \in \Ncal(i)} \weightnorm_{\edgeji}  d_j - c d_i
    = c \sum_{j \in \Ncal(i)} \weight_{\edgeji} - c d_i
    = c d_i - c d_i = 0.
\end{align*}




\subsection{\textsc{FFT} for Taylor Coefficients} \label{sec:app:fft}
Here, we show that the Taylor series coefficients can be calculated using the discrete Fourier transform.
Interested readers may turn to \cite{ellacott_83, bergamaschi_efficient_2003} to see how this can be generalized to Chebyshev and Faber coefficients. 

For a holomorphic function $f(z)$ with a convergent power series expanded in $\alpha \in \Cim$,
\begin{align*}
    f(z) &= \sum_{k=0}^\infty a_k \left(z - \alpha \right)^k\\
    a_k &= \frac{f^{(k)}(\alpha)}{k!},
\end{align*}
Cauchy's integral formula states that
\begin{align*}
    \frac{f^{(k)}(\alpha)}{k!} &= \frac{1}{2 \pi i} \oint_\gamma \frac{f(z)}{(z - \alpha)^{k + 1}} \mathrm{d}z,
\end{align*}
where $\gamma$ is the circular boundary of a disc, oriented counterclockwise, in the complex plane which contains $\alpha$.
On this boundary, we can choose the parameterization $z = \alpha + e^{2 \pi  i \theta}$ with $\theta \in [0, 1)$:
\begin{align*}
    \frac{f^{(k)}(\alpha)}{k!} &= \int_0^1 f(\alpha + e^{2 \pi  i \theta}) e^{-2 \pi  i k \theta}\mathrm{d}\theta.
\end{align*}
By discretizing this formula using $M$ steps, $\theta = \frac{\kappa}{M}$ with $\kappa \in \{0, \dots, M-1 \}$, we obtain
\begin{align*}
    \frac{f^{(k)}(\alpha)}{k!} &= \frac{1}{M} \sum_{\kappa =0}^{M-1} f(\alpha + e^{2 \pi  i \frac{\kappa}{M}}) e^{-2 \pi  i k \frac{\kappa}{M}}
\end{align*}
which is the discrete Fourier transform.

\subsection{Error Bound} \label{sec:app:error_bound}
We here prove the error bound \eqref{eq:lapexpm_error_bound} using \citep[Ch. 4]{Higham08} as guidance.
For unfamiliar readers, all matrix norms, including $\| \cdot \|_1$, have the following useful properties \cite{horn_johnson_2012}:
\begin{align}
    \label{eq:matrix_norm_scalar}
    \|c \mathbf{X} \| &= |c| \| \mathbf{X} \| \text{ for all } c \in \Cim \\
    \label{eq:matrix_norm_addition}
    \| \mathbf{X} + \mathbf{Y}  \| &\leq \| \mathbf{X} \| + \| \mathbf{Y}  \| \\
    \label{eq:matrix_norm_multiplication}
    \| \mathbf{X} \mathbf{Y}  \| &\leq \| \mathbf{X} \| \| \mathbf{Y}  \|.
\end{align}
Consequently,
\begin{align}
    \label{eq:expm_ineq}
    \| \exp (\mathbf{X}) \| 
    &= \left\| \eye + \mathbf{X}+ \frac{\mathbf{X}^2}{2!} + \frac{\mathbf{X}^3}{3!} + \cdots  \right\| 
    \leq 1 + \|\mathbf{X}\| + \frac{\|\mathbf{X}\|^2}{2!} + \frac{\|\mathbf{X}\|^3}{3!} + \cdots
    = \exp (\|\mathbf{X} \|).
\end{align}

The error bound has two terms, the Taylor approximation error and the rounding error, 
\begin{align} \label{eq:general_error_bound}
    \| \expmlap - \mathbf{\hat{T}}_K (\lapnorm, \tau) \|_1 
    &\leq \| \expmlap - \mathbf{T}_K (\lapnorm, \tau) \|_1 
    + \|\mathbf{T}_K (\lapnorm, \tau) - \mathbf{\hat{T}}_K (\lapnorm, \tau)  \|_1.
\end{align}
Since the Taylor polynomial is implemented using explicit powers, i.e., Algorithm 4.3 in \citep[Ch. 4]{Higham08}, we can apply Theorem 4.5 in the same text to obtain the bound on the rounding error term,
\begin{equation*}
    \|\mathbf{T}_K (\lapnorm, \tau) - \mathbf{\hat{T}}_K (\lapnorm, \tau)\|_1 \leq \tilde{\gamma}_{Kn} \exp( \tau \| \lapnorm - \eye\|_1).
\end{equation*}

For the Taylor approximation error, we use Theorem 4.8 in \citep[Ch. 4]{Higham08},
\begin{equation}
    \left\| f(\mathbf{X}) - \sum_{k=0}^K a_k \left(\mathbf{X} - \eye \right)^k \right\|_1
    \leq  \frac{1}{(K+1)!} \max_{0 \leq t \leq 1} 
    \| \left(\mathbf{X} - \eye \right)^{K+1} 
    f^{(K+1)}( \eye + t(\mathbf{X} - \eye))  \|_1,
\end{equation}
with $f(\lapnorm) = \expmlap$,
\begin{align*}
\left\| \expmlap - T_{K}(\lapnorm) \right\|_1 
\leq&
    \frac{1}{(K+1)!} \max_{0 \leq t \leq 1} 
    \| \left(\lapnorm - \eye \right)^{K+1}
    (-\tau)^{K+1}
    \exp(-\tau (\eye + t(\lapnorm - \eye)))  \|_1.
\end{align*}
Using \eqref{eq:matrix_norm_multiplication} we get
\begin{align*}
    \left\| \expmlap - T_{K}(\lapnorm) \right\|_1 
\leq&
    \frac{\tau^{K+1}\| \lapnorm - \eye \|^{K+1}_1}{(K+1)!} 
    \max_{0 \leq t \leq 1} 
    \|\exp(-\tau (\eye + t(\lapnorm - \eye)))  \|_1,
\end{align*}
and the rightmost factor can be simplified by using \eqref{eq:expm_ineq} and \eqref{eq:matrix_norm_addition}:
\begin{align*}
    \max_{0 \leq t \leq 1} 
    \|\exp(-\tau (\eye + t(\lapnorm - \eye)))  \|_1 
    &\leq
    \max_{0 \leq t \leq 1} \exp(\tau + \tau t\|\lapnorm - \eye\|_1)
    \leq
    \exp(\tau + \tau\|\lapnorm - \eye\|_1).
\end{align*}
Moreover, since $\| \lapnorm - \eye \| \leq 1$ by the definition of $\lapnorm$, $\| \lapnorm - \eye \|^{K+1} \leq 1$ and $\exp(\tau\|\lapnorm - \eye\|_1) \leq \exp(\tau)$, so the Taylor error bound is 
\begin{align*}
    \left\| \expmlap - T_{K}(\lapnorm) \right\|_1 
    \leq 
        \frac{\tau^{K+1}}{(K+1)!} \exp(2\tau).
\end{align*}

We insert the bounds back into \eqref{eq:general_error_bound} and get
\begin{equation*}
    \| \expmlap - \mathbf{\hat{T}}_K \|_1
\leq
    \left(\frac{\exp(\tau)\tau^{K+1} }{(K+1)!} 
    + 
    \tilde{\gamma}_{Kn}
    \right)
    \exp(\tau)
     .
\end{equation*}

\section{Additional Theory Details and Proofs} \label{app:proofs}

\subsection{Heat Transport Maximality Criteria} \label{app:proof_heat_transport_maximal}
\begingroup
\def\thetheorem{\ref{lemma:heat_transport_maximal}}
\begin{lemma}
A graph for which the following three criteria are fulfilled for every $r \leq R$ is $(j, R)$-heat transport maximal.
\begin{enumerate}
    \item There exist no edges going from $\Nperiph(j, r)$ to $\Ncore(j, r)$.
    \item All nodes in $\mathcal{C}(j, r)$ have at least one outgoing edge.
    \item All outgoing edges from nodes in $\mathcal{C}(j, r)$ go to $\Nperiph(j, r)$.
\end{enumerate}
\end{lemma}
\addtocounter{theorem}{-1}
\endgroup
\begin{proof}


For the proof, we use the fact that
\begin{align*}
    u_b(j, \tau, r) 
    = \int_0^\tau \dot{u}_b(j, s, r)  \text{d}s +  u_b(j, 0, r) 
    = \int_0^\tau \dot{u}_b(j, s, r)  \text{d}s +  1.
\end{align*}
Thus, proving that $\dot{u}_b(j, \tau, r) $ is minimal with respect to $\Alphab$ implies that $u_b(j, \tau, r)$ is minimal.
We use proof by induction over $R$.

\textbf{Base case:} 
For $R=0$ we have
\begin{align*}
    \dot{u}_b(j, \tau, 0) = \sum_{i \in \Ncore(j, 0)} \dot{u}_i(\tau) =  \dot{u}_j(\tau).
\end{align*}
We expand the RHS using the diffusion equation \eqref{eq:heat_transport_diff_eq},
\begin{align} \label{eq:base_case_diff_eq}
    \dot{u}_j(\tau) 
    &= \sum_{k \in V} -\mathcal{L}_{\edgeuv{k}{j}} u_k(\tau) 
    =\sum_{k \in \Nperiph(j, r)} \alpha_{\edgeuv{k}{j}} u_k(\tau) - \eyestar_{jj} u_j(\tau).
\end{align}
Since $u_k(\tau) \geq 0 $ and $\alpha_{\edgeuv{k}{j}} \geq 0$, see Section \ref{sec:diffusion_on_graphs}, it is clear that the sum should be minimized and $\eyestar_{jj} u_j(\tau)$ maximized for $\dot{u}_j(\tau)$ to be minimal.
This is achieved exactly by setting all $\alpha_{\edgeuv{k}{j}}$ in the sum to $0$, which is criteria (1) in the lemma.
Second, for $\eyestar_{jj}=1$, node $j$ needs to have at least one outgoing edge, criteria (2).
Since we do not allow self-loops, the target node must be in $\Nperiph(j, r)$, criteria (3).
By this, \eqref{eq:base_case_diff_eq} reduces to 
\begin{align*}
    \dot{u}_j(\tau) = - u_j(\tau),
\end{align*}
which together with the initial condition  $u_0(j) = 1$ fully determines $u_j(\tau)$.
Thus, the three criteria implies that a graph is $(j, 0)$-hear transport maximal.

\textbf{Induction case:} 
For the induction case, we use the induction assumption that the graph is $(j, R-1)$-heat transport maximal from fulfilling the three criteria.

Again, we consider the derivative of heat in $\Ncore(j, R)$
and expand it using the diffusion equation \eqref{eq:heat_transport_diff_eq}:
\begin{align*}
    \dot{u}_b(j, R, \tau) 
    &= \sum_{i \in \Ncore(j, R)} \dot{u}_i(\tau) 
    = \sum_{i \in \Ncore(j, R)} \sum_{k \in V} -\mathcal{L}_{\edgeuv{k}{i}} u_k (\tau)
    = \sum_{i \in \Ncore(j, R)} \sum_{k \in V} \alpha_{\edgeuv{k}{i}} u_k (\tau)  - \sum_{i \in \Ncore(j, R)} \eyestar_{ii} u_i (\tau).
\end{align*}
We can decompose the first double sum based on the source nodes as 
\begin{align*}
    \begin{split}
    \sum_{i \in \Ncore(j, R)} \sum_{k \in V} \alpha_{\edgeuv{k}{i}} u_k (\tau) 
    &= \sum_{i \in \Ncore(j, R)} \sum_{k \in  \Ncore(j, R-1)} \alpha_{\edgeuv{k}{i}} u_k (\tau) 
    + \sum_{i \in \Ncore(j, R)} \sum_{k \in  \mathcal{C}(j, R)}  \alpha_{\edgeuv{k}{i}} u_k (\tau) + \sum_{i \in \Ncore(j, R)} \sum_{k \in  \Nperiph(j, R)}  \alpha_{\edgeuv{k}{i}} u_k (\tau) .
\end{split}
\end{align*}

The first term in this decomposition accounts for edges with the source node $k \in \Ncore(j, R-1)$ and target node $i \in \Ncore(j, R)$.
Per definition, all edges with a source in $\Ncore(j, R-1)$ will have all their targets in $\Ncore(j, R)$, and since $\lapnorm$ is normalized $\sum_{i \in \Ncore(j, R)} \alpha_{\edgeuv{k}{i}} = \eyestar_{kk} $.
Thus, we can rewrite the first term as
\begin{align*}
    \sum_{i \in \Ncore(j, R)} \sum_{k \in  \Ncore(j, R-1)} \alpha_{\edgeuv{k}{i}} u_k (\tau) 
    &= \sum_{k \in  \Ncore(j, R-1)} u_k (\tau) \sum_{i \in \Ncore(j, R)} \alpha_{\edgeuv{k}{i}} 
     = \sum_{k \in  \Ncore(j, R-1)} \eyestar_{kk}  u_k (\tau).
\end{align*}

We can now express $\dot{u}_b(j, R, \tau) $ using the decomposition and simplify it to three terms
\begin{align*}
    \begin{split}
        \dot{u}_b(j, R, \tau) 
        &= \sum_{k \in  \Ncore(j, R-1)} \eyestar_{kk}  u_k (\tau) 
        + \sum_{i \in \Ncore(j, R)} \sum_{k \in  \mathcal{C}(j, R)}  \alpha_{\edgeuv{k}{i}} u_k (\tau) \\
        & \qquad + \sum_{i \in \Ncore(j, R)} \sum_{k \in  \Nperiph(j, R)}  \alpha_{\edgeuv{k}{i}} u_k (\tau) 
        - \sum_{i \in \Ncore(j, R)} \eyestar_{ii} u_i (\tau)
    \end{split} \\
    &=  \underbrace{
    \sum_{i \in \Ncore(j, R)} \sum_{k \in  \mathcal{C}(j, R)}  \alpha_{\edgeuv{k}{i}} u_k (\tau)
    }_{\circled{A}}
    + \underbrace{
    \sum_{i \in \Ncore(j, R)} \sum_{k \in  \Nperiph(j, R)}  \alpha_{\edgeuv{k}{i}} u_k (\tau) 
    }_{\circled{B}}
    - \underbrace{
    \sum_{i \in \mathcal{C}(j, R)} \eyestar_{ii} u_i (\tau)
    }_{\circled{C}}.
\end{align*}
Similar to the base case, the positive terms $\circled{A}$ and $\circled{B}$ should be minimized while the negative term $\circled{C}$ should be maximized.
Setting all $\alpha_{\edgeuv{k}{i}} = 0$ in $\circled{B}$ implements criteria (1), and if each $i \in \mathcal{C}(j, R)$ have at least one outgoing edge $\eyestar_{ii} = 1$ for all $i$ in $\circled{C}$, criteria (2).
Finally, if all of these outgoing edges go to $\Nperiph(j, R)$, criteria (3), term $\circled{A}$ is zero.

This result in the differential equation,
\begin{align*}
    \dot{u}_b(j, R, \tau) &= - \sum_{i \in \mathcal{C}(j, R)} u_i (\tau),
\end{align*}
where we can rewrite the RHS as
\begin{align*}
    - \sum_{i \in \mathcal{C}(j, R)} u_i (\tau) 
    =
   u_b(j, R-1, \tau) - u_b(j, R-1, \tau) - \sum_{i \in \mathcal{C}(j, R)} u_i (\tau) 
    =
     u_b(j, R-1, \tau) - u_b(j, R, \tau).
\end{align*}
Together with the initial condition $u_b(j, R, 0)=1$, and the induction assumption that $ u_b(j, R-1, \tau)$ is minimal,  the differential equation $\dot{u}_b(j, R, \tau) =  u_b(j, R-1, \tau) - u_b(j, R, \tau)$ implies that  $u_b(j, R, \tau)$ is also minimal.

\end{proof}

\subsection{Additional Source-Star Graph Details} \label{app:source_star}


The index of the source-star graph \eqref{eq:ss_index} can be inverted using the equations 
\begin{align*}
    l &= 
    \begin{cases}
        \lceil i / d \rceil & \text{if $\beta  = 1$}\\
        \left\lceil \log_{\beta} \left( 1 + i \frac{\beta - 1}{d} \right) \right\rceil & \text{if $\beta > 1$}
    \end{cases} 
    &&  \text{and} &
    \nu_l =
    \begin{cases}
    0 & \text{if $i = 0$} \\
    i - d (l -1) & \text{if $\beta  = 1$} \\
      i - d \frac{\beta^l - 1}{\beta - 1} & \text{if $\beta > 1$}.
    \end{cases}
\end{align*}
Furthermore, $\lceil \nu_l / \beta^{l-1} \rceil \in \{1, \dots, d \}$ can be used to determine in which of the $d$ branches a node in a source-star graph resides on.

We can express out-degree normalized Laplacian of a source-star graph as a block matrix
\begin{equation} \label{eq:laplace:source-star}
\begingroup 
\setlength\arraycolsep{2pt}
\lapnormstar(d, \beta, \ell) = 
    \left(
    \begin{array}{ccccccc}
    1  & 
    \bigzero_{1 \times d} &
    \bigzero_{1 \times \beta d} &
    \bigzero_{1 \times \beta^2 d} &
    \cdots &
    \bigzero_{1 \times \beta^{\ell \text{-} 2}d} &
    \bigzero_{1 \times \beta^{\ell \text{-} 1}d}
    \\
  \text{-}d^{\text{-}1} \eye_{d \times 1} 
  & \eye_{d \times d} 
  & \bigzero_{d \times \beta d}
  & \bigzero_{d \times \beta^2 d}
  & \cdots  
  & \bigzero_{d \times \beta^{\ell \text{-} 2} d}
  & \bigzero_{d \times \beta^{\ell \text{-} 1} d}\\
  \bigzero_{\beta d \times 1} &
  \text{-}\beta^{\text{-}1} \eye_{\beta d \times d} &
  \eye_{\beta d \times \beta  d} &
  \bigzero_{\beta d \times \beta^2 d} &
  &
  \vdots &
  \vdots\\
  \bigzero_{\beta^2 d \times 1} &
  \bigzero_{\beta^2 d \times d} & 
  \text{-} \beta^{\text{-}1} \eye_{\beta^2 d \times \beta d}  &
  \eye_{\beta^2 d \times \beta^2 d} &
  &
  \vdots & 
  \vdots 
  \\
  \vdots & 
  \vdots & 
  & 
  &
  \ddots &
  \vdots &
  \vdots 
  \\
  \bigzero_{\beta^{\ell \text{-} 2} d \times 1} &
  \bigzero_{\beta^{\ell \text{-} 2 } d \times d} &
  \cdots &
  \cdots &
  \cdots &
  \eye_{\beta^{\ell \text{-} 2} d \times \beta^{\ell \text{-} 2} d} &
  \bigzero_{\beta^{\ell \text{-} 2} d \times \beta^{\ell \text{-} 1} d}\\
  \bigzero_{\beta^{\ell \text{-} 1} d \times 1} &
  \bigzero_{\beta^{\ell \text{-} 1} d \times d} &
  \cdots &
  \cdots &
  \cdots &
  \text{-} \beta^{\text{-} 1} \eye_{\beta^{\ell \text{-} 1} d \times \beta^{\ell \text{-} 2} d}&
  \bigzero_{\beta^{\ell \text{-} 1} d \times \beta^{\ell \text{-} 1} d}
\end{array}
\right),
\endgroup
\end{equation}
where $\bigzero_{a \times b} $ contains all zeros entries, $\eye_{a \times a} $ is the identity matrix and  $\eye_{\beta a \times a} $ is a matrix of $\beta$ row stacked $\eye_{a \times a} $ identity matrices.

\subsection{Source Star Graph Minimal Heat} \label{app:proof_least_contained_heat}
\begingroup
\def\thetheorem{\ref{lemma:least_contained_heat}}
\begin{lemma}
For any directed graph $G=(V, E)$ and node $j \in V$, with out-degree $d_j \geq 1$ and with $\ell$ being the distance to the furthest away, but reachable, node from $j$, i.e.\ $\ell = \max_{k \in V} \{\spdist(j, k) | \spdist(j, k) < \infty \}$, there exists a source-star graph $\Gss(d_j, \beta, \ell, n_{\text{isolated}}) = (V^{(\text{ss})}, E^{(\text{ss})})$, with $\Ncore^{(\text{ss})}(0, R) \subseteq V^{(\text{ss})}$ such that $|\Ncore^{(\text{ss})}(0, R)| \geq |\Ncore(j, R)|$ and $u_b (j, R, \tau)  \geq u_b^{(\text{ss})} (0, R, \tau)$, i.e.\
\begin{align*}
    \sum_{i \in \Ncore(j, R)} u_i(\tau) \geq  \sum_{i \in \Ncore^{(\text{ss})}(0, R)} u_i(\tau),
\end{align*}
for any $R \leq \ell - 1$ and all $\tau \in \Rreal_{+}$.
\end{lemma}
\addtocounter{theorem}{-1}
\endgroup
\begin{proof}

We first construct the source-star graph $\Gss = (V^{(\text{ss})}, E^{(\text{ss})})$ to have the property,  $|\Ncore^{(\text{ss})}(0, R)| \geq |\Ncore(j, R)|$ for all $R \leq \ell  - 1$, then we show that $u_b (j, R, \tau)  \geq u_b^{(\text{ss})} (0, R, \tau)$ hold for this graph.

Let $n = |V|$ be the number of nodes in $G$, $n_{\infty} = |V| - |\Ncore(j, \ell)|$ be the number of nodes not reachable from $j$.
To fully specify $\Gss$, $\beta$ and $n_{\text{isolated}}$ need to be determined, which we do as follows:
\begin{enumerate}
    \item Chose $\beta \in \{ 1, 2, \dots \}$ such that $|\Ncore(v, r)| \leq |\Ncore^{(ss)}(0, r)| = 1 + d_j \sum_{l=0}^{r} \beta^{l-1} \quad \forall \ r \in \{1, \dots, \ell\}$.
    \item Set  $n_{\text{isolated}} = n_{\infty}$. 
\end{enumerate}
From (1), $|\Ncore^{(\text{ss})}(0, R)| \geq |\Ncore(j, R)|$ is ensured for all $R \geq \ell -1$, and (2) ensures $n^{(\text{ss})} = |V^{(\text{ss})}| \geq n$, which we use next to show that $\Gss$ fulfils the second stipulation.

To do so, we construct an intermediate graph $\tilde{G} = (\tilde{V}, E)$, which share edge set with $G$ and $(j,R)$-balls $\Ncore(j, R)$.
The only difference is that $\tilde{V}$ contains an additional $(n^{(\text{ss})} - n)$ isolated nodes, $\mathcal{N}_{\infty} = \{k \}^{n^{(\text{ss})}-1}_{k=n} $, compared to $V$, i.e., $\tilde{V} = V \cup \mathcal{N}_{\infty} $, meaning that $|\tilde{V}|=n^{(\text{ss})}$.
We use these additional nodes to construct sets $\mathcal{M}(j, R) = \Ncore(j, R) \cup \{k \}_{k=n}^{n +\Delta n_R - 1}$, where $\Delta n_R =|\Ncore^{(\text{ss})}(0, R)| - |\Ncore(j, R)|$.
With this construction, $|\mathcal{M}(j, R)| =|\Ncore^{(\text{ss})}(0, R)|$ for all $R \leq \ell - 1$, and since the added isolated nodes do not participate in the diffusion process,we  have
\begin{align*}
    sum_{i \in\mathcal{M}(j, R)} u_i(\tau) = \sum_{i \in \Ncore(j, R)} u_i(\tau) = u_b (\tau, R) 
\end{align*}
for each $R \leq \ell$ and $\tau$.

Now, since $|\tilde{V}| = |V^{(\text{ss})}|$, and $|\mathcal{M}(j, R)| =|\Ncore^{(\text{ss})}(0, R)|$, there exists an isomorphism between $\tilde{V}$ and $V^{(\text{ss})}$ which maps each node in $\mathcal{M}(j, R)$ to a node in $\Ncore^{(\text{ss})}(0, R)$ for each $R \leq \ell -1$, including mapping $j \in \tilde{V}$ to  $0 \in V^{(\text{ss})}$.
Then, by Lemma \ref{lemma:heat_transport_equivalence} and the fact that the source-star graph is $(0, R)$-heat transport maximal over $\Gss$ for every $R \leq \ell - 1$, we have
\begin{align*}
     \sum_{i \in \mathcal{M}(j, R)} u_i(\tau) \geq \sum_{i \in \Ncore^{(\text{ss})}(0, R)} u_i(\tau) 
    = u_b^{(\text{ss})} (0, R, \tau)
\end{align*}
for all $R \leq \ell - 1$ and $\tau$, which completes the proof.

\end{proof}

\subsection{Source Star Graph Analytical Expression} \label{app:ss_analytical_expression}

\begingroup
\def\thetheorem{\ref{lemma:ss_analytical_expression}}
\begin{lemma}
Given a source-star graph $\Gss(d, \beta, \ell)$ with out-degree normalized Laplacian $\lapnormstar$ as in \eqref{eq:laplace:source-star} and a diffusion process as defined by \eqref{eq:heat_transport_diff_eq} with $j=0$, we have
\begin{equation}
    u^{(ss)}_i(\tau) =
    [\exp (-\tau \lapnormstar ) \mathbf{b}]_i = 
    \begin{cases}
        \exp(-\tau) & \text{if $i = 0$} \\
        \displaystyle \frac{d^{-1}}{\beta^{l-1}}\frac{ \tau^l}{l!} \exp(- \tau) & \text{if $i \in \{1, \dots, n - d \beta^{\ell -1}  - 1 \}$} \\
        \displaystyle  \frac{d^{-1}}{\beta^{\ell-1}} P(\ell, \tau) & \text{if $i \in \{n - d \beta^{\ell -1}, \dots, n - 1\} $},
    \end{cases}
\end{equation}

where $P(\ell, \tau)$ is the normalized version of the lower incomplete gamma function \cite{NIST_DLMF8.2.E4}.

\end{lemma}
\addtocounter{theorem}{-1}
\endgroup
\begin{proof}
\newcommand{\branchscalar}{d^{-1}}

For the proof, we use $\lapnorm = \lapnormstar(d, \beta, \ell)$.
From the definition of the matrix exponential \eqref{eq:exp_taylor} we have
\begin{align}
    \expmlap \mathbf{b} &= \sum_{k=0}^{\infty}\frac{(- \tau)^k }{k!} \lapnorm^k \mathbf{b},
    &&
    \label{eq:laplace:source-star:taylor2}
    [\expmlap \mathbf{b}]_i 
    = \sum_{k=0}^{\infty}\frac{(- \tau)^k }{k!} b_{\nu_l}(l, k)
\end{align}
where $ b_{\nu_l}(l, k) = [\lapnorm^k \mathbf{b}]_i$.

From the structure of the source-star graph Laplacian in \eqref{eq:laplace:source-star}, the recurrence relation $\lapnorm^k \mathbf{b} = \lapnorm (\lapnorm^{k-1} \mathbf{b})$
can be expressed as a system of recurrence relations using $b_{\nu_l}(l, k)$:
\begin{align}
    b_{0}(0, k) &= 1 && k \geq 0 \label{eq:recurrence1}\\
     b_{\nu_1}(1, k) &=  b_{\nu_1}(1, k - 1)- \branchscalar &&  k \geq 1 \label{eq:recurrence2}\\
    b_{\nu_l}(l , k) &= b_{\nu_l}(l , k - 1) - \beta^{-1}b_{\nu_l}(l -1 , k - 1) && k \geq 1,\ l \in \{2, \dots, \ell - 1\} \label{eq:recurrence3}\\
    b_{\nu_\ell}(\ell , k) &= -  \beta^{-1} b_{\nu_{\ell - 1}}(\ell -1, k - 1) && k \geq 1,\ l = \ell \label{eq:recurrence4}
\end{align}
with $\nu_l \in \{1, \dots, d \beta^{l-1} \}$ and the initial condition $b_{\nu_1}(l, 0)$ for $ l > 0 $.
The solution to the recurrence relation \eqref{eq:recurrence2} is  $b_{\nu_1}(1, k) = -k \branchscalar $.
Using this as initial condition to the two variable recurrence \eqref{eq:recurrence3}, we find the following solution, see Appendix \ref{sec:app:recursion},
\begin{equation}
    b_\nu(l, k) = 
    \begin{cases}
       \displaystyle \branchscalar \beta^{-(l-1)} \binom{k}{l} (-1)^{l} & \text{if $ l \leq k$} \\
        0  & \text{otherwise}, 
    \end{cases}
\end{equation}
for $l \in \{1, \dots, \ell - 1\}$, and consequently \eqref{eq:recurrence4} is solved by
\begin{equation}
    b_\nu(\ell, k) = 
    \begin{cases}
        \displaystyle \branchscalar \beta^{-(\ell-1)} \binom{k - 1}{\ell - 1} (-1)^{\ell} & \text{if $ \ell \leq k$} \\
        0  & \text{otherwise}.
    \end{cases}
\end{equation}

Inserting these solutions back into \eqref{eq:laplace:source-star:taylor2}, we get for $i=0$
\begin{align*}
    [\expmlap \mathbf{b}]_{0} &= \sum_{k=0}^{\infty}\frac{(- \tau)^k }{k!} b_\nu(0, k) 
    = \sum_{k=0}^{\infty}\frac{(- \tau)^k }{k!}
    = \exp(- \tau),
\end{align*}
and for $i \in \{1, \dots, n - d \beta^{\ell -1}  - 1 \}$ we get
\begin{align*}
    [\expmlap \mathbf{b}]_{i} &= \sum_{k=0}^{\infty}\frac{(- \tau)^k }{k!} b_{\nu_l}(l, k) \\
    &= \frac{\branchscalar}{\beta^{l-1}}  (-1)^{l} \sum_{k=l}^{\infty}\frac{(- \tau)^k }{k!}  \binom{k}{l} \\
    &= \frac{\branchscalar}{\beta^{l-1}} (-1)^{l}\sum_{k=0}^{\infty}\frac{(- \tau)^{k+l} }{(k + l)!}  \binom{k + l}{l} \\
    &= \frac{\branchscalar}{\beta^{l-1}} \frac{\tau^l}{l!} \sum_{k=0}^{\infty}\frac{(- \tau)^k }{k!} 
    = \frac{\branchscalar}{\beta^{l-1}} \frac{ \tau^l}{l!} \exp(- \tau),
\end{align*}
and for $i \in \{(\ell - 1 )d  + 1, \dots, \ell d \} $
\begin{align*}
    [\expmlap \mathbf{b}]_{i} &= \sum_{k=0}^{\infty}\frac{(- \tau)^k }{k!} b_{\nu_\ell}(\ell, k) \\
    &= \frac{\branchscalar}{\beta^{\ell-1}}  (-1)^{\ell} \sum_{k=\ell}^{\infty}\frac{(- \tau)^k }{k!}  \binom{k - 1}{\ell -1} \\
    &= \frac{\branchscalar}{\beta^{\ell-1}} (-1)^{\ell} \sum_{k=0}^{\infty}\frac{(- \tau)^{k + \ell} }{(k + \ell)!}  \binom{k + \ell - 1}{\ell -1} \\
    &= \frac{\branchscalar}{\beta^{\ell-1}}  \frac{\tau^{\ell}}{(\ell -1)! } \sum_{k=0}^{\infty}(- \tau)^{k}  \frac{(k + \ell - 1)!}{(k + \ell)! k!} \\
    &=\frac{\branchscalar}{\beta^{\ell-1}}  \frac{ \tau^{\ell}}{(\ell -1)! } \sum_{k=0}^{\infty} \frac{(- \tau)^{k} }{(k + \ell) k!} \\
    &= \frac{\branchscalar}{\beta^{\ell-1}}  \frac{\gamma(\ell, \tau)}{(\ell -1)!} 
    = \frac{\branchscalar}{\beta^{\ell-1}} P(\ell, \tau).
\end{align*}
For the two equalities in the last step, we first used the identity 
\begin{equation*}
    \gamma(\ell, \tau) = \tau^\ell \sum_{k=0}^{\infty} \frac{(- \tau)^{k} }{(k + \ell) k!}
\end{equation*}
from \citep[Eq. 2.3]{boyadzhiev2007polyexponentials}, where $\gamma(\ell, \tau)$ is the lower incomplete gamma function, and then the definition for the normalized lower incomplete gamma function \cite{NIST_DLMF8.2.E4}
\begin{equation*}
    P(\ell, \tau) = \frac{\gamma(\ell, \tau)}{\Gamma(\ell)}.
\end{equation*}

\end{proof}

\begingroup
\def\thetheorem{\ref{col:50percent_heat}}
\begin{corollary}
 For any source-star graph $G_{\text{source-star}}(d, \beta, \ell)$, at time $\tau$, the total heat within in the $(0,R)$-ball $\Ncore(0, R)$, where $R < \ell$, is $\Gamma(R+1, \tau) / \Gamma(R+1)$, i.e.\ 
 \begin{equation}
  u_b^{(\text{ss})} (\tau, R, 0) = \sum_{i \in \Ncore^{(\text{ss})}(0, R)} u_i(\tau) = 
     \frac{\Gamma(R+1, \tau)}{\Gamma(R+1)}.
\end{equation}
\end{corollary}
\addtocounter{theorem}{-1}
\endgroup
\begin{proof}
First we observe that
\begin{align*}
    \sum_{i \in \Ncore(0, R)} u_i(\tau) 
    &=
    [\expmlap \mathbf{b}]_0 + \sum_{l=1}^R \sum_{\nu_l = 1}^{d \beta^{l-1}} [\expmlap \mathbf{b}]_{l, \nu_l} \\
    &= \exp(-\tau) + \sum_{l=1}^R \sum_{\nu_l = 1}^{d \beta^{l-1}} \frac{d^{-1}}{\beta^{l-1}} \frac{\tau^l}{l!} \exp(-\tau) \\
    &= \exp(-\tau) + \sum_{l=1}^R\frac{\tau^l}{l!} \exp(-\tau) \\
    &= \sum_{l=0}^R \frac{\tau^l}{l!} \exp(-\tau) 
\end{align*}
Using the Taylor expansion of $\exp(\tau)$ we have
\begin{align*}
    \sum_{l=0}^R \frac{\tau^l}{l!} &= \sum_{l=0}^\infty \frac{\tau^l}{l!}  - \sum_{l={R+1}}^\infty \frac{\tau^l}{l!}  
    = \exp(\tau) - \sum_{l={R+1}}^\infty \frac{\tau^l}{l!}. 
\end{align*}
Thus
\begin{align*}
    \sum_{l=0}^R \frac{\tau^l}{l!} \exp(-\tau)  &= 
    1 - \exp(-\tau) \sum_{l={R+1}}^\infty \frac{\tau^l}{l!}
    = 1 - \exp(-\tau) \sum_{l={0}}^\infty \frac{\tau^{l + R + 1}}{(l + R +1)!}
    = \frac{\Gamma(R+1, \tau)}{\Gamma(R+1)}
\end{align*}
where the last equality follows from \cite{NIST_DLMF8.7.E3}.

\end{proof}

\newpage  

\subsection{Recursion Relation} \label{sec:app:recursion}



\begin{lemma}
The recurrence relation 
\begin{align}
    F(n, m) &= F(n, m - 1) - bF(n - 1, m - 1) && 1 \leq n \leq m \label{eq:rec_relation}\\
    F(0, m) &= - a m  && 1 \leq m \label{eq:rec_rel_minit}\\
    F(n, 0) &= 0 && 1 \leq n  \label{eq:rec_rel_ninit}\\
    F(0, 0) &= 0,  \label{eq:rec_rel_0init}
\end{align}
with $a \in \Rreal_{+}$ and $b \in \Rreal_{+}$, 
is solved by 
\begin{equation} \label{eq:rec_relation_solution}
    F(n, m) = 
    \begin{cases}
        a b^n \binom{m}{n + 1} (-1)^{n +1} & \text{if $n < m$} \\
        0  & \text{otherwise}.
    \end{cases}
\end{equation}
\end{lemma}
\begin{proof}
Proof by induction.

\textbf{Base case} for $n=0$, \eqref{eq:rec_relation_solution} gives
\begin{align*}
    F(0, m) &= \begin{cases}
    a b^0 \binom{m}{1} (-1)^{1} & \text{if $m > 0$} \\
    0  & \text{otherwise}.
    \end{cases}
    = \begin{cases}
        - am & \text{if $m > 0$} \\
        0  & \text{otherwise},
    \end{cases}\\
\end{align*}
which agrees with \eqref{eq:rec_rel_minit} and \eqref{eq:rec_rel_0init}.
For $m=0$ \eqref{eq:rec_relation_solution} gives $F(n, 0) = 0$, which agrees with \eqref{eq:rec_rel_ninit} and \eqref{eq:rec_rel_0init}.

\textbf{Induction step}:
The induction assumptions are:
\begin{align*}
    F(n, m-1) &= 
    \begin{cases}
        a b^n \binom{m-1}{n + 1} (-1)^{n +1} & \text{if $n < m - 1$} \\
        0  & \text{otherwise}.
    \end{cases}\\
    F(n-1, m-1) &= 
    \begin{cases}
        a b^{n-1} \binom{m-1}{n } (-1)^{n} & \text{if $n - 1 < m - 1$} \\
        0  & \text{otherwise}.
    \end{cases}
\end{align*}

If $n \geq m$, both \eqref{eq:rec_relation} and \eqref{eq:rec_relation_solution} is zero, so this case it OK.
If $n = m - 1$, we have, starting from \eqref{eq:rec_relation},
\begin{align*}
    F(n, m) 
    &= F(n, m - 1) - bF(n - 1, m - 1)
    = 0 - b  a b^{n-1} \binom{m-1}{n } (-1)^{n}
    = a b^n (-1)^{n+1}\binom{m}{n+1},
\end{align*}
which equals \eqref{eq:rec_relation_solution}.
Note that we in the last equality used that $\binom{m-1}{n } = \binom{n}{n} = \binom{m}{n+1} = 1$ when $n = m - 1$.
The final case is $n < m - 1$, for which
\begin{align*}
    F(n, m) 
    &= F(n, m - 1) - bF(n - 1, m - 1)\\
    &=  b^n \binom{m-1}{n + 1} (-1)^{n +1} - b a b^{n-1} \binom{m-1}{n } (-1)^{n} \\
    &= a b^n (-1)^{n +1} \left( \binom{m-1}{n + 1} + \binom{m-1}{n}   \right) \\
    &= a b^n (-1)^{n +1} \binom{m}{n + 1},
\end{align*}
where the last step follows from Pascal's triangle, and which again equals \eqref{eq:rec_relation_solution}.

\end{proof}




\section{Additional Results} \label{app:additional_results}

\subsection{Scalability} \label{app:scalability}
Here, we show the results from the scalability experiment described in Section \ref{sec:scalability} for $5$ edges per node in the Barabási-Albert graph.

\begin{figure}[htp]
    \centering
    \includegraphics[width=0.5\linewidth]{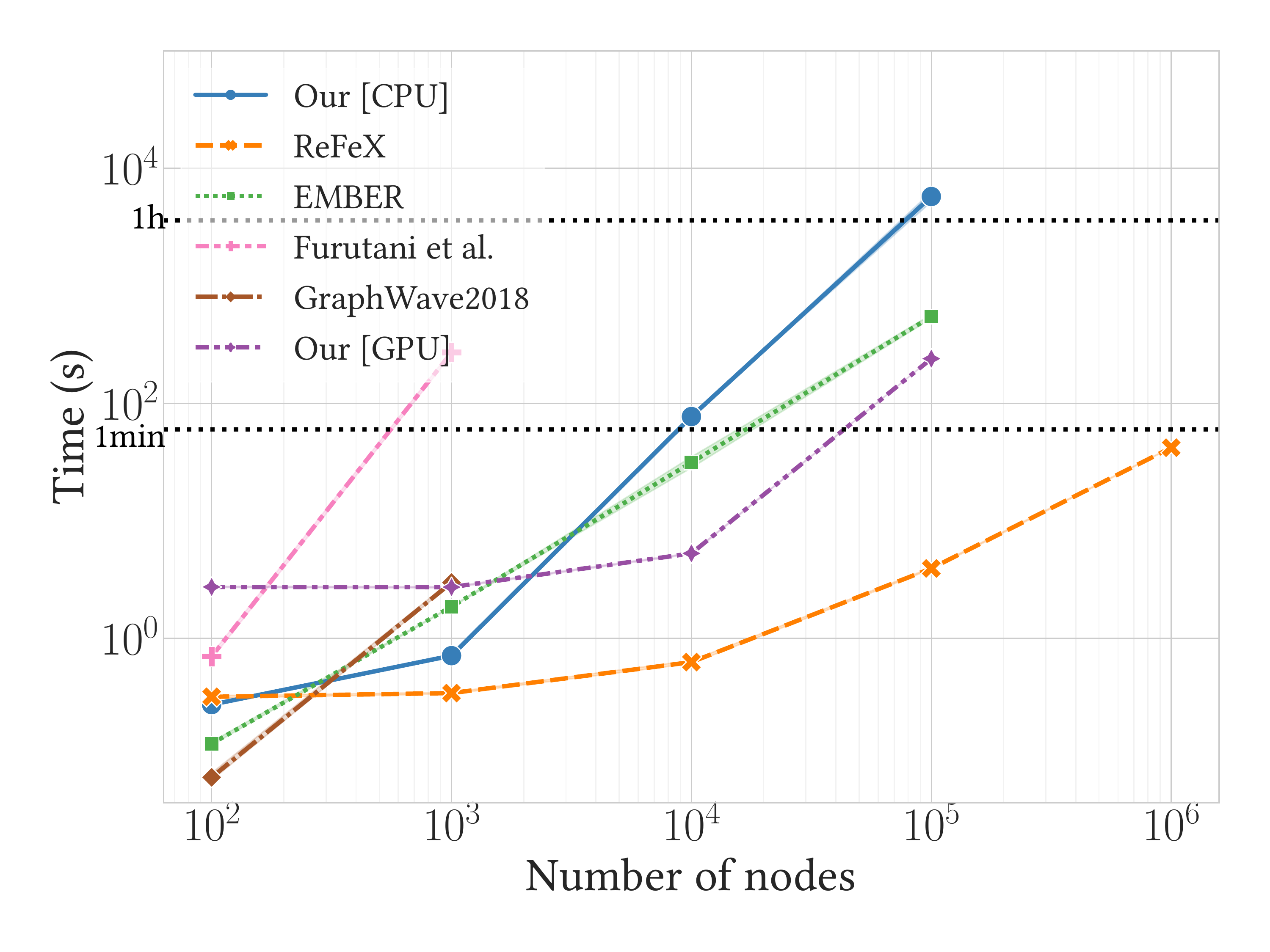}
        \Description{%
    A plot showing the scalability performance of six different embedding methods on a Barabási-Albert graph  with five edge per node. The y axes show the mean computation time in seconds, and the x axes the number of nodes in the graphs.
    }
    \caption{%
    The scalability experiment results using $5$ edges per node in a Barabási-Albert graph. The lines show the mean runtime in second vs number of nodes in the graph. Standard deviations are shown as shaded areas around each line, though these are hard to distinguish due to being insignificantly small.%
    }
    \label{fig:app:scale:results}
\end{figure}

\subsection{Network Alignment} \label{app:alignment}

Here, we show additional results from the network alignment experiment on the Enron dataset.
In Figure \ref{fig:alignment_enron_all_internal_undir}, when the Enron graphs are treated as undirected.
As visible, the accuracies are generally lower than when edge directions are used, as is to be expected.

In Figure \ref{fig:alignment_enron_all_internalk10}, the top-10 accuracies for the same experiment is shown. Qualitatively, the results are similar to the top-1 results.
Again, the increased performance of \textsc{ReFeX} and the in and out degree baseline increases much more in the weighted setting compared to \textsc{Digraphwave} since these methods includes the structure of both the weighted and unweighted graph in their embeddings.

\begin{figure}[htp]
    \begin{subfigure}{0.46\textwidth}
    \centering
    \includegraphics[width=\linewidth]{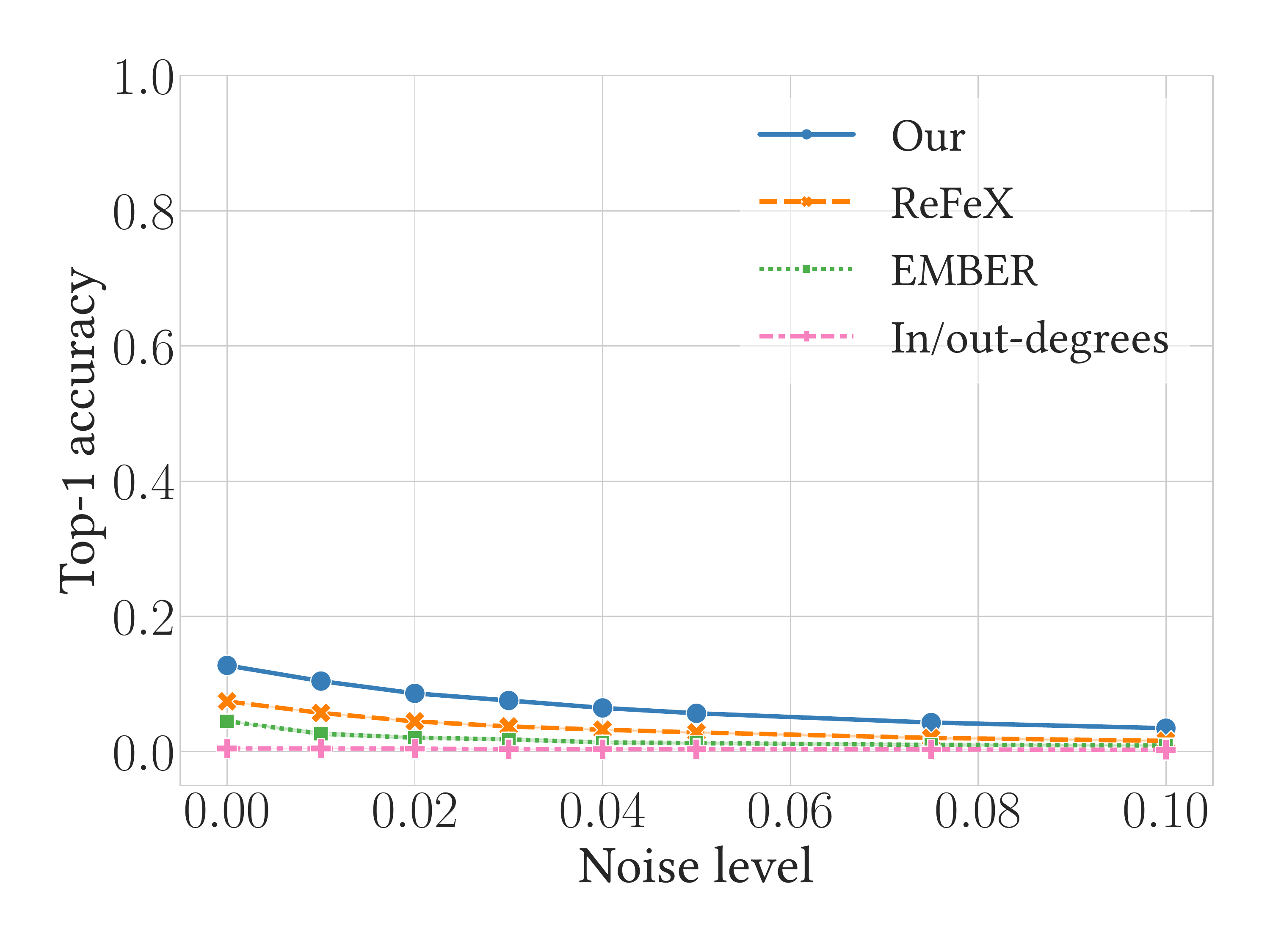}
    \caption{Undirected and unweighted}
    \end{subfigure}
    ~
    \begin{subfigure}{0.46\textwidth}
    \centering
    \includegraphics[width=\linewidth]{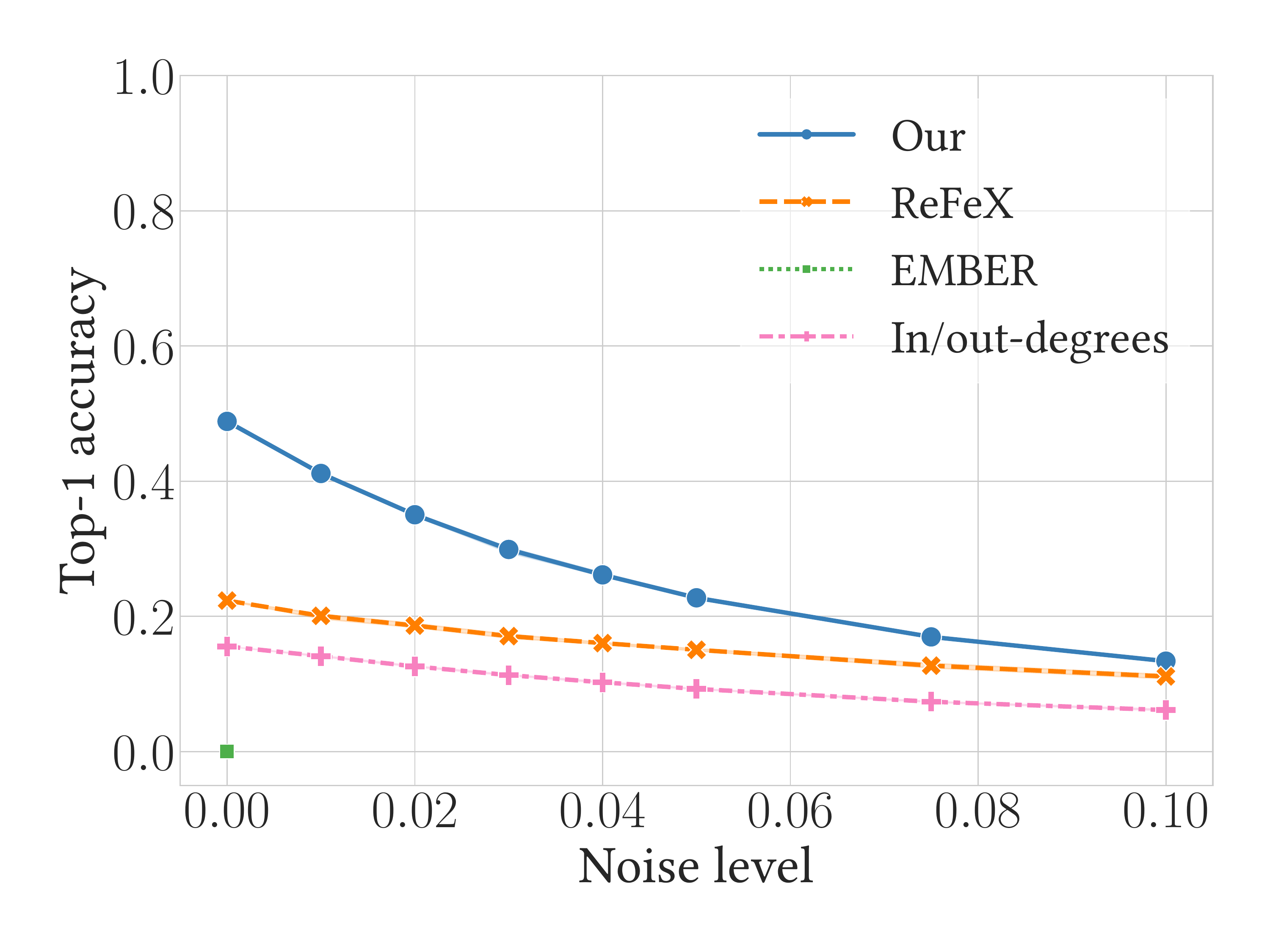}
    \caption{Undirected and weighted}
    \end{subfigure}
    \Description{Two plots placed side-by-side containing the top-1 accuracy results of the network alignment experiments on the Enron dataset. }
    \caption{Top-1 accuracy for the network alignment task on the Enron dataset when edge directions are ignored.}
    \label{fig:alignment_enron_all_internal_undir}
\end{figure}

\begin{figure}[htp]
    \begin{subfigure}{0.46\textwidth}
    \centering
    \includegraphics[width=\linewidth]{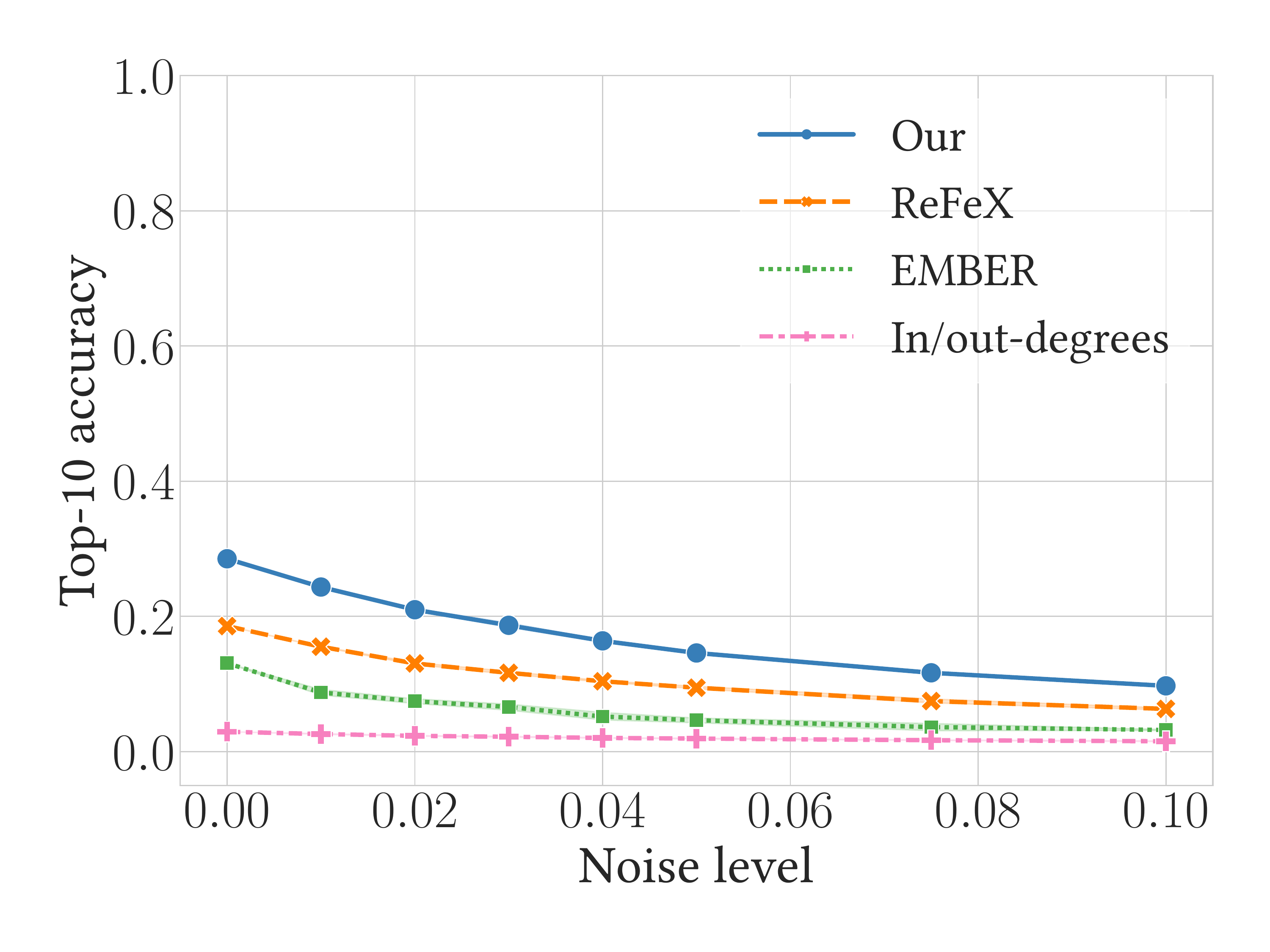}
    \caption{Undirected and unweighted}
    \end{subfigure}
    ~
    \begin{subfigure}{0.46\textwidth}
    \centering
    \includegraphics[width=\linewidth]{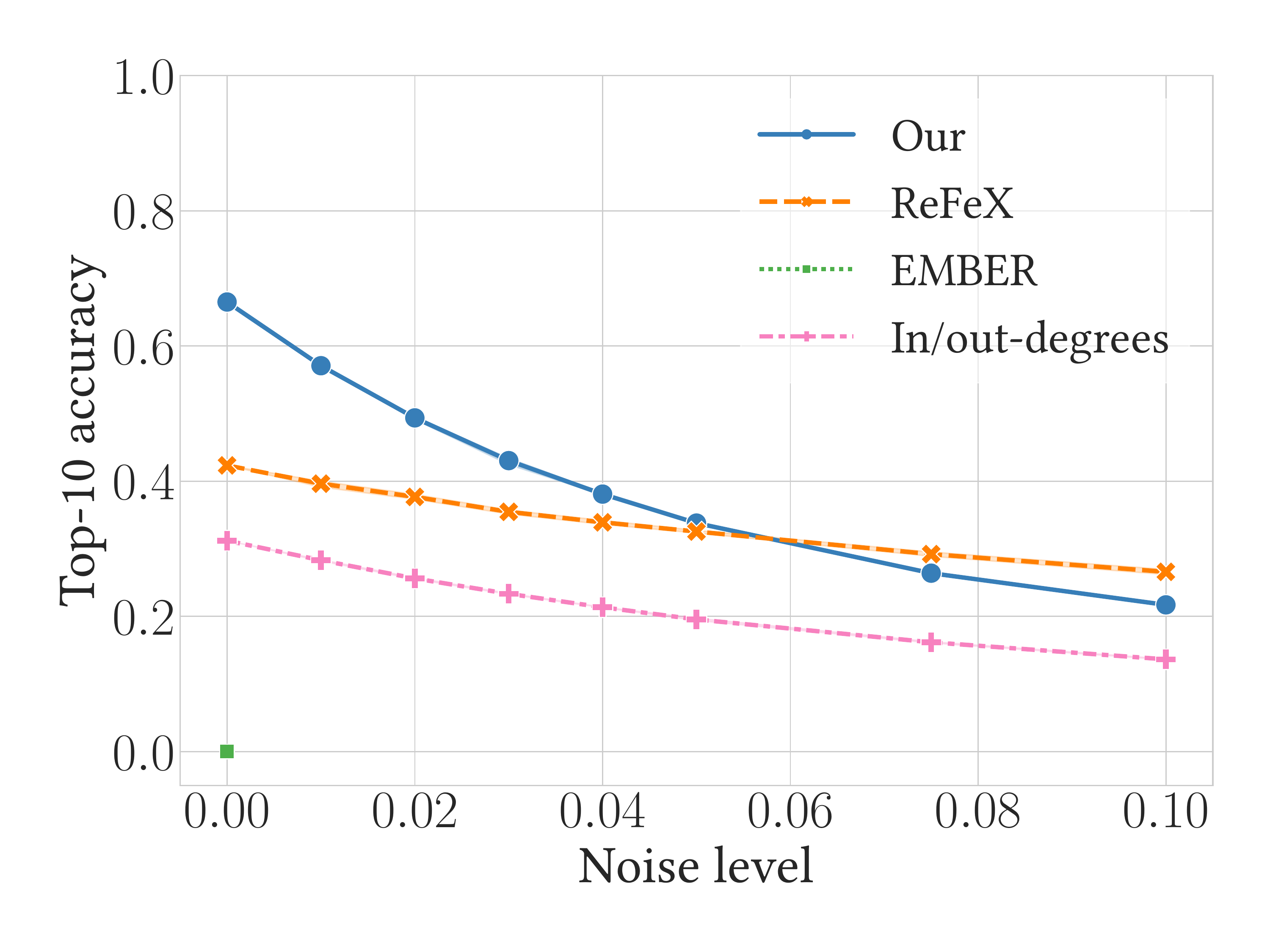}
    \caption{Undirected and weighted}
    \end{subfigure}
    \\
    \begin{subfigure}{0.46\textwidth}
    \centering
    \includegraphics[width=\linewidth]{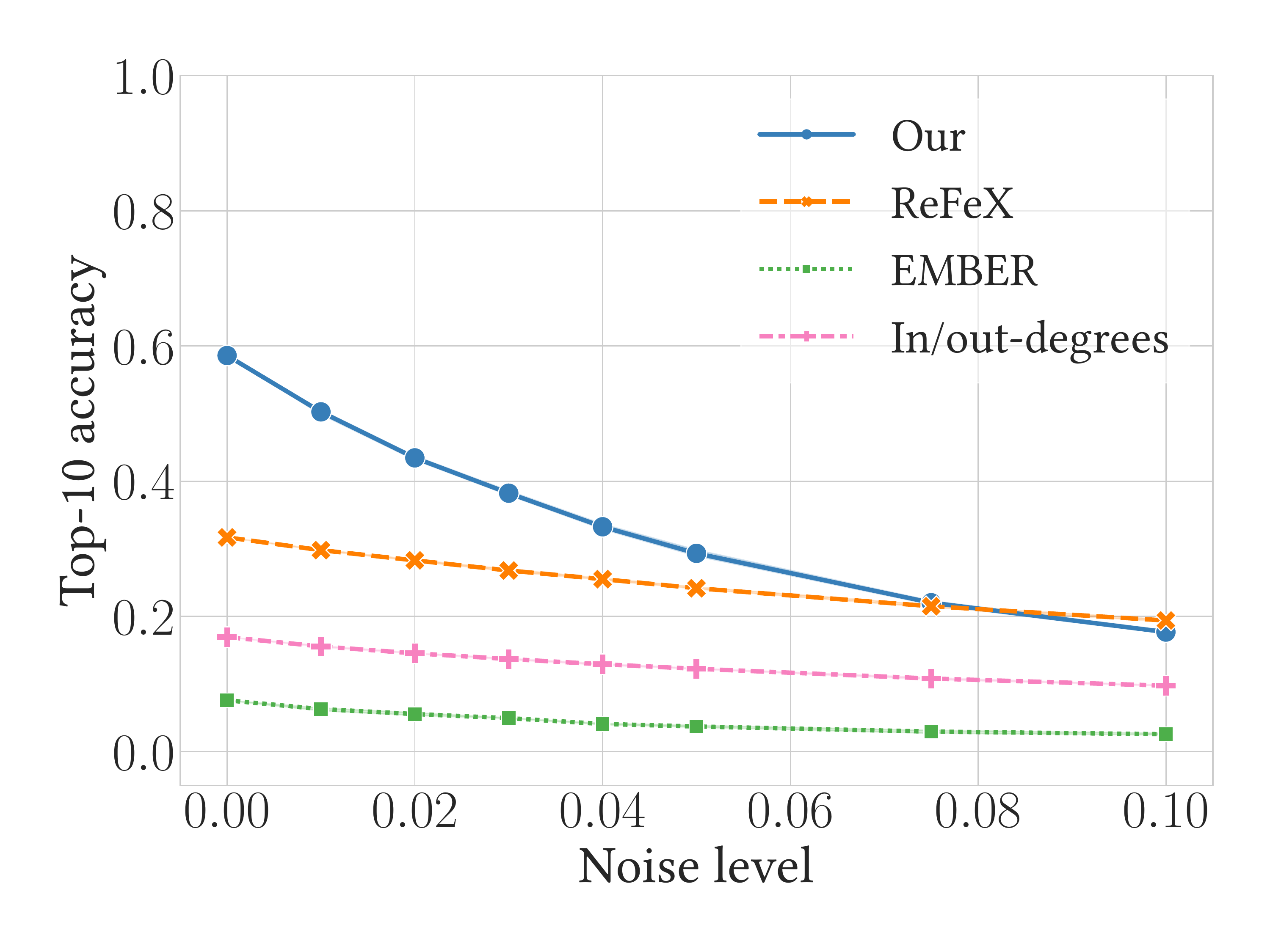}
    \caption{Directed and unweighted}
    \end{subfigure}
    ~
    \begin{subfigure}{0.46\textwidth}
    \centering
    \includegraphics[width=\linewidth]{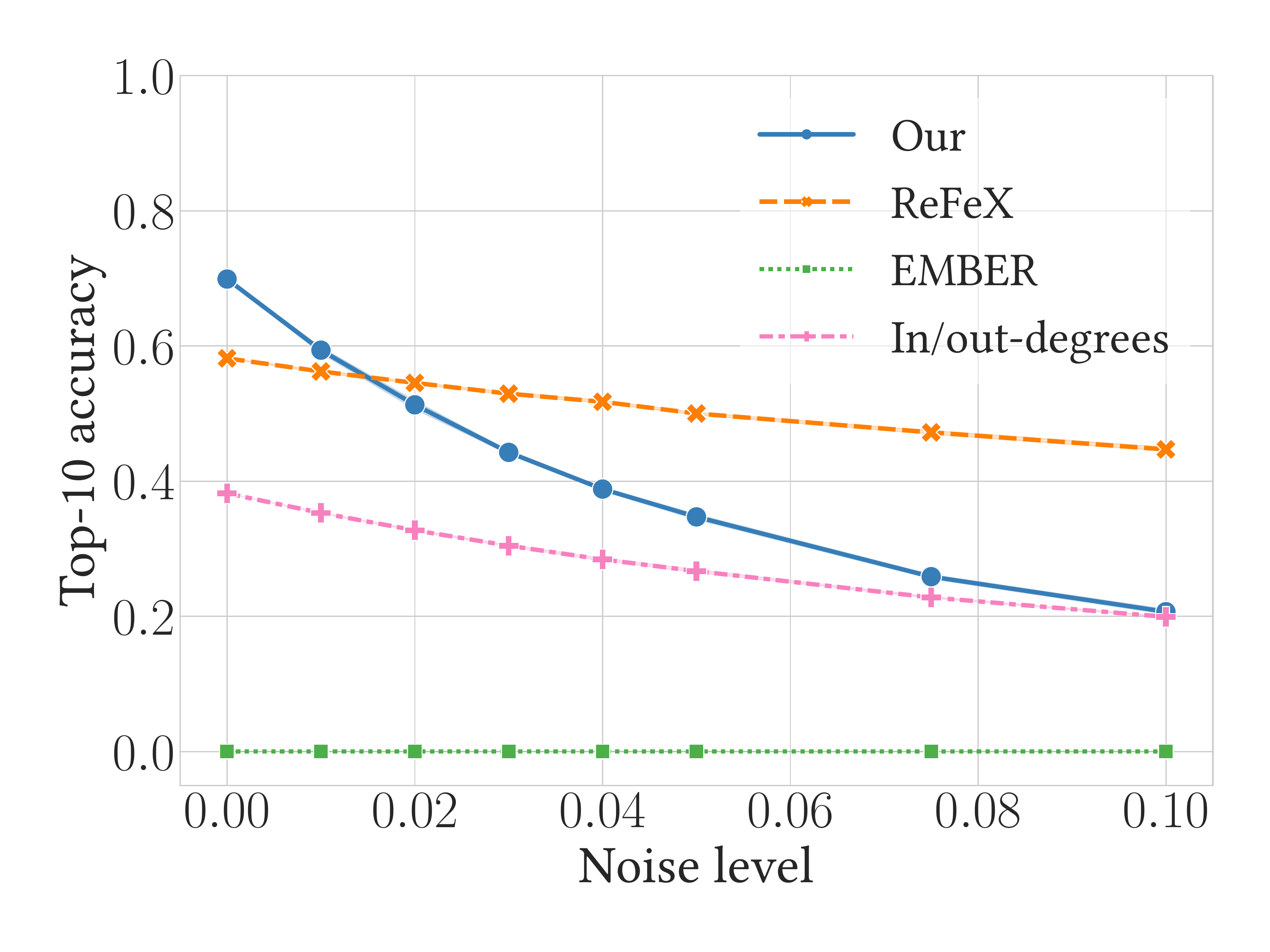}
    \caption{Directed and weighted}
    \end{subfigure}
    \Description{Four plots placed in a two by two grid, containing the top-10 accuracy results of the network alignment experiments on the Enron dataset. }
    \caption{Top-10 accuracies for the network alignment experiment on the Enron dataset.}
    \label{fig:alignment_enron_all_internalk10}
\end{figure}

\end{document}